\documentclass[A4,12pt]{article}

\usepackage{acro}
\usepackage{amsmath}
\usepackage{amssymb}
\usepackage{amsthm}
\usepackage{bbm}
\usepackage{centernot}
\usepackage{color}
\usepackage{fullpage}
\usepackage{graphicx}
\usepackage[authoryear]{natbib}
\usepackage{url}
\usepackage[usenames,dvipsnames]{xcolor}
\usepackage[pagebackref=true,colorlinks = true,urlcolor = Blue,linkcolor = Blue,citecolor = Blue,pdfborder={0 0 0}]{hyperref} 
\usepackage{cleveref}
\usepackage{physics}
\usepackage{mathtools}
\usepackage{enumitem}

\usepackage{minitoc}

\crefname{paragraph}{section}{sections}
\Crefname{paragraph}{Section}{Sections}

\theoremstyle{plain}
\newtheorem{lemma}{Lemma}
\newtheorem{coroll}{Corollary}
\newtheorem{theorem}{Theorem}
\newtheorem{proposition}{Proposition}

\theoremstyle{definition} 
\newtheorem{assumption}{Assumption}

\newtheorem{definition}{Definition}
\newtheorem{remark}{Remark}
\newtheorem{example}{Example}

\newenvironment{definitionprime}[1]
{%
\begin{definition}}
{\end{definition}%
}

\Crefname{assumption}{Assumption}{Assumptions}
\Crefname{coroll}{Corollary}{Corollaries}

\global\long\def\verts#1{\lvert#1\rvert}%
\global\long\def\Verts#1{\lVert#1\rVert}%

\global\long\def\la{\langle}%
\global\long\def\ra{\rangle}%

\newcommand{\GD}{\mathrm{GD}}
\newcommand{\KGD}{\mathrm{KGD}}
\newcommand{\SD}{\mathrm{SD}}

\newcommand{\KLD}{\mathrm{KLD}}
\global\long\def\rkhs#1{\mathcal{H}_{#1}}%
\global\long\def\rkhsball#1{\mathcal{B}_{#1}}%
\global\long\def\genscore#1#2{b_{#2}}%
\global\long\def\genpdf#1{\rho_{#1}}%
\global\long\def\seqidx{n}%
\global\long\def\polyorder{\alpha}%
\global\long\def\diss{\gamma}%
\global\long\def\finitemomentspace#1{{\cal P}_{#1}(\mathbb{R}^d)}%
\global\long\def\idmat{I_{d \times d}}%
\global\long\def\weight{a}%
\newcommand{\toL}[1]{\stackrel{{#1}}{\to}}
\newcommand{\eobj}{\mathcal{J}}
\newcommand{\DL}{\mathcal{D}_{L^1}^1}
\newcommand{\Op}[1]{\mathcal{T}_{#1}}
\newcommand{\LStOp}[1]{\mathcal{S}_{#1}}
\newcommand{\target}{P}
\newcommand{\vargrad}{\nabla_{\mathrm{V}}}

\newcommand{\R}{\mathbb{R}}

\newcommand{\cA}{\mathcal{A}}

\newcommand{\cF}{\mathcal{F}}

\newcommand{\cK}{\mathcal{K}}
\newcommand{\cL}{\mathcal{L}}
\newcommand{\cP}{\mathcal{P}}
\newcommand{\cR}{\mathcal{R}}
\newcommand{\cS}{\mathcal{S}}

\newcommand{\Id}{\mathrm{I}_d}
\newcommand{\dom}{\mathrm{dom}}
\DeclareAcronym{gd}{short = GD, long = Gradient Discrepancy}
\DeclareAcronym{kgd}{short = KGD, long = Kernel Gradient Discrepancy}
\DeclareAcronym{ksd}{short = KSD, long = Kernel Stein Discrepancy}
\DeclareAcronym{kld}{short = KLD, long = Kullback--Leibler divergence}
\DeclareAcronym{mcmc}{short = MCMC, long = Markov chain Monte Carlo}
\DeclareAcronym{mmd}{short = MMD, long = maximum mean discrepancy}
\DeclareAcronym{rkhs}{short = RKHS, long = reproducing kernel Hilbert space}
\DeclareAcronym{mfnn}{short = MFNN, long = mean field neural network}
\DeclareAcronym{mfld}{short = MFLD, long = mean field Langevin dynamics}
\DeclareAcronym{svgd}{short = SVGD, long = Stein variational gradient descent}
\DeclareAcronym{vgd}{short = VGD, long = variational gradient descent}
\DeclareAcronym{ula}{short = ULA, long = unadjusted Langevin algorithm}
\DeclareAcronym{ode}{short = ODE, long = ordinary differential equation}
\DeclareAcronym{dct}{short = DCT, long = dominated convergence theorem}
\DeclareAcronym{tvs}{short = TVS, long = topological vector space}

\DeclareMathOperator*{\argmin}{arg\,min}

\newcommand{\nll}{\centernot{\ll}}

\usepackage{changepage}   %

\begin{document}

\doparttoc
\faketableofcontents 

\title{A Computable Measure of Suboptimality for Entropy-Regularised Variational Objectives}

\author{Cl\'{e}mentine Chazal$^1$, Heishiro Kanagawa$^2$\footnote{Now at Fujitsu Research of Europe, UK.},  
Zheyang Shen$^2$, \\
Anna Korba$^1$, Chris. J. Oates$^{2,3,}$\footnote{Correspondence should be addressed to Chris. J. Oates, School of Mathematics, Statistics \& Physics, Newcastle University, Newcastle-upon-Tyne, NE1 7RU, UK. Email: \url{chris.oates@ncl.ac.uk}.} \\
\small $^1$CREST, ENSAE, IP Paris, France \\
\small $^2$Newcastle University, UK \\
\small $^3$The Alan Turing Institute, UK 
}

\maketitle

\begin{abstract}
Several emerging post-Bayesian methods target a probability distribution for which an entropy-regularised variational objective is minimised.
This increased flexibility introduces a computational challenge, as one loses access to an explicit unnormalised density for the target.
To mitigate this difficulty, we introduce a novel measure of suboptimality called \emph{gradient discrepancy}, and in particular a \emph{kernel} gradient discrepancy (KGD) that can be explicitly computed.
In the standard Bayesian context, KGD coincides with the kernel Stein discrepancy (KSD), and we obtain a novel characterisation of KSD as measuring the size of a variational gradient.
Outside this familiar setting, KGD enables novel sampling algorithms to be developed and compared, even when unnormalised densities cannot be obtained. 
To illustrate this point several novel algorithms are proposed and studied, including a natural generalisation of Stein variational gradient descent, with applications to mean-field neural networks and predictively oriented posteriors presented.
On the theoretical side, our principal contribution is to establish sufficient conditions for desirable properties of KGD, such as continuity and convergence control. 
\end{abstract}

\noindent \textit{Keywords:} \; post-Bayesian, reproducing kernel, sampling, variational gradient

\section{Introduction}\label{sec:intro}

Consider a distribution $\target$, defined as the unique minimiser of a (relative) entropy-regularised variational objective
\begin{align}
\eobj(Q) \coloneq \cL(Q) + \mathrm{KLD}(Q || Q_0) \label{eq: objective}
\end{align}
for a given \emph{loss function} $\cL$, i.e. regularised by the Kullback--Leibler divergence between $Q$ and a reference distribution $Q_0$.
As a shorthand, we refer to $\target$ as the \emph{variational target}.
Such objectives appear under various different guises in statistics and machine learning; in particular the Bayesian posterior solves \eqref{eq: objective} with a linear loss $\cL$ that is the average negative log-likelihood loss 
\begin{align}
\cL(Q) = - \int \log p(y | x) \; \dd Q(x),  \label{eq: std bayes}
\end{align}
where $p(y | x)$ is a model for data $y$ with parameters $x$, and the prior taking the role of $Q_0$.
In this case, numerical computation of $\target$ can proceed using established methods such as \ac{mcmc}, since the linearity of the loss $\cL$ implies a (generalised) Bayes' rule \citep[][Theorem 1]{knoblauch2022optimization} and thus one has explicit access to an unnormalised density for the target.

Several emerging \emph{post-Bayesian} methods also take the form \eqref{eq: objective}, but for these methods the loss $\cL$ is \emph{nonlinear}, precluding access to an unnormalised density for the target.
A prominent example is the \emph{predictively oriented} (PrO) \emph{posterior} pursued in \citet{masegosa2020learning,sheth2020pseudo,jankowiak2020deep,jankowiak2020parametric,morningstar2022pacm,shen2024prediction,mclatchie2025predictively,lai2024predictive}, which lifts a parametric statistical model $p(\cdot | x)$ to a mixture model and learns the mixing distribution via minimisation of a loss function of the form
$$
\cL(Q) = - \frac{1}{\lambda_N} \sum_{i=1}^N S \left( \int p(\cdot | x) \; \dd Q(x) , \delta_{y_i} \right) , 
$$
where $\lambda_N > 0$ is a learning rate, $\{y_i\}_{i=1}^N$ is a dataset, and $S$ is any (proper) scoring rule \citep{gneiting2007strictly}.
The lifting of the original statistical model to a mixture model ensures that parameter uncertainty does not collapse in the large-data limit when the model is non-trivially misspecified \citep{mclatchie2025predictively}.
However, computation with \emph{nonlinear} $\mathcal{L}$ can be challenging, as numerical methods are relatively under-developed (e.g. compared to \ac{mcmc}), with (biased) particle-based algorithms such as \ac{mfld} \citep{del2013mean} 
being among the limited state-of-the-art.

The motivation of this paper is to address the following basic question pertaining to \eqref{eq: objective}:  
Given a collection of samples $\{x_i\}_{i=1}^n$, how well do these samples approximate the (variational) target?
To draw an analogy; for standard sampling methods such as \ac{mcmc}, convergence diagnostics are well-established \citep[][Chapter 7]{brooks2011handbook}, while Stein discrepancies emerged as computable measures of sample quality when one has access to an unnormalised form of the target \citep{gorham2015measuring} - the most popular example being the \ac{ksd} due to its closed form \citep{chwialkowski2016kernel,liu2016kernelized} and desirable theoretical properties \citep{gorham2017measuring,barp2024targeted,KanBarGreMac2025}.
However, it is unclear how to proceed for general entropy-regularised variational objectives of the form \eqref{eq: objective}; one cannot compute $\eobj(Q_n)$ for the empirical measure $Q_n = \frac{1}{n} \sum_{i=1}^n \delta_{x_i}$ since $Q_n$ will typically not be absolutely continuous with respect to $Q_0$, meaning the entropy term in $\eobj(Q_n)$ is not well-defined.
Focusing on particle-based methods such as \ac{mfld}, one can monitor the gradient at each particle location to try and understand if the dynamics has converged, but this does not account for the error introduced when discretising the probability flow using particles at the outset.
The absence of a computable measure of sample approximation quality for targets defined as minimisers of \eqref{eq: objective} is arguably a major barrier to both methodological and algorithmic development, as well as to the practical deployment of post-Bayesian methodologies; we seek to remove this barrier in this work.

The outline and contributions of this paper are as follows:
\begin{itemize}
    \item \Cref{sec: methods} introduces \emph{\acl{gd}} (GD), a statistical divergence that captures closeness of a distribution $Q$ to the target $\target$ by measuring the size of a variational gradient $\vargrad \eobj(Q)$.
    A specific instance, called \emph{\acl{kgd}} (KGD), can be computed in closed-form for all distributions $Q_n$ with finite support.
    In the familiar case where $\mathcal{L}$ is a linear functional, one has access to an unnormalised density for the target and \ac{gd} coincides with Langevin Stein discrepancy \citep{gorham2015measuring}.
    This leads to a novel interpretation of this Stein discrepancy as measuring the size of a variational gradient.
    \item \ac{gd} enables efficient sampling algorithms for entropy-regularised variational objectives to be developed and compared.
    As illustrations, in \Cref{sec: applications} we demonstrate how \ac{kgd} (i) provides a principled approach to tuning the step size in \ac{mfld}, (ii) enables the first \emph{extensible} sampling algorithm (i.e. adding one particle at a time until a desired accuracy is achieved) in the context of \eqref{eq: objective}, (iii) enables direct optimisation of a particle set through gradient descent, (iv) enables parametric variational inference without requiring explicit densities for the parametric model, and (v) facilitates the first convergence analysis for nonlinear Stein variational gradient descent \citep{wang2019nonlinear}.
    \item The theoretical properties of \ac{kgd} are set out in \Cref{sec: theory}.  
    In particular, we establish sufficient conditions for continuity and convergence control in a range of topologies, spanning weak convergence to Wasserstein convergence, meaning in particular that the failure of an algorithm to properly approximate higher order moments can be detected.
    Strikingly, our results imply that in regular settings convergence is \emph{characterised} by \ac{kgd}, 
    meaning that \emph{any} consistent algorithm for approximating $\target$ must necessarily be performing (asymptotic) minimisation of the \ac{kgd}.
    Our analysis is rich enough to include loss functions $\mathcal{L}$ of \emph{interaction energy form}, for example encountered with PrO posteriors based on kernel scoring rules, but is not completely general (e.g. our convergence control result excludes PrO posteriors based on the logarithmic scoring rule).
\end{itemize}

\noindent A discussion concludes  \Cref{sec: discuss}.
Code to reproduce all experiments in the manuscript can be downloaded from \url{https://github.com/clementinechazal/kgd-code}.

\section{Methods}
\label{sec: methods}

The notation we will use is introduced in \Cref{subsec: notation main}, then precise definitions of \ac{gd} and \ac{kgd} are presented in \Cref{subsec: gd}.

\subsection{Set-Up and Notation}
\label{subsec: notation main}

To simplify presentation we restrict attention to $\finitemomentspace{}$, the set of Borel probability distributions on $\mathbb{R}^d$.

\paragraph*{Operations on $\mathbb{R}^d$}
Let $\|\cdot\|$ denote any (equivalent) norm on $\mathbb{R}^d$ (or $\mathbb{R}^{d \times d}$) and $\|\cdot\|_{\mathrm{op}}$ the associated operator norm (in the case of $\mathbb{R}^d$).
On occasion we will specifically distinguish $\|x\|_p \coloneqq (\sum_i |x_i|^p)^{1/p}$ using the subscript $p \in [1,\infty)$.
For $f,g : \mathbb{R}^d \rightarrow \mathbb{R}$, write $f(x) \lesssim g(x)$ if there exists a finite constant $C$ such that $f(x) \leq C g(x)$ for all $x \in \mathbb{R}^d$.
Let $C^r(\R^d, \R^m)$ denote the set of $r$ times continuously differentiable functions from $\mathbb{R}^d$ to $\mathbb{R}^m$ and $C^r(\R^d)$ the shorthand for  $ C^r(\R^d, \R)$. 
For a differentiable function $f : \mathbb{R}^d \rightarrow \mathbb{R}$, let $\partial_i f$ denote the partial derivative of $f$ with respect to the $i$th argument, let $\nabla f$ denote the gradient vector whose elements are $\partial_i f$, and let $\nabla^2 f$ denote the matrix of mixed partial derivatives $\partial_i\partial_j f$.
For a differentiable vector field $f : \mathbb{R}^d \rightarrow \mathbb{R}^d$, let $\nabla \cdot f$ denote the divergence of the vector field.
For a multivariate function, let $\partial_{i,j}$, $\nabla_i$, etc, indicate the action of the differential operators with respect to the $i$th argument.
Let $f_{-}$ denote the negative part $f_{-}: x \mapsto \min\{0 , f(x)\}$ of a function $f$. 

\paragraph*{Measures}
A sequence $(Q_n)_{n \in \mathbb{N}} \subset \finitemomentspace{}$ is said to \emph{weakly converge} to $Q \in \finitemomentspace{}$ if $\int f \mathrm{d}Q_n \rightarrow \int f \mathrm{d}Q$ for all continuous and bounded $f : \R^d \rightarrow \R$.
Let $\finitemomentspace{k}$ denote the subset of distributions $Q$ in $\finitemomentspace{}$ for which $\int \|x\|^k \; \mathrm{d}Q < \infty$.
Let $Q \ll Q_0$ denote that $Q$ is absolutely continuous with respect to $Q_0$, and $\dd Q / \dd Q_0$ its Radon--Nikodym density. 
Let $\mathcal{L}^1(Q) := \{f : \mathbb{R}^d \rightarrow \mathbb{R} : \int \|f(x)\| \; \mathrm{d}Q(x) < \infty  \}$ denote\footnote{This should not be confused with the notation $\cL$ for the loss function in \eqref{eq: objective}.}  
the set of $Q$-integrable functions on $\mathbb{R}^d$.
Let $Q^{\otimes r}$ denote the $r$-fold product measure.
Let $T_\# Q$ denote the distribution of $T(X)$ where $X \sim Q$.

\paragraph*{Operations on $\finitemomentspace{}$}
A functional $\mathcal{F} : \finitemomentspace{} \rightarrow \mathbb{R}$ is said to be \emph{weakly continuous} if $\mathcal{F}(Q_n) \rightarrow \mathcal{F}(Q)$ whenever $(Q_n)_{n \in \mathbb{N}} \subset \finitemomentspace{}$ converges weakly to $Q \in \finitemomentspace{}$.
For a suitably regular functional $\mathcal{F} : \finitemomentspace{} \rightarrow (-\infty, \infty]$, the \emph{first variation} at $Q \in \dom(\cF)=\{Q\in \finitemomentspace{}: \cF(Q)<\infty\}$ is defined as a map $\cF'(Q) : \R^d \rightarrow \R$ such that $\lim_{\epsilon\to 0}\frac{1}{\epsilon} \{ \cF(Q+\epsilon \chi)-\cF(Q) \} = \int\cF'(Q) \; \mathrm{d}\chi$ for all perturbations $\chi$ of the form $\chi=R - Q$ with $R \in \cP(\R^d)$; note that if it exists, the first variation is unique up to an additive constant. 
For $Q \ll Q_0$, $\KLD(Q \Vert Q_0) \coloneqq \int \log (\dd Q / \dd Q_0)\; \dd Q$, while for $Q \nll Q_0$ we set $\KLD(Q || Q_0) = \infty$.

\subsection{Gradient Discrepancy}
\label{subsec: gd}

To begin with some intuition, consider the task of finding the minimiser of a function $J : \mathbb{R}^d \rightarrow \mathbb{R}$. 
If $J$ is differentiable with a unique stationary point, then the closeness of a candidate state $x$ to optimality can be quantified by measuring the size of the gradient, i.e. $\|\nabla J(x)\|$.
The idea of Gradient Discrepancy (\ac{gd}) is to lift this argument from $\mathbb{R}^d$ to $\finitemomentspace{}$. 
For this, we will need an appropriate notion of a gradient:

\begin{definition}[Variational gradient]
\label{def: var grad}
    The \emph{variational gradient} of $\mathcal{F}$ at $Q$, denoted $\vargrad  \mathcal{F}(Q):\R^d\to\R^d$, when it exists, is defined as the gradient of the first variation; i.e. $\vargrad  \mathcal{F}(Q)(x) \coloneqq \nabla_x \mathcal{F}'(Q)(x)$ for each $x \in \R^d$.
\end{definition}

\noindent The variational gradient as we have defined it applies to functionals on $\finitemomentspace{}$, but it is closely related to several existing notions when restricted to certain subspaces of $\finitemomentspace{}$.
For instance, the variational gradient coincides with the \emph{intrinsic derivative} \citep{albeverio1996differential} when restricting to $\finitemomentspace{k}$ for any $k \geq 1$, and with the \emph{Lions derivative} \citep{cardaliaguet2010notes} from the theory of mean field games when restricting to $\finitemomentspace{2}$.
Similarly, the variational gradient is closely related to the \emph{Wasserstein gradient} from optimal transport on the Wasserstein space, i.e., $\finitemomentspace{2}$ equipped with the Wasserstein-2 metric 
\citep[see e.g.][for background]{Otto01t,lanzetti2025first}.

Armed with an appropriate notion of gradient, we can now define what it means to be a stationary point of $\eobj$ in \eqref{eq: objective}.
Eschewing technical considerations for the moment, the variational gradient of $\eobj$ is formally
\begin{align}
    \vargrad \eobj (Q)(x) & = \vargrad \mathcal{L}(Q)(x) + \nabla \log \frac{\mathrm{d}Q}{\mathrm{d}Q_0}(x) , \label{eq: formal var grad}
\end{align}
which vanishes if and only if the \emph{self-consistency} equation
\begin{align}
\dv[]{Q}{Q_0} \propto \exp(-\cL'(Q) )  \label{eq: self-consistency}
\end{align}
is satisfied.
Some care is required to rigorously make this argument; this is because the variational gradient of the \ac{kld} term in \eqref{eq: objective} exists only in a restricted sense, due to the \ac{kld} itself not being well-defined on all of $\mathcal{P}(\mathbb{R}^d)$. 
For our purposes it is sufficient to \emph{define} stationary points as solutions of \eqref{eq: self-consistency}. 
Additional background relating minimisers of $\eobj$, the roots of $\vargrad \eobj$, and the self-consistency equation is provided in \Cref{app: defn first var}. 
Sufficient conditions for existence and uniqueness of solutions to \eqref{eq: self-consistency} are provided in \Cref{sec:minimizer}.

\begin{definition}[Stationary point]\label{def:stationary_point}
    A distribution $Q \in \finitemomentspace{}$ is called a \emph{stationary point} of $\eobj$ in \eqref{eq: objective} if $Q$ satisfies the self-consistency equation \eqref{eq: self-consistency}.
\end{definition}

\noindent Thus, since our target distribution $\target$ is a stationary point of $\eobj$ in \eqref{eq: objective}, a quantitative measure of the size of the variational gradient $\vargrad  \eobj(Q)$ appears to be a sensible way to measure dissimilarity of $Q$ to $\target$.
However, the objective \eqref{eq: objective} is only well-defined for distributions $Q$ with $Q \ll Q_0$. 
This important technical point will now be resolved.

For the moment, assume that $Q$ and $Q_0$ have positive and differentiable densities $q$ and $q_0$ on $\mathbb{R}^d$.
To relax this requirement we will compute quantities that do not involve the density of $Q$. 
Let $v : \mathbb{R}^d \rightarrow \mathbb{R}^d$ be a vector field; we are going to probe the size of the variational gradient as projected along this vector field.
To this end, we introduce the class of linear differential operators
\begin{align}
    \Op{Q} v(x) \coloneqq \left[ (\nabla\log q_0)(x) - \vargrad  \cL(Q)(x)  \right] \cdot v(x) + (\nabla \cdot v)(x) , \qquad Q \in \mathcal{P}(\mathbb{R}^d) .  \label{eq: gen st op}
\end{align}
Using \eqref{eq: formal var grad}, and assuming sufficient regularity for an integration-by-parts formula to hold,
\begin{align}
	& \hspace{-10pt} \int \vargrad  \eobj(Q)(x) \cdot v(x) \; \dd Q(x) \label{eq: int of W grad} \\
	&  \hspace{-10pt}  = \int \left[ \vargrad  \cL(Q)(x) - (\nabla \log q_0)(x) \right] \cdot v(x) \; \dd Q(x) + \int (\nabla \log q)(x) \cdot v(x) \; \dd Q(x) \nonumber \\
	&  \hspace{-10pt}  = \int \left[ \vargrad  \cL(Q)(x) - (\nabla \log q_0)(x) \right] \cdot v(x) \; \dd Q(x) - \int (\nabla \cdot v)(x) \; \dd Q(x) \nonumber \\
	&  \hspace{-10pt}  = - \int \Op{Q} v(x) \; \dd Q(x) . \nonumber
\end{align}
In particular, the density $q$ does not feature in \eqref{eq: gen st op}, suggesting the following definition which can be generally applied and interpreted as measuring the size of the variational gradient of $\eobj$:

\begin{definition}[Gradient Discrepancy for entropy-regularised objectives]
\label{def: grad disc}
    For a given set $\mathcal{V}$ of differentiable vector fields on $\mathbb{R}^d$, we define the \acl{gd} as
    \begin{align}
    	\GD(Q) \coloneqq \sup_{\substack{v \in \mathcal{V} \;  \mathrm{s.t.} \\ (\Op{Q} v)_{-} \in \cL^1(Q)}} \left| \int \Op{Q} v(x) \; \dd Q(x) \right|  \label{eq: first def}
    \end{align}
    as a map $\GD : \finitemomentspace{} \to [0,\infty]$, where $\Op{Q}$ is defined in \eqref{eq: gen st op}.\footnote{The restriction $(\Op{Q} v)_{-} \in \cL^1(Q)$ 
    is imposed so that all integrals are well-defined, with infinite integrals being allowed.
    It should also be interpreted as requiring $\vargrad  \mathcal{L}(Q)$ (and hence $\Op{Q}$) to be well-defined.}
\end{definition}
\noindent In particular, $\GD(Q)$ is well-defined whenever the variational gradient $\vargrad \cL(Q)$ is well-defined, ameliorating the requirement for $Q$ to be absolutely continuous with respect to $Q_0$.
In particular we will see several examples where $\GD(Q_n)$ is well-defined for $Q_n$ with finite support.

Though the derivation is short, this construction represents a substantial generalisation of classical arguments due to \citet{stein1972bound}, \citet{hyvarinen2005estimation}, as well as the more recent work of \citet{gorham2015measuring}, all of whom required access to an unnormalised density of the target:

\begin{example}[Langevin Stein Discrepancy]
\label{ex: Langevin Stein}
In the standard Bayesian setting, the loss function $\cL$ in \eqref{eq: std bayes} satisfies $\vargrad  \cL(Q)(x) = - \nabla_x \log p(y|x)$.
Hence $\Op{Q}$ in \eqref{eq: gen st op} reduces to the ($Q$-independent) \emph{Langevin Stein operator} $\LStOp{P} v(x) \coloneqq (\nabla \log p)(x) \cdot v(x) + (\nabla \cdot v)(x)$, where $p(x)\propto q_0(x) p(y|x)$ is a density for $P$, and \eqref{eq: first def} is the \emph{Langevin Stein discrepancy}
\begin{align*}
	\SD(Q) \coloneqq \sup_{\substack{ v \in \mathcal{V} \; \mathrm{s.t.} \\ (\LStOp{P} v)_{-} \in \cL^1(Q) }} \left| \int \LStOp{P} v(x) \; \dd Q(x) \right| 
\end{align*}
of \citet{gorham2015measuring}.
Historically, the name derives from a connection to the generator approach of \citet{barbour1990stein}, in which one measures the change to $Q$ after an infinitesimal perturbation using a $P$-invariant Langevin diffusion. 
Our construction reveals a novel interpretation of $\SD(Q)$ as measuring the size of the variational gradient of $\eobj$.
\end{example}

Unlike \Cref{ex: Langevin Stein}, the proposed \ac{gd} is applicable in situations where one does not have access to an unnormalised density for the target.
To emphasise practicality, we focus on an instance of \ac{gd} for which the supremum in \eqref{eq: first def} can be explicitly computed - leveraging reproducing kernel Hilbert spaces.
To this end, recall that a bivariate function $K : \mathbb{R}^d \times \mathbb{R}^d \rightarrow \mathbb{R}^{d \times d}$ is a (matrix-valued) \emph{kernel} if it is (i) \emph{transpose-symmetric}, meaning $K(x,y) = K(y,x)^\top$ for all $x,y \in \mathbb{R}^d$, and (ii) \emph{positive semi-definite}, meaning $\sum_i \sum_j  u_i \cdot K(x_i,x_j) u_j  \geq 0$ for all $\{x_i\}_{i=1}^n \subset \mathbb{R}^d$, all $\{u_i\}_{i=1}^n \subset \mathbb{R}^d$, and all $n \in \mathbb{N}$. 
To each such kernel $K$, there exists a unique (vector-valued) \ac{rkhs}.
This Hilbert space will be denoted $\mathcal{H}_K$, and is characterised by (i) $K(\cdot,x) u \in \mathcal{H}_K$ for all $x \in \mathbb{R}^d$, $u \in \mathbb{R}^d$, and (ii) $\langle v , K(\cdot,x) u \rangle_{\mathcal{H}_K} = v(x) \cdot u$ for all $v \in \mathcal{H}_K$, $x \in \mathbb{R}^d$, and $u \in \mathbb{R}^d$; see \citet{carmeli2010vector} for background.

\begin{definition}[Kernel Gradient Discrepancy]\label{def:kgd}
Let $K : \mathbb{R}^d \times \mathbb{R}^d \rightarrow \R^{d \times d}$ be a matrix-valued kernel such that $\partial_{1, i}\partial_{2, i}K$ exists for all $i \in \{1, \dots, d\}$. 
Let $\rkhsball{K} = \{v \in \rkhs{K}: \|v\|_{\rkhs{K}} \leq 1\}$ denote the unit ball in $\mathcal{H}_K$.
The \acf{kgd} is defined as
\begin{equation}
	\KGD_{K}(Q) 
        \coloneqq \sup_{\substack{ v \in \mathcal{B}_K\ \text{s.t.} \\ (\Op{Q}v)_{-} \in \cL^1(Q) } } \left\lvert \int \Op{Q} v(x) \; \dd Q(x) \right\rvert 
\end{equation}
where, compared to \eqref{eq: first def}, we have taken $\mathcal{V}$ equal to $\mathcal{B}_K$.
\end{definition}

\begin{proposition}[Computable form of \ac{kgd}]
\label{lem: computable}
    Let $Q_0$ have a density $q_0 > 0$ 
    and define $\genpdf{Q}(x) \coloneqq q_0(x) \exp(-\cL'(Q)(x) )$.
    Let $K$ be a kernel for which the elements of $\{\genpdf{Q} v : v\in \mathcal{H}_K\}$ are partially differentiable functions. 
    Suppose $\Op{Q} \mathcal{H}_K \subset \cL^1(Q)$.
    Then 
    \begin{equation}
	\KGD_K(Q) = \left( \iint k_K^Q(x,x') \; \dd Q(x) \dd Q(x') \right)^{1/2}  \label{eq: integral KGD}
    \end{equation}
    where 
    \begin{align}
        k_K^Q(x,x') 
        \coloneqq \sum_{i=1}^d \sum_{j=1}^d \frac{1}{\genpdf{Q}( x) \genpdf{Q}(x')} \partial_{x_j'} \partial_{x_i} \left( \genpdf{Q}( x)  K_{i,j}( x, x') \genpdf{Q}( x') \right) \label{eq: PQ kernel}
    \end{align} 
    is a $Q$-dependent (scalar-valued) kernel. 
\end{proposition}

\noindent The proof is contained in \Cref{app: prelim KGD}.
Note that $\KGD_K(Q_n)$ can be exactly computed for $Q_n$ with finite support at cost $\Omega(d n^2)$, as explained below in \eqref{eq: KGD empirical}, provided the kernel in \eqref{eq: PQ kernel} can be pointwise evaluated.
A simple example of \eqref{eq: PQ kernel} is given in \Cref{rem: explicit formula kgd} of \Cref{app: prelim KGD}.
The use of \ac{rkhs} is a common trick from the machine learning literature, being used heavily in the form of maximum mean discrepancy \citep{smola2007hilbert,sriperumbudur2011universality} and \acp{ksd} \citep{liu2016kernelized,chwialkowski2016kernel,oates2017control}.
Indeed, \ac{kgd} coincides with \ac{ksd} for the average negative log-likelihood loss $\cL$ in \eqref{eq: std bayes}.
Through selection of an appropriate kernel one can obtain several desirable properties of the \ac{kgd}, such as continuity and convergence control; explicit theoretical guarantees will be presented in \Cref{sec: theory}.
Before proceeding to this theory, we first highlight some of the use-cases of \ac{kgd}, in \Cref{sec: applications}.

\section{Illustrative Applications}
\label{sec: applications}

The aim of this section is to highlight some of the novel functionalities and opportunities that are unlocked by our computable measure of suboptimality for
(relative) entropy-regularised variational objectives.
First, we demonstrate how \ac{kgd} can be applied to tuning and comparing existing numerical methods targeting the minimiser $\target$ of \eqref{eq: objective}, in \Cref{subsec: measuring}.
Then, in \Cref{sec: new algs} we demonstrate the potential of \ac{kgd} for developing efficient numerical algorithms, presenting the first extensible sampling algorithm in the context of \eqref{eq: objective}, deriving two new distinct approaches to variational inference without requiring explicit densities for the parametric model, and providing theoretical foundations for nonlinear Stein variational gradient descent \citep{wang2019nonlinear}.

\subsection{Measuring Sample Quality}
\label{subsec: measuring}

A popular algorithm for approximating the solution of \eqref{eq: objective} is \acf{mfld} \citep{del2013mean}, which is based on stochastically evolving a collection of $n$ particles according to 
\begin{align*}
    X_i^{t+1} = X_i^t + \epsilon [(\nabla \log q_0) - \vargrad  \mathcal{L}(Q_n^t)](X_i^t) + \sqrt{2 \epsilon} Z_t^i , \; Z_t^i \stackrel{\mathrm{iid}}{\sim} \mathcal{N}(0,1), \; Q_n^t \coloneqq \frac{1}{n} \sum_{j=1}^n \delta_{X_j^t} ,
\end{align*}
and can be seen as a generalised form of the \ac{ula} \citep{roberts1996exponential} suitable for a variational target.
As with \ac{ula}, the samples produced using \ac{mfld} are biased, but unlike \ac{ula} one cannot apply a Metropolis correction %
since one does not have access to the (unnormalised) density of the target.
This renders the practical selection of the step size $\epsilon$ difficult; convergence to stationarity is slow for $\epsilon$ small, while larger $\epsilon$ incurs a bias of $O(\epsilon)$ \citep{nitanda2022convex}.
As \ac{ula} is a first-order discretisation of the Wasserstein gradient flow of the reverse \ac{kld} \citep{wibisono2018sampling}, higher-order discretisations can be considered, but these too involve parameters whose selection can be difficult.
To resolve this problem, \ac{kgd} provides as a computable alternative to the intractable variational objective $\eobj$, enabling sample quality to be measured and enabling $\epsilon$ to be tuned.
This idea will now be illustrated:

\begin{example}[Mean field neural network]
\label{ex: mfnn}
A \ac{mfnn} is a function of the form
\begin{align}
    f(z) = \mathbb{E}_{X \sim Q}[\Phi(z,X)]  \label{eq: mfnn def}
\end{align}
where $\Phi(\cdot,x)$ represents a small neural network (e.g. a single neuron) with parameters $x \in \mathbb{R}^d$ \citep{nitanda2017stochastic}. %
Introduced as an expressive yet mathematically tractable model, \eqref{eq: mfnn def} simplifies the mathematical analysis of neural networks by approximating the large number of interactions between neurons $\Phi$ with a mean field.
Training an entropy-regularised \ac{mfnn} can be cast as optimisation of $Q$ based on $\eobj$  where $\mathcal{L}$ is the (scaled) empirical risk, i.e.
\begin{equation}\label{eq:empirical_risk}
\mathcal{L}(Q) = \frac{\lambda}{N} \sum_{i=1}^N \ell(y_i , \mathbb{E}_{X \sim Q}[\Phi(z_i,X)] ) ,
\end{equation}
where the regression or classification task is to predict the response $y_i$ from the covariates $z_i$, minimising a predictive loss $\ell$, and the parameters of the neural network are distributed according to $Q$ \citep{hu2021mean,chizat2022mean,nitanda2025propagation}. 
For illustration, we consider a univariate regression task %
with the squared error loss $\ell(y,y') = (y-y')^2$, for which the variational gradient
\begin{align*}
    \vargrad  \mathcal{L}(Q)(x) = - \frac{2\lambda}{N} \sum_{i=1}^N (y_i - \mathbb{E}_{X \sim Q}[\Phi(z_i,X)] ) \nabla_x \Phi(z_i,x)
\end{align*}
follows from calculations in Section 3.1 of \citet{hu2021mean}. 
A synthetic dataset was generated as described in \Cref{app: detail mfnn}, and we let $\Phi(z,x) = w_2 \cdot \tanh (w_1 \cdot z + b_1) + b_2$ be a 2-layer neural network with parameter $x = (w_1,b_1,w_2,b_2)$, so that $d = 4$, and $Q_0 = \mathcal{N}(0,0.5 I_d)$. 
In this experiment we initialise $n = 100$ particles as independent samples from $Q_0$ and then run $T =  10^3$ steps of \ac{mfld} with step size $\epsilon$, letting $Q_\epsilon$ be the empirical distribution so-obtained.
Full details for this experiment are contained in \Cref{app: detail mfnn}.

Our claim is that minimisation of $\KGD_K(Q_\epsilon)$ provides a principled criterion for selection of the step size $\epsilon$ in \ac{mfld}.
To assess this claim, we calculated the test (generalisation) error of the \ac{mfnn}, i.e. \eqref{eq:empirical_risk} computed on a set of samples not seen during training, with $Q_\epsilon$ as $\epsilon$ is varied.
Results in \Cref{fig: tuning} are based on the kernel $K(x,x') = k(x,x') I_{d\times d}$ with $k(x,x') = (1 + \|x-x'\|^2)^{-1/2}$, and indicate that both the test error and $\KGD_K(Q_\epsilon)$ are approximately minimised when $\epsilon \in [10^{-4},10^{-3}]$, supporting the use of \ac{kgd} in this context.
As a further validation, we formed an \emph{ad hoc} kernel density approximation $\hat{Q}_\epsilon$ (i.e. a smoothed version of $Q_\epsilon$; \Cref{app: detail tuning}) and reported estimates of $\eobj(\hat{Q}_\epsilon)$, since $\eobj(Q_\epsilon)$ is not well-defined.
From \Cref{fig: tuning} we observe that the surrogate objective $\eobj(\hat{Q}_\epsilon)$ is minimised 
in the same range obtained using \ac{kgd}.
Further, $\KGD_K(Q_\epsilon)$ and $\KGD_K(\hat{Q}_\epsilon)$ are in close agreement, suggesting the kernel density estimate $\hat{Q}_\epsilon$ is a reasonable approximation to $Q_\epsilon$ in this low-dimensional setting, so $\eobj(\hat{Q}_\epsilon)$ should be an accurate reflection of $\eobj(Q_\epsilon)$.  

Together, these results illustrate the potential of \ac{kgd} as a tuning criterion for algorithms based on \ac{mfld}.
For example, \ac{kgd} can provide a principled solution to the challenging practical problem of balancing the step size $\epsilon$, the number of particles $N$, and time horizon $T$, subject to a fixed computational budget.
\end{example}

\begin{figure}[t]
\centering
    \includegraphics[width=\textwidth]{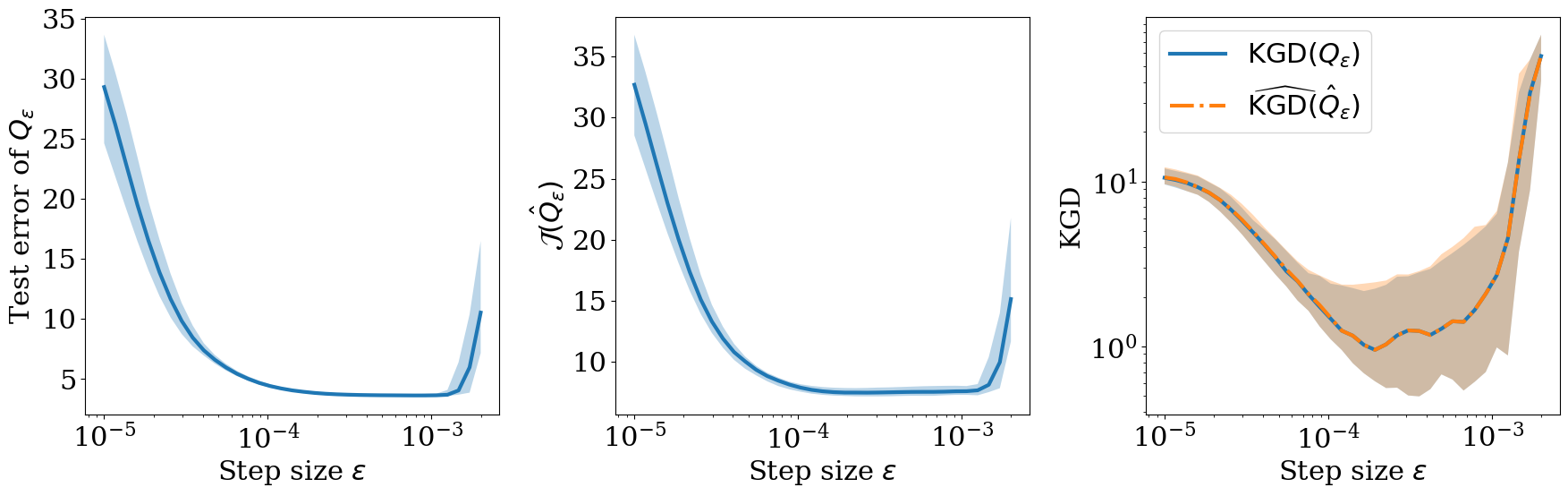}
    \caption{\ac{kgd} provides a computable measure of sample quality for approximations $Q_\epsilon$ to the minimiser $\target$ of $\eobj$ in \eqref{eq: objective}, which we use here to select the step size $\epsilon$ in \acf{mfld} in the setting of \Cref{ex: mfnn}.  
    For $\epsilon$ too small the dynamics does not converge within the $T = 10^3$ iteration limit, while for $\epsilon$ too big the $O(\epsilon)$ bias is substantial.
    The \ac{kgd} is minimised when $\epsilon \in [10^{-4},10^{-3}]$, which is consistent with minimising both the test (generalisation) error and the surrogate objective $\eobj(\hat{Q}_\epsilon)$.
    In total 50 replicates were performed, with the median, $10$th and $90$th percentiles reported.
    }
    \label{fig: tuning}
\end{figure}

\subsection{Efficient Numerical Algorithms}
\label{sec: new algs}

Beyond tuning existing algorithms, \ac{kgd} offers an opportunity to develop novel algorithms with the potential to be more efficient than \ac{mfld}.
This is an important goal, since numerical methods are under-developed in the post-Bayesian setting, and because post-Bayesian targets $P$ are typically more complicated than their Bayesian counterparts due to the absence of a Bernstein--von Mises phenomenon when the model is misspecified \citep{mclatchie2025predictively}.
Here we distinguish between methods that exploit differentiation through the \ac{kgd} (\Cref{sec: autodiff KGD}), and more widely-applicable methods for which differentiation through the \ac{kgd} is not required (\Cref{sec: no autodiff}).

\begin{figure}[t]
\centering
    \includegraphics[width=\textwidth]{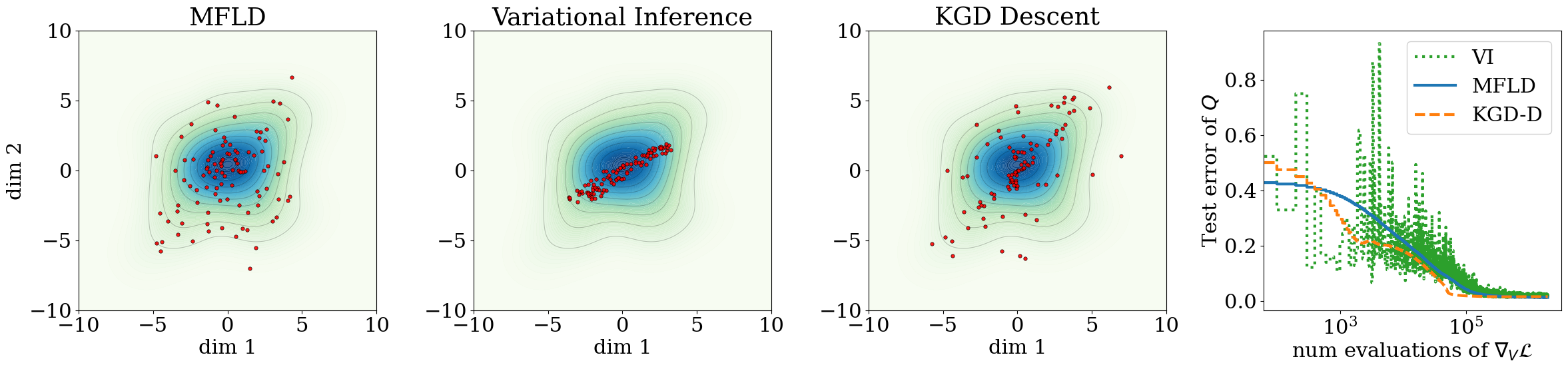}
    \caption{Comparing \acf{mfld} (left) with new sampling schemes that exploit differentiation of the \ac{kgd}, in the setting of \Cref{ex: mfnn}.
    Centre left: The parametric variational inference method of \Cref{sec: variational}.
    Centre right:  The \ac{kgd} descent method of \Cref{sec: KGD descent}.
    All experiments shown here entailed a comparable computational cost.
    The contour lines represent a two-dimensional marginal of the target $\target$, and were obtained at substantial computational cost using a small step size and a large number of particles in \ac{mfld}. 
    Right:  The performance of these methods was measured using the test generalisation error, as a function of the number of backpropagations (evaluations of $\vargrad \cL$) performed on the neural network, which is the principal computational bottleneck.
    Full details are contained in \Cref{app: detail mfnn}.
    }
    \label{fig: different methods}
\end{figure}

\subsubsection{Methods Based on Differentiation of KGD}
\label{sec: autodiff KGD}

For regular problems where automatic differentiation can be exploited, such as the \ac{mfnn} in \Cref{ex: mfnn}, two novel algorithms are presented:

\paragraph{Parametric Variational Inference}
\label{sec: variational}

Parametric variational inference methods, which range in their sophistication from mean field variational Bayes \citep{fox2012tutorial} to methods based on normalising flows \citep{papamakarios2021normalizing}, are widely used in the standard Bayesian context and can be applied also to \eqref{eq: objective}.
Indeed, if we let $\{Q_\theta\}_{\theta \in \Theta}$ be a parametric approximating class for which the density of each $Q_\theta$ is explicit, we can practically search for a parameter $\theta$ for which $\eobj(Q_\theta)$ is minimised.
\ac{kgd} opens up the possibility of performing parametric variational inference \emph{without} requiring explicitly tractable densities for each $Q_\theta$; for example, one can consider $Q_\theta = T_\#^{\theta} \mu_0$ for arbitrary transport maps $T^\theta : \mathbb{R}^{d_0} \rightarrow \mathbb{R}^d$ and reference distributions $\mu_0 \in \cP(\R^{d_0})$ of arbitrary dimension $d_0 \in \mathbb{N}$, relaxing the usual requirement for $T^\theta$ to be a diffeomorphism with an easily-computed Jacobian determinant. %
Indeed, while for a given probability density $Q_{\theta}$ it is usually not possible to exactly compute $\KGD_K(Q_\theta)$ or its gradient with respect to $\theta$, it is generally easy to sample from $Q_{\theta}$ and to construct an approximation using Monte Carlo. 
To this end, we consider \ac{kgd} as a computable variational objective
$$
\theta_\star \in \argmin_{\theta \in \Theta} \; \KGD_K(Q_\theta)
$$ 
to select an appropriate parameter $\theta_\star$.
Beyond simply being computable, the freedom to select the kernel $K$ offers an opportunity to tune the estimator for best performance at the task at hand \citep{chen2025weighted}.
In particular, this generalises the \ac{ksd} variational inference approach of \citet{fisher2021measure}.

The favorable sample complexity of \ac{kgd}, with a dimension-independent convergence rate, ensures this objective can be efficiently stochastically approximated: 

\begin{proposition}[Sample complexity of KGD] \label{prop: sample complexity}
Consider the setting of \Cref{prop:continuity} (cf. \Cref{sec:continuity}) with $Q \in \finitemomentspace{2\polyorder}$ for some $\polyorder \geq 0$.
Let $(x_i)_{i \in \mathbb{N}}$ be a sequence of independent samples from $Q \in \mathcal{P}(\mathbb{R}^d)$ and let $Q_n \coloneqq \frac{1}{n} \sum_{i=1}^n \delta_{x_i}$.
Then 
\begin{align*}
   \sqrt{n}( \KGD_K^2(Q_n) - \KGD_K^2(Q) ) \stackrel{\mathrm{d}}{\rightarrow} \mathcal{N}(0,\sigma_{K,Q}^2) 
\end{align*}
for some\footnote{Here $\mathcal{N}(\mu,\sigma^2)$ denotes a normal distribution with mean $\mu$ and variance $\sigma^2$, with the case $\sigma = 0$ being interpreted as a Dirac distribution $\delta_\mu$.} $\sigma_{K,Q}^2 \in [0,\infty)$.
\end{proposition}

\noindent The proof is contained in \Cref{app: sample complexity}; the exact expression for $\sigma_{K,Q}^2$ is complicated and dimension-dependent but it can in principle be deduced from the explicit calculations in the proof.
The setting of \Cref{prop:continuity} allows for the kernel $K$ to grow at a rate that depends on $\polyorder$ (cf. \Cref{asm: kernel condition 1} of \Cref{asm: for continuity}), 
so that one can trade-off growth of the kernel with the existence of moments of $Q$; see \Cref{sec:continuity}. 
As an illustration, we took $T_\#^{\theta}$ to be a fully-connected feed-forward neural network composed of three hidden layers, with ELU activation function, a uniform reference measure $\mu_0 = \mathcal{U}([-3,3]^{d_0} )$ in dimension $d_0 = 4$, and learned $\theta_\star$ by Adam optimisation on the (squared) \ac{kgd}.
Note that, in contrast to architectures used in normalizing flows \citep{dinh2016density}, $Q_\theta$ does not admit a density with an easily computable Jacobian determinant, and as such (stochastic) gradient descent on the original objective $\eobj(Q_\theta)$ would have been impractical.
The resulting approximation is displayed in the centre left panel of \Cref{fig: different methods}.
It can be seen that the approximation quality is modest compared to \ac{mfld}; we suspect this is due to the non-convexity of the optimisation problem in $\theta$, meaning the optimiser may not have converged.
Full details are contained in \Cref{app: detail mfnn}.

\paragraph{KGD Descent}
\label{sec: KGD descent}

When we no longer require explicit densities for the parametric approximating distributions $Q_\theta$, a natural extension of the approach in \Cref{sec: variational} is to take $Q_\theta$ to be an empirical distribution $Q_n$ parametrised by the support points $\{x_i\}_{i=1}^n$.
This can be instantiated as gradient descent on $(x_1,\dots,x_n) \mapsto \KGD_K^2(Q_n)$, noting that the \ac{kgd} of an empirical distribution with finite support can be explicitly computed as
\begin{align}
    \KGD_K^2(Q_n) = \frac{1}{n^2} \sum_{i=1}^n \sum_{j=1}^n k_K^{Q_n}(x_i,x_j) , \qquad Q_n \coloneqq \frac{1}{n} \sum_{i=1}^n \delta_{x_i} .  \label{eq: KGD empirical}
\end{align}
A continuous-time gradient descent then corresponds to the coupled system of \acp{ode}
\begin{align}
\frac{\mathrm{d}x_i^t}{\mathrm{d}t} & = - \vargrad \KGD_K^2(Q_n^t)(x_i^t) , \qquad Q_n^t \coloneqq \frac{1}{n} \sum_{j=1}^n \delta_{x_j^t} ,
\label{eq: kgdd}
\end{align}
with $\{x_i^0\}_{i=1}^n$ initialised, for example, as independent samples from some $\mu_0 \in \cP(\R^d)$. 
Numerically integrating \eqref{eq: kgdd} forward in time results in a deterministic algorithm that generalises (to nonlinear $\cL$) the \ac{ksd} descent algorithm of \citet{korba2021kernel}.
Note that the `descent' here is with respect to the (squared) \ac{kgd}, while \ac{mfld} performs descent with respect to $\eobj$.
For illustration, we performed \ac{kgd} descent using $n = 100$ particles initialised at $\mu_0 = \mathcal{N}(0,9\idmat)$, and using Adam optimisation to integrate the system of \acp{ode} in \eqref{eq: kgdd}.
The resulting approximation is displayed in the centre right panel of \Cref{fig: different methods}; this method achieved the lowest test errors and the lowest values of \ac{kgd}.
Full details are contained in \Cref{app: detail mfnn}.

\subsubsection{Generally-Applicable Methods Based on KGD}
\label{sec: no autodiff}

The methods described in \Cref{sec: autodiff KGD} exploited differentiation through the \ac{kgd} but, as the following example illustrates, this is not always practical:

\begin{example}[Predictively Oriented Posteriors]
\label{ex: LV}
Given a (possibly misspecified) parametric statistical model $p(\cdot | x)$ for independent data $\{y_i\}_{i=1}^N \subset \mathcal{Y}$, consider lifting $p(\cdot | x)$ to a mixture model $\int p(\cdot | x) \, \mathrm{d}Q(x)$, now parametrised by the mixing distribution $Q$.
A \emph{PrO posterior} associated to a kernel scoring rule is defined as the minimiser $P$ of the entropy-regularised objective \eqref{eq: objective} with
\begin{align}
\cL(Q) = \frac{1}{2 \lambda_N} \mathrm{MMD}^2\left( \int p(\cdot | x) \; \mathrm{d}Q(x) , \frac{1}{N} \sum_{i=1}^N \delta_{y_i} \right) \label{eq: wass grad PCUQ}
\end{align}
where the learning rate $\lambda_N$ is user-determined \citep[see e.g.][]{shen2024prediction,mclatchie2025predictively,lai2024predictive}.
In the well-specified case, $P$ collapses to a point mass on the true data-generating parameter $x^\dagger$ as $N \rightarrow \infty$, while in the case of `non-trivial' misspecification \citep[in the sense of][Definition 2]{mclatchie2025predictively} an asymptotic collapse of uncertainty is avoided\footnote{For instance, one can contrast learning a Gaussian model in the standard Bayesian framework with \eqref{eq: wass grad PCUQ}; the latter enables consistent approximation of all (sufficiently regular) non-Gaussian $P$ via the learning of a non-trivial mixing distribution $Q$, in effect lifting the Gaussian model to a Gaussian mixture model \citep[][Chapter 3]{goodfellow2016deep}.}.
This is in contrast to Bayesian methods, which typically collapse to a point mass in the large data limit whether or not the model is well-specified.
To avoid overloading notation, let $\kappa : \mathcal{Y} \times \mathcal{Y} \rightarrow \mathbb{R}$ be the kernel defining the \ac{mmd} \citep[see][for background]{smola2007hilbert}.
Following identical calculations to Section 5.1 of \citet{shen2024prediction}, we obtain the variational gradient
\begin{align*}
\vargrad  \cL(Q)(x) & = \frac{1}{\lambda_N} \int \nabla_1 w(x,x') \; \dd Q(x') 
\end{align*}
where, up to a term that is constant in $x$,  %
\begin{align*}
    w(x,x') & \coloneqq \iint \kappa(y,y') p(y|x) p(y'|x') \; \mathrm{d}y \mathrm{d}y' - \frac{1}{N} \sum_{i=1}^N \int \kappa(y_i,y) p(y | x) \; \mathrm{d}y + \text{cst} .
\end{align*}
Here we consider parameter inference for the Lotka--Volterra \ac{ode} model following Section 6.1 of \citet{shen2024prediction}; full details for the experiments that we are about to report are contained in \Cref{app: detail pcuq}.
Evaluation of the variational gradient of \eqref{eq: wass grad PCUQ} requires differentiating through the \acp{ode}, which entails a non-trivial computational cost.
\end{example}

\noindent One could in principle apply the methods described in \Cref{sec: autodiff KGD} to \Cref{ex: LV} as efficient alternatives to \ac{mfld}, but differentiating through the \ac{kgd} here entails computing second-order derivatives of the solution of the \ac{ode}, which is notoriously difficult and should generally be avoided.
As such, alternative methods are needed.
Two such methods will now be presented:

\begin{figure}[t]
\centering
    \includegraphics[width=\textwidth]{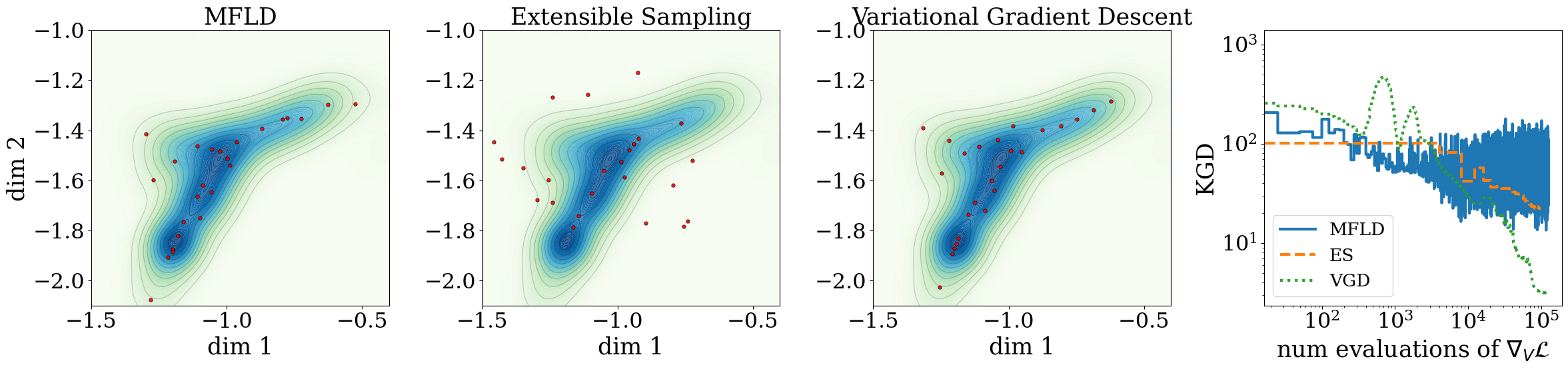}
    \caption{Comparing \acf{mfld} (left) with two new sampling schemes that do not require differentiation of the \ac{kgd}, in the setting of \Cref{ex: LV}.
    Centre left:  The extensible sampling method of \Cref{sec: extensible}, which generates a sequence $(x_n)_{n \in \mathbb{N}}$%
    .
    Centre right:  The variational gradient descent method of \Cref{sec: svgd}.
    The contour lines represent the target $\target$, and were obtained at substantial computational cost using a small step size and a large number of particles in \ac{mfld}.
    Right:  The performance of these methods was measured using \ac{kgd}, as a function of the number of \ac{ode} solves (evaluations of $\vargrad \cL$), which is the principal computational bottleneck.
    Full details are contained in \Cref{app: detail pcuq}.
    }
    \label{fig: different methods 2}
\end{figure}

\paragraph{Extensible Sampling}
\label{sec: extensible}

An \emph{extensible} sampling algorithm adds one particle at a time such that the empirical distribution of the particles converges to the target.
The simplest example of an extensible sampling algorithm is Monte Carlo; the benefit of an extensible algorithm is that one is able to continue running the algorithm until a desired level of accuracy has been achieved.
However, to-date there does not appear to exist an extensible sampling algorithm suitable for targets defined as minimisers of \eqref{eq: objective}. 
For example, in \ac{mfld} one must specify a step size and a finite number of particles at the outset; due to this discretisation bias, the resulting approximation does not converge to $\target$, even in the ergodic limit.
Our aim here is to illustrate an extensible algorithm based on \ac{kgd}, the simplest of which would be the greedy algorithm
\begin{align}
x_n \in \argmin_{x \in \mathbb{R}^d} \; \KGD_K \left( \frac{1}{n} \delta_x + \frac{1}{n} \sum_{i=1}^{n-1} \delta_{x_i} \right) \label{eq: greedy}
\end{align}
which adds one sample $x_n$ at a time to the existing samples $\{x_i\}_{i=1}^{n-1}$ in such a manner that the \ac{kgd} between the empirical measure and the target is minimised.
Greedy algorithms can be realised using any (deterministic or randomised) numerical optimisation method applied to \eqref{eq: greedy}.
Their potential has been demonstrated in the standard Bayesian setting, where \ac{kgd} coincides with \ac{ksd}  \citep{chen2018stein,chen2019stein,teymur2021optimal,riabiz2022optimal}; our contribution unlocks this class of algorithms for variational targets.
Samples produced using the greedy algorithm applied to \Cref{ex: LV} are shown in the second panel of \Cref{fig: different methods 2}.
Compared to samples from \ac{mfld}, we can see wider spacing among the particles; we conjecture that this occurs because once a sub-optimal sample is selected (e.g. too far into the tail, due to imperfect numerical solution of \eqref{eq: greedy}) then future samples are required to `compensate' for the bias that was introduced (e.g. by being too far into the opposite tail).
This interpretation is supported by inspection of how the samples are sequentially generated in \Cref{fig: different methods 2 evolution} of \Cref{app: detail pcuq}.

\paragraph{Variational Gradient Descent}
\label{sec: svgd}

\Ac{svgd} is a nonparametric variational inference method for the standard Bayesian setting, which has attracted considerable recent interest \citep{liu2016stein}.
A heuristic argument was used  to generalise \ac{svgd} to (relative) entropy-regularised variational objectives such as \eqref{eq: objective} in \citet{wang2019nonlinear}.
However, the consistency of \ac{svgd} in this generalised setting was an open problem; we resolve this by presenting non-asymptotic error bounds in which \ac{kgd} features as a key ingredient.

To motivate this generalisation, consider the directional derivatives 
$$
\left. \frac{\mathrm{d}}{\mathrm{d}\epsilon} \eobj((\Id + \epsilon v)_\# Q) \right|_{\epsilon = 0} 
$$
as specified by a suitable vector field $v : \R^d \rightarrow \R^d$, where $\Id$ is the identity map on $\R^d$.
For the purpose of optimisation, we seek a vector field $v$ for which the rate of decrease in $\eobj$ is maximised. 
To this end, 
let $Q_t^\epsilon \coloneqq (1-t) Q + t(\Id + \epsilon v)_\# Q$ denote a path connecting $Q$ and $(\Id + \epsilon v)_\# Q$ in $\finitemomentspace{}$.
Then, from the fundamental theorem of calculus applied to $t \mapsto \eobj(Q_t^\epsilon)$ as $t$ ranges from $0$ to $1$, and definition of the first variation, 
\begin{align*}
    \lim_{\epsilon \rightarrow 0} \frac{\eobj((\Id + \epsilon v)_\# Q) - \eobj(Q)}{\epsilon} & = \lim_{\epsilon \rightarrow 0} \frac{1}{\epsilon} \int_0^1 \int \eobj'(Q_t^\epsilon)(x) \; \mathrm{d}((\Id + \epsilon v)_\# Q - Q)(x) \mathrm{d}t \\
    & = \lim_{\epsilon \rightarrow 0} \frac{1}{\epsilon} \int_0^1 \int ( \eobj'(Q_t^\epsilon)(x + \epsilon v(x)) - \eobj'(Q_t^\epsilon)(x) ) \; \mathrm{d}Q(x) \mathrm{d}t \\
    & =\int \vargrad  \eobj(Q)(x) \cdot v(x) \; \mathrm{d}Q(x) ,
\end{align*}
where we have assumed sufficient regularity to interchange limit and integral.
The final term we already encountered in \eqref{eq: int of W grad}, so we can evoke \eqref{eq: gen st op} to get
\begin{align}
    \left. \frac{\mathrm{d}}{\mathrm{d}\epsilon} \eobj((\Id + \epsilon v)_\# Q) \right|_{\epsilon = 0} = - \int \Op{Q} v(x) \; \dd Q(x) .  \label{eq: KL to KGD}
\end{align}
That is, the directional derivative of the objective $\eobj$ in \eqref{eq: objective} can be expressed as an explicit $Q$-dependent linear functional applied to the vector field.

Next, following the same logic as \citet{wang2019nonlinear}, we let\footnote{The construction can be extended to general matrix-valued kernels, following \citet{wang2019stein} in the case of \ac{svgd}.} $K(x,x') = k(x,x') \idmat$ and pick the vector field $v_Q$ from the unit ball of $\mathcal{H}_K$ for which the magnitude of the negative gradient in \eqref{eq: KL to KGD} is maximised, i.e. 
\begin{align*}
    v_Q(\cdot) \propto \int \{k(x,\cdot ) (\nabla \log q_0 - \vargrad  \cL(Q))(x) + \nabla_1 k(x, \cdot) \} \; \dd Q(x),
\end{align*}
to arrive at
\begin{align}
\left. \frac{\mathrm{d}}{\mathrm{d}\epsilon} \eobj((\Id + \epsilon v_Q )_\# Q) \right|_{\epsilon = 0}
= - \KGD_K(Q) . \label{eq: descent and separation}
\end{align}
This reveals an explicit relationship between steepest descent on $\eobj$ and \ac{kgd}.
Assuming the kernel was selected to ensure \ac{kgd} \emph{characterises stationarity}, meaning that $\KGD_K(Q)$ vanishes if and only if $Q$ is a stationary point of \eqref{eq: objective} (cf. \Cref{subsec: separation}), following the path $(\mu_t)_{t \geq 0}$ of steepest descent on $\eobj$, initialised say at $\mu_0$, should lead to a strictly decreasing $\eobj(\mu_t)$, unless or until $\mu_t$ and $\target$ are equal.
To numerically approximate this gradient descent, we initialise $\{x_i^0\}_{i=1}^n$ as independent samples from $\mu_0$ at time $t = 0$ and then update $\{x_i^t\}_{i=1}^n$ deterministically, via the coupled system of \acp{ode}
\begin{align}
\frac{\mathrm{d}x_i^t}{\mathrm{d}t} & = \frac{1}{n} \sum_{j=1}^n k(x_i^t , x_j^t) (\nabla \log q_0 - \vargrad  \cL(Q_n^t))(x_j^t) + \nabla_1 k(x_j^t , x_i^t) , \quad Q_n^t \coloneqq \frac{1}{n} \sum_{j=1}^n \delta_{x_j^t}
\label{eq: gen svgd}
\end{align}
up to a time horizon $T$.
To emphasise the generality of this algorithm we call it simply \emph{\acl{vgd}} (\acs{vgd}).

Since the above presentation was informal, the consistency of \ac{vgd} will now be established.
For a sufficiently regular $\cF : \finitemomentspace{} \rightarrow \R$, let $\nabla \vargrad \cF(Q)$ denote the function $x \mapsto \nabla_x^2 \cF'(Q)(x)$
and let $\vargrad  \cdot \vargrad  \mathcal{F}(Q)$ denote the function $(x,y) \mapsto \sum_i \partial_i \mathcal{G}_{x,i}'(Q)(y)$ where $\mathcal{G}_x(Q) \coloneq \vargrad \mathcal{F}(Q)(x)$.
In addition, we overload $\genpdf{Q}$ to also act as a shorthand for the distribution whose density is proportional to $q_0(x) \exp(-\cL'(Q)(x) )$, which is well-defined under \Cref{assum: rho well def} in \Cref{prop: svgd converge}.

The following extends the analysis of \ac{svgd} in Theorem 1 in \citet{banerjee2025improved} to the more general \ac{vgd}:

\begin{proposition}[Convergence of \ac{vgd}]
\label{prop: svgd converge}
Assume that:
\begin{enumerate}
    \item \emph{(Integrability)} $\exp(-\cL'(Q)) \in \mathcal{L}^1(Q_0)$ for all $Q \in \finitemomentspace{}$ \label{assum: rho well def} %
    \item \emph{(Loss)} \label{item: loss reg asm}
    for each $Q \in \dom(\cL)$, the first variation $\cL'(Q)$ exists and is $C^2(\mathbb{R}^d)$, and $\vargrad \cdot \vargrad \cL(Q)$ exists.
    For each $n \in \mathbb{N}$, $i \in \{1,\dots,n\}$ and $y \in \mathbb{R}^d$, and with $Q_n \equiv Q_n(\mathbf{x})$ denoting the empirical distribution of $\mathbf{x} = (x_1,\dots,x_n) \in \mathbb{R}^{d \times n}$,
    the map $\mathbf{x} \mapsto \cL(Q_n)$ is real-valued and $C^1(\mathbb{R}^{d \times n})$ with $\nabla_{x_i} \cL(Q_n) = \frac{1}{n} \vargrad \cL(Q_n)(x_i)$, 
    and the map $(\mathbf{x},y) \mapsto \vargrad \cL(Q_n)(y)$ is $C^1(\mathbb{R}^{d \times n} \times \mathbb{R}^d)$ with $\nabla_{x_i} \cdot [\vargrad \cL(Q_n)(y)] = \frac{1}{n} \vargrad \cdot \vargrad \cL(Q_n)(y)(x_i)$.
    Moreover,
    \begin{enumerate}
        \item[a.] $\mathfrak{b}_0 := \sup_{Q \in \dom(\cL)} \big\| \nabla_{\mathrm{V}} \mathcal{L}(Q)(0) \big\|_2 < \infty$
        \item[b.] $\mathfrak{M} \coloneq \sup_{Q \in \dom(\cL) , \, x \in \mathbb{R}^d} \| \nabla \vargrad  \mathcal{L}(Q)(x) \|_{\mathrm{op}} < \infty$, where $\|\cdot\|_{\mathrm{op}}$ is the operator norm associated with $\|\cdot\|_2$, 
        \item[c.] $\sup_{Q \in \dom(\cL) , \, x,x' \in \mathbb{R}^d} | \vargrad  \cdot \vargrad  \mathcal{L}(Q)(x)(x') | < \infty$ .
    \end{enumerate}
    \item \emph{(Regularisation)} the map $x \mapsto \log q_0(x)$ is $C^2(\mathbb{R}^d)$ with
    \begin{enumerate}
        \item[a.] $\sup_{x \in \mathbb{R}^d} \| \nabla^2 \log q_0(x) \| < \infty$, 
        \item[b.] $\int_{\mathbb{R}^d} q_0(x) \exp\!\left( \tfrac{\mathfrak{M}}{2}\|x\|_2^2 + \mathfrak{b}_0 \|x\|_2 \right) \mathrm{d}x < \infty$.
    \end{enumerate}
    \item \emph{(Initialisation)} %
    $\mu_0$ has bounded support, and has a density that is bounded.
    \item \emph{(Kernel)} Let $k(x,x') = \phi(x-x')$ be a positive-definite kernel, for some $\phi \in C^2(\mathbb{R}^d)$. 
\end{enumerate}
Then the dynamics defined in \eqref{eq: gen svgd} satisfies
    \begin{align*}
        \frac{1}{T} \int_0^T \mathbb{E}[ \KGD_K^2(Q_n^t) ] \; \mathrm{d}t \leq \frac{C_0}{T} + \frac{C_K}{n}
    \end{align*}
    for some finite ($n$ and $T$-independent) constants $C_0$ and $C_K$. 
\end{proposition}

\noindent The proof is contained in \Cref{app: svgd}.
\Cref{prop: svgd converge} states that increasing the number of particles $n$ and the time horizon $T$ ensures convergence in time-averaged mean-square \ac{kgd}; the expectation arises due to the random initialisation of $\{x_i\}_{i=1}^n$ i.i.d. sampled from $\mu_0$.
The assumptions in \Cref{item: loss reg asm} can be illustrated in the setting of \Cref{ex: LV}, where $\nabla_{\mathrm{V}} \mathcal{L}(Q)(x) = \int \nabla_1 w(x,y) \mathrm{d}Q(y)$, $\nabla \vargrad  \mathcal{L}(Q)(x) = \int \nabla_1^2 w(x,x') \; \mathrm{d}Q(x')$ and $\vargrad  \cdot \vargrad  \mathcal{L}(Q)(x)(x') = \nabla_1 \cdot \nabla_2 w(x,x')$, so our assumption on $\mathcal{L}$ in \Cref{prop: svgd converge} amounts to requiring that $\nabla_1 w$, $\nabla_1^2 w$ and $\nabla_1 \cdot \nabla_2 w$ are bounded; note that we do not require $w$ itself to be bounded.
The two chain rules in \Cref{item: loss reg asm}, $\nabla_{x_i} \cL(Q_n) = \frac{1}{n} \vargrad \cL(Q_n)(x_i)$ and $\nabla_{x_i} \cdot [\vargrad \cL(Q_n)(y)] = \frac{1}{n} \vargrad \cdot \vargrad \cL(Q_n)(y)(x_i)$, relate differentiation with respect to the location of a particle to the variational gradient; for linear losses and for interaction energies such as \eqref{eq: wass grad PCUQ} both follow by direct differentiation.

Samples produced using \ac{vgd} are presented in the rightmost panel of \Cref{fig: different methods 2}.
Here we employed a kernel $k(x,x') = \sum_{i=1}^3(1 + \|x-x'\|^2/\ell_i^2)^{-1/2}$  with lengthscales $\ell_1 = 0.001$, $\ell_2 = 0.01$ and $\ell_3 = 0.1$, which we found to work well.
As with extensible sampling (\Cref{sec: extensible}), \ac{vgd} provides a more structured representation of the target than \ac{mfld}.
However, compared to extensible sampling, the high probability region is better covered under \ac{vgd}.

\medskip

This concludes our discussion of algorithms in this manuscript, but several further opportunities are highlighted in \Cref{sec: discuss}.
Of the algorithms we presented, the most accurate approximations tended to be produced by \ac{kgd} descent and \ac{vgd}, which both out-performed \ac{mfld} at least insofar as can be measured using \ac{kgd}.
This motivates the question: \emph{is \ac{kgd} a meaningful performance measure?}
This is a difficult question as we do not have access to $P$ in general.

One way to proceed is to formulate \eqref{eq: objective} so that certain marginals of $P$ must be identical via a symmetry argument; comparing the marginals produced by numerical methods thus provides an indication of how well the numerical methods are working.
This approach was pursued in \Cref{app: marginals}, and we found that the smallest values of \ac{kgd} were achieved where the corresponding marginals were closest to being identical.
However, this is not a rigorous argument to establish \ac{kgd} is meaningful.
A rigorous justification is provided in \Cref{sec: theory}, through detailed theoretical analysis of \ac{kgd}, the main contribution of this work.

\section{Theoretical Guarantees}
\label{sec: theory}

This section establishes theoretical foundations for \ac{kgd}.
To begin we consider when \ac{kgd} characterises stationarity, meaning that $\KGD_K(Q) = 0$ if and only if $Q$ is a stationary point of \eqref{eq: objective}, with sufficient conditions presented in \Cref{subsec: separation}.
Recall that characterising stationarity arose as a desirable property in the context of \ac{vgd} in \Cref{sec: svgd}; however, this property alone does not imply that small values of \ac{kgd} correspond to the closeness of $Q$ to a stationary point.
To address this, sufficient conditions for continuity of $Q \mapsto \KGD_K(Q)$ are presented in \Cref{sec:continuity}.
Sufficient conditions for \eqref{eq: objective} to have a unique stationary point $\target$ are presented in \Cref{sec:minimizer}.
Finally, we leverage continuity to establish conditions under which \ac{kgd} controls $\polyorder$-convergence, meaning that $\KGD_K(Q_n) \rightarrow 0$ implies $\int h \; \dd Q_n \rightarrow \int h \; \dd \target$ for all continuous $h$ with $|h(x)| \lesssim 1 + \|x\|^\polyorder$, in \Cref{sec:cv_control}.
This provides theoretical justification for our algorithms that aim to minimise \ac{kgd} in \Cref{sec: new algs}.

\begin{assumption}
\label{asm: diff loss}
    The loss function $\cL : \mathcal{P}(\mathbb{R}^d) \rightarrow \mathbb{R}$ admits a variational gradient $\vargrad\cL(Q)$ for each $Q \in \finitemomentspace{}$. 
\end{assumption}

\noindent \Cref{asm: diff loss} is couched in terms of $\vargrad \cL$ as an aesthetic choice for the main text, but our analysis extends to a weaker notion of functional differentiability, such that the results in this section hold also in settings where the domain of $\cL$ is a subset of $\finitemomentspace{}$; we refer the interested reader to a discussion of this technical point in \Cref{app: relax def gd}.

\begin{assumption}
\label{asm: q0 support}
    There exists a positive density $q_0$ for the reference distribution $Q_0$. 
\end{assumption}

\noindent 
Examples of functionals $\cL(Q)$ with explicit variational gradients on $\mathcal{P}(\R^d)$ 
include potential energies $\cL_1(Q) = \int u(x) \; \dd Q(x)$ and interaction energies $\cL_2(Q) = \frac{1}{2} \iint w(x-x')\ \dd Q(x)\dd Q(x')$. 
 In particular, their variational gradients read as $\vargrad \cL_1(Q) = \nabla u$ and $\vargrad \cL_2(Q) = (\nabla v) \star Q$. 
 More generally, if $\cL$ is well-defined for discrete measures, i.e.,  $\mathrm{L}(x_1 , \dots , x_n) \coloneqq \cL(Q_n)$ for $Q_n = \frac{1}{n} \sum_{i=1}^n \delta_{x_i}$, then the variational gradient can be identified with the Euclidean gradient up to a multiplicative factor: $\vargrad  \cL(Q_n)(x_i) = n \nabla_i \mathrm{L}(x_1,\dots,x_n)$, following a similar argument to \citet[][Theorem 2]{wang2019nonlinear}.

\subsection{Characterising Stationarity}
\label{subsec: separation}

A basic property that we would expect is $\KGD_K(P) = 0$; sufficient conditions for this are contained in \Cref{app: identity}.
This section focuses instead on establishing the stronger property that $\KGD_K(Q) =0$ if and only if $Q$ is a stationary point.
Recall that characterising stationarity arose as a desirable property in the discussion of \Cref{sec: svgd}.
Our arguments build on corresponding analysis for \ac{ksd} in \citet{barp2024targeted} and require the concept of \emph{$C_0^1(\mathbb{R}^d, \mathbb{R}^d)$-universality}:

\begin{definition}[$C_0^1(\mathbb{R}^d, \mathbb{R}^d)$-universality]
\label{def: universal kernel}
    Let $C_0^1(\mathbb{R}^d,\mathbb{R}^{d})$ denote the normed space of continuously differentiable functions $f : \mathbb{R}^d \rightarrow \mathbb{R}^d$ for which $f$ and $\nabla f$ vanish at infinity, equipped with the norm
    \begin{align*}
        \|f\|_{C_0^1} \coloneqq \sup_{x \in \mathbb{R}^d} \|f(x)\| + \|\nabla f(x)\|. 
    \end{align*}
    A (matrix-valued) kernel $L$ is said to be \emph{universal to} $C_0^1(\mathbb{R}^d,\mathbb{R}^{d})$ if $\mathcal{H}_L$ is dense in $C_0^1(\mathbb{R}^d,\mathbb{R}^{d})$.
\end{definition}

\begin{assumption}\label{asm:contdiff}
Both $q_0 \in C^1(\R^d)$ and $\nabla_V \cL(Q) \in C^0(\R^d)$ for each $Q \in \finitemomentspace{}$. 
\end{assumption}

\begin{assumption}\label{asm:separation}
Let $\theta: \mathbb{R}^d \to [c, \infty)$ with some constant $c > 0$. 
The following hold: 
\begin{enumerate}[label=(\roman*)]
    \item $\theta \in C^1(\R^d)$ and $\nabla(1/\theta)$ is continuous and bounded. 
    \item $\Verts{\nabla \log q_0(x)} \lesssim \theta(x)$, and $\Verts{\vargrad  \cL(Q)(x)}\lesssim \theta(x)$ for each $Q \in \finitemomentspace{}$. \label{item:sep-growth-density-loss}
    \item $K$ is given by $K(x, y) = L(x,y) / \{\theta(x)\theta(y)\}$, 
    where $L$ is a $C_0^1(\R^d, \R^d)$-universal kernel with $\rkhs{L} \subset C_0^1(\R^d, \R^d)$.
\end{enumerate}
    
\end{assumption}

\begin{theorem}[\ac{kgd} characterises stationarity]\label{thm:separation} 
    Suppose that \Cref{asm: diff loss,asm: q0 support,asm:separation,asm:contdiff} are satisfied.
    Then for all $Q \in \mathcal{P}(\mathbb{R}^d)$, $\KGD_K(Q) = 0$ if and only if $Q$ is a stationary  point of $\eobj$.
\end{theorem}

\noindent The proof is contained in \Cref{appendix:separation}.
In the case of a unique stationary point $\target$, this result guarantees separation of $P$ from \emph{all} alternatives $Q$ in $\mathcal{P}(\mathbb{R}^d)$, not just from $Q$ with continuously differentiable densities \citep[e.g. as considered for \ac{ksd} in][]{chwialkowski2016kernel,liu2016kernelized}.

Though universality is natural in this context the concept is somewhat technical, and we therefore present a weaker user-friendly statement in \Cref{prop: sufficient KSD separating}.
A function $f$ is said to have \emph{at most root-exponential growth} if $f(x) = O(\mathrm{exp}(C \sum_{i=1}^d \sqrt{|x_i|}))$ for some $C > 0$. 
A continuous function $k : \mathbb{R}^d \times \mathbb{R}^d \rightarrow \mathbb{R}$ is a translation-invariant (scalar) kernel if and only if there exists a finite non-negative (Borel) measure $\mu$ on $\mathbb{R}^d$, called the \emph{spectral measure}, such that $k(x,y) = (2 \pi)^{-d} \int \exp(-\mathrm{i} (x-y) \cdot \omega ) \; \mathrm{d}\mu(\omega)$ for all $x,y \in \mathbb{R}^d$ \citep[][Theorem 6.6]{wendland2004scattered}.
For example, the Gaussian kernel $k(x,x') = \exp(-\|x-x'\|_2^2 / \lambda^2)$ has spectral measure $\mu(\dd \omega) = f(\omega) \dd \omega$ with density $f(\omega) = \{(\lambda / (2\sqrt{\pi})\}^d\cdot \exp(- \lambda^2 \Verts{\omega}_2^2 / 4) $. 

\begin{proposition}[Simple conditions for characterising stationarity]
\label{prop: sufficient KSD separating}
    Suppose that \Cref{asm: diff loss,asm: q0 support,asm:contdiff} are satisfied.
    Assume $\theta$ in \Cref{asm:separation} is of at most root-exponential growth. 
    Take $K=k \idmat$ where %
    $k(x,y) = \phi(x-y)$ for $\phi \in C^2(\mathbb{R}^d)$ with a spectral density bounded away from zero on compact sets.
    Then $\KGD_K(Q) = 0$ whenever $Q$ is a stationary point and $\theta \in \cL^1(Q)$. 
    Moreover, $\KGD_K(Q) > 0$ whenever $Q$ is not a stationary  point of $\eobj$. 
\end{proposition}

\noindent The proof is contained in \Cref{appendix:separation}.
Most of the standard (scalar) kernels $k$, such as the Mat\'{e}rn kernels (of order $\nu \geq 3/2$), the inverse multi-quadric kernels, and the Gaussian kernels meet the conditions of \Cref{prop: sufficient KSD separating}.

\subsection{Continuity}\label{sec:continuity}

Characterising stationarity does not imply that `similar' distributions are associated with small values of \ac{kgd}.
This section therefore establishes conditions for continuity of \ac{kgd}. 
To do so, we need to clarify what `similar' means in this context.
For $\polyorder \in [0, \infty)$, a sequence $(Q_n)_{n \in \mathbb{N}} \subset \finitemomentspace{}$ is said to $\polyorder$-\emph{converge} to $Q \in  \finitemomentspace{\polyorder}$, denoted $Q_n \toL{\polyorder} Q$,  if $\int h \; \mathrm{d} Q_n \rightarrow \int h \; \mathrm{d}Q$ 
for every continuous function $h : \mathbb{R}^d \rightarrow [0,\infty)$ of \emph{$\polyorder$-growth}:  $h(x) \lesssim 1 + \|x\|^\polyorder$. 
(The standard notion of weak convergence is recovered when $\polyorder=0$.)

As with \ac{ksd}, \ac{kgd} can be written as %
a double integral against $Q$ of a kernel. 
In the case of \ac{ksd}, proving continuity with respect to $\polyorder$-convergence is relatively straightforward.
For \ac{kgd} the situation is more complex because the kernel %
that is integrated
in the discrepancy depends on the distribution $Q$, hence we cannot apply the same argument. 
A favourable situation is when $\cL(Q)$ is given in \emph{interaction energy form}\footnote{The assumption of interaction energy form can be slightly relaxed; we defer all discussion to \Cref{cor:continuity-composite} in \Cref{sec: interaction energy proof}.} as $\cL(Q) = \int w(x_1,\dots,x_r) \; \mathrm{d}Q^{\otimes r}(x_1,\dots,x_r)$, which is the case for all of the examples we discussed in \Cref{sec:intro,sec: applications}.
Indeed, after algebraic manipulation, the loss functions for both \Cref{ex: mfnn} and \Cref{ex: LV} can be expressed as $\mathcal{L}(Q) = \int u(x) \; \mathrm{d}Q(x) + \iint v(x,x') \; \mathrm{d}Q(x) \mathrm{d}Q(x')$ for certain $u(x)$ and $v(x,x')$, up to an additive constant.
Then we can take $w(x,x') = \frac{1}{2}(u(x) + u(x')) + v(x,x')$ to write $\cL(Q) = \int w(x,x') \; \mathrm{d}Q^{\otimes 2}(x,x')$.
In general, the following are assumed:

\begin{assumption}
\label{asm: for continuity} 
Assume the following hold for some $\alpha \in [0,\infty)$ and $\beta \in [0, \infty)$: 
\begin{enumerate}[label=(\roman*)]
    \item $\Verts{\nabla \log q_0(x)} \lesssim 1+\Verts{x}^\beta$.  \label{item:asm for continuity growth prior}
    \item The loss $\cL(Q)$ has a variational gradient 
    \begin{equation*}
        \vargrad  \cL(Q)(x) =
        \sum_{i=1}^r \int \nabla_i w(x_1, \dots, x_r)\; \dd Q_{x,i} (x_1,\dots, x_r),
    \end{equation*}
    where $w:\prod_{i=1}^r\mathbb{R}^d \to \mathbb{R}$ is continuously differentiable 
    and satisfies $\Verts{\nabla_i w(x_1, \dots, x_r)} \lesssim (1+\Verts{x_{i}}^\beta) \prod_{j\neq i} (1+\Verts{x_j}^{\polyorder})$ for each $i \in \{1,\dots, r\}$; and 
        \[Q_{x, i}\coloneqq  Q\otimes \dots \otimes Q \otimes \underbrace{\delta_{x}}_{i\mathrm{th\ component}} \otimes Q \otimes \dots \otimes Q.\]  \label{asm: kernel condition 4}
   \item $\sqrt{\Verts{K(x, x)}_{\mathrm{op}}} \lesssim (1+\Verts{x})^{\polyorder-\beta}.$ \label{asm: kernel condition 1} 
    \item The map $(x,y) \mapsto \partial_{1, i}\partial_{2, i}K(x, y)$ is continuous for each $i \in \{1,\dots,d\}$. \label{asm: kernel condition 2}
    \item $\max_{i\in\{1,\dots, d\}} \sqrt{\Verts{\partial_{1, i}\partial_{2, i}K(x, x)}_{\mathrm{op}}} \lesssim (1+\Verts{x})^{\polyorder-\beta}$. \label{asm: kernel condition 3}
\end{enumerate}
\end{assumption}

\noindent For example, if the \ac{mfnn} activation function $\Phi$ in \Cref{ex: mfnn} has both $\Phi$ and $\nabla_2 \Phi$ bounded (e.g. a sigmoid) then \Cref{asm: kernel condition 4} holds for all $\alpha,\beta \geq 0$, while if the activation function has linear growth and $\nabla_2 \Phi$ bounded (e.g. softplus) then \Cref{asm: kernel condition 4} holds for all $\alpha,\beta \geq 1$.
In general, loss functions of interaction energy form satisfy \Cref{asm: kernel condition 4} of \Cref{asm: for continuity} under appropriate regularity on $w$; see \Cref{cor:wgrad4intenergy} in \Cref{sec: interaction energy proof}.

\begin{theorem}[KGD is continuous] \label{prop:continuity}
   Let \Cref{asm: q0 support,asm:contdiff,asm: for continuity} hold.
    Then, we have $\KGD_K(Q_\seqidx) \to \KGD_K(Q)$ whenever $Q_\seqidx \toL{\polyorder} Q$. 
\end{theorem}

\noindent The proof can be found in \Cref{subsec:proof-continuity}.
In the case of a unique stationary point $\target$, \Cref{thm:separation,prop:continuity} imply that the \ac{kgd} `detects $\polyorder$-convergence', in the sense that $\KGD_K(Q_n) \rightarrow 0$ whenever $Q_n \toL{\polyorder} \target$.

Below, we propose a practical kernel that satisfies our assumptions:

\begin{definition}[Recommended kernel]
\label{def: recommended kernel}
    Fix $c > 0$, let $\weight_s(x) \coloneqq \bigl(c^{2}+\Verts x_{2}^{2}\bigr)^{s/2}$ for any $s \in \mathbb{R}$, and define the linear (scalar-valued) kernel $k_{\mathrm{lin}}(x,y)= c^2 +  x \cdot y$.
    Then, for $\polyorder$ and $\beta$ as in \Cref{asm: for continuity}, the recommended form for the (matrix-valued) kernel $K$ is 
    \begin{equation}
     K(x,y)=\weight_{\polyorder-\beta}(x)\left(L(x,y)+\bar{k}_{\mathrm{lin}}(x,y)\idmat\right)\weight_{\polyorder-\beta}(y)
    ,\label{eq:q-growth-approx-kernel}   
    \end{equation}
    where $L$ is a (matrix-valued) kernel
    and 
    \[
    \bar{k}_{\mathrm{lin}}(x,y)=\frac{k_{\mathrm{lin}}(x,y)}{\sqrt{k_{\mathrm{lin}}(x,x)} \sqrt{k_{\mathrm{lin}}(y,y)}}
    \]
    is a normalised version of the linear kernel.  
\end{definition}

\noindent The recommended kernel in \Cref{def: recommended kernel} satisfies \Cref{asm: kernel condition 1,asm: kernel condition 2,asm: kernel condition 3} of \Cref{asm: for continuity}, as well as the assumptions that we will need later in \Cref{sec:cv_control}.
The proof of the following result is contained in \Cref{app: sufficient kernel proof}: 

\begin{proposition}[Sufficient conditions for \Cref{asm: kernel condition 1,asm: kernel condition 2,asm: kernel condition 3} in \Cref{asm: for continuity}]
\label{prop:kernel-conds-for-continuity}
If one uses the kernel in \Cref{def: recommended kernel}, where $(x,y) \mapsto \partial_{1, i}\partial_{2, i} L(x, y)$ is continuous; 
and both $\sup_{x \in \R^d} \Verts{L(x, x)}_{\mathrm{op}} < \infty$ and  
$\sup_{x \in \mathbb{R}^d} \| \partial_{1,i} \partial_{2,i} L(x,x) \| < \infty$ for all $i \in \{1,\dots,d\}$, then \Cref{asm: kernel condition 1,asm: kernel condition 2,asm: kernel condition 3} in \Cref{asm: for continuity} are satisfied.
\end{proposition}

\noindent The results in \Cref{fig: different methods,fig: different methods 2} are reproduced using this recommended kernel in \Cref{fig: different methods recommended kernel,fig: different methods 2 recommended kernel} of \Cref{app: detail mfnn,app: detail pcuq}.

\subsection{Existence and Uniqueness of Stationary Points} \label{sec:minimizer}

Under mild regularity, the minimiser $\target$ of \eqref{eq: objective} is a stationary point of $\eobj$; see \Cref{app: defn first var}.
Since we are working with stationary points, we are interested in when the following assumption will hold:

\begin{assumption}
\label{asm:uniq}
The objective $\eobj$ in \eqref{eq: objective} has $\target$ as its unique stationary point.
\end{assumption}

\noindent For generalised Bayesian inference, \Cref{asm:uniq} is closely related to the variational formulation of Bayes theorem \citep[][Theorem 1]{knoblauch2022optimization}, which in effect ensures the generalised posterior is well-defined \citep[see e.g.][Appendix A]{wild2023rigorous}.
For more general loss functions $\mathcal{L}$ we have the following result:

\begin{proposition}[Existence and uniqueness of stationary points]
\label{prop: exist unique}
    Let \Cref{asm: diff loss,asm: q0 support} hold.
    Assume that
    \begin{enumerate}
        \item \emph{(convexity)} $\mathcal{L}$ is convex and lower-bounded; 
        \item \emph{(regularity)} $(Q,x) \mapsto \mathcal{L}'(Q)(x)$ is (weakly) continuous in $Q$ and bounded in $(Q,x)$.  \label{asm: bd L}
    \end{enumerate}
    Then $\eobj$ admits a unique minimiser $P$ and \Cref{asm:uniq} is satisfied.
\end{proposition}

\noindent The proof of \Cref{prop: exist unique} is contained in \Cref{app: exist unique stationary}.
\Cref{asm: bd L} of \Cref{prop: exist unique} appears to rule out reasonable loss functions, such as the Bayesian loss $\cL(Q) = \int u(x) \; \dd Q(x)$ whenever the negative log-likelihood $u(x)$ is unbounded.
Our results do in fact extend to loss functions defined only on a subset of $\finitemomentspace{}$, but the discussion is technical and deferred to \Cref{app: relax def gd}.
Here we note a simple workaround in the Bayesian case; let $\tilde{Q}_0$ be the distribution with $\dd \tilde{Q}_0 / \dd Q_0 \propto \exp(-u)$ and then consider instead $\tilde{\eobj}(\cdot) \coloneqq \KLD( \cdot ||\tilde{Q}_0)$, for which $\tilde{\eobj}$ has the same stationary points as $\eobj$ and \Cref{asm: bd L} is trivially satisfied.

\subsection{Convergence Control}\label{sec:cv_control}

The most important property from the perspective of developing computational methodology to approximate the minimiser $P$ of \eqref{eq: objective} is $\polyorder$-\emph{convergence control}, meaning that $\KGD_K(Q_n) \rightarrow 0$ implies $Q_n \toL{\polyorder} \target$.
Indeed, this property is essential to justify the development of algorithms that aim to minimise \ac{kgd}, as presented in \Cref{sec: extensible,sec: variational}.
Convergence control boils down to uniform integrability control, as a consequence of the continuity we already established in \Cref{sec:continuity}.
Proofs for this section build on the recent work of \citet{KanBarGreMac2025}.

\begin{definition}[Generalised dissipativity]
\label{def: gen diss}
    A vector field $v: \mathbb{R}^d \to \mathbb{R}^d$ is said to satisfy \emph{generalised dissipativity}, if there exist 
    some $r_1 >0$, $r_2 \geq 0$ and 
    $\diss > 1/2$ such that 
    \begin{align}
        -v(x) \cdot x \geq r_1 \|x\|^{2\diss} - r_2 \label{eq: dissipativity}
    \end{align}
    for all $x \in \R^d$. 
\end{definition}

\noindent Generalised dissipativity\footnote{Similar but non-identical definitions appear in \citet[][Assumption 1]{lytras2024tamed}, where the term \emph{weak dissipativity} was used, in \citet[][Theorem 11]{barp2024targeted} and in \citet[][Assumption 1]{KanBarGreMac2025}.} 
implies that $v$ points inwards outside some ball, and grows at least at the rate $\Verts{x}^{2\diss - 1}$.
For $v = \nabla \log q$, where $q$ is the density of a distribution $Q$, 
the case $\gamma = 1$ corresponds to $Q$ having a Gaussian tail, and the condition is satisfied of course by Gaussians but also by some Gaussian mixtures, e.g. with common isotropic variance \citep[Lemma B.1]{cordero2025non}. 
Generalised dissipativity with $\gamma < 1$ allows $Q$ to have tails 
arbitrarily close to exponential, 
while $\diss > 1$ would enforce a sub-Gaussian tail. 

\begin{assumption}\label{asm:enforce_tightness}
  \begin{enumerate}[label=(\roman*)]
      \item $\nabla \log q_0$ satisfies generalised dissipativity for some 
      $\diss > 1/2$, 
      and  
     $\beta = 2\diss-1$
      in \Cref{asm: for continuity}. \label{item:asm-concontrol-prior}
      \item The loss $\cL$ satisfies the following: $\sup_{Q \in \finitemomentspace{}, x\in\mathbb{R}^d} \| \vargrad  \cL(Q)(x) \| < \infty$
          and $\cL'(Q)(x)  = o(\Verts{x}^{2\diss})$ for each $Q \in \finitemomentspace{}$ absolutely continuous with respect to $Q_0$. 
      \label{item:asm-convcontrol-growth}
      \item The kernel $K$ is given in the form of \Cref{def: recommended kernel} with $L(x,y) = \ell(x, y)\idmat$, where $\ell(x,y) = \phi(x-y)$ for $\phi \in C^2(\mathbb{R}^d)$ with a spectral density bounded away from zero on compact sets.
  \end{enumerate}
\end{assumption}

\noindent 
A consequence of \Cref{item:asm-concontrol-prior,item:asm-convcontrol-growth} of \Cref{asm:enforce_tightness} is that $\target \in \finitemomentspace{\polyorder}$, as explained in \Cref{sec: diss implies moments}, allowing us to discuss $\polyorder$-convergence.

\begin{theorem}[\ac{kgd} controls $\polyorder$-convergence]
\label{thm: convergence control}
    Fix $\polyorder \in [0, \infty)$. 
    Under \Cref{asm: diff loss,asm: q0 support,asm:contdiff,asm:uniq,asm: for continuity,asm:enforce_tightness}, 
    $\KGD_K(Q_n) \rightarrow 0$ implies $Q_n \toL{\polyorder} \target \in \finitemomentspace{\polyorder}$
    for any sequence $(Q_n)_{n\in \mathbb{N}}\subset \finitemomentspace{}$. 
\end{theorem}

\noindent The proof is contained in \Cref{app: convergence control}.
Since \Cref{prop:continuity} and \Cref{thm: convergence control} have common assumptions, we have shown that \ac{kgd} exactly characterises $\polyorder$-convergence:

\begin{coroll}[\ac{kgd} characterises $\polyorder$-convergence]
\label{cor: metrise}
    Under \Cref{asm: diff loss,asm: q0 support,asm:contdiff,asm:uniq,asm: for continuity,asm:enforce_tightness}, we have $Q_\seqidx \toL{\polyorder} \target$ if and only if $\KGD_K(Q_\seqidx) \to 0$. 
\end{coroll}

\noindent That is, not only does minimisation of \ac{kgd} provide a valid way to approximate the target $P$, but \emph{any} consistent algorithm for approximating $\target$ must necessarily be performing (asymptotic) minimisation of the \ac{kgd}, due to \Cref{cor: metrise}.
This perhaps surprising result 
provides strong justification for using \ac{kgd} as a measure of suboptimality for entropy-regularised variational objectives.

\section{Discussion}
\label{sec: discuss}

This paper introduced the first computable measure of sub-optimality for entropy-regularised variational objectives, illustrating some of the novel functionalities and opportunities for methodological development that are now unlocked, and presenting detailed theoretical support.

\paragraph*{Related Work}

Numerical methods for approximating the variational target $P$ are mainly based on \ac{mfld} \citep[e.g.][and references therein]{del2013mean,suzuki2023convergence}. 
Alternative discretisations of the Wasserstein gradient flow of the objective could be considered, e.g. proximal gradient schemes~\citep{lascu2024linear} based on implicit time discretisations steps, or explicit discretisations \citep[e.g.][]{kook2024sampling}. 
A common theme is that these methods all incur a bias, due to (time and space) discretisation or parametric approximation, which is not easily quantified.
\ac{kgd} therefore provides a practical tool through which such methods can be compared and tuned.

Gradients of the \ac{kld} have been considered from several different perspectives:
A connection between score-matching \citep{hyvarinen2005estimation} and the time derivative of the \ac{kld} between Brownian motions was described in \citet{lyu2009interpretation}.
In a similar spirit, time derivative of the \ac{kld} along a Wasserstein gradient flow was proposed as a training objective in \citet{wang2020wasserstein} where it was called \emph{Wasserstein minimum velocity learning}.
The idea of measuring the Euclidean gradient of \ac{kld} with respect to parameters appearing in one of the arguments has been explored in machine learning, and was termed \emph{GradNorm} in \citet{huang2021importance}.
However, these papers do not consider general entropy-regularised objective functions as in \eqref{eq: objective}.

\paragraph*{Opportunities and Limitations}

Though there are limitations to \ac{gd} and \ac{kgd}, these mostly present as concrete opportunities for subsequent work:

From a theoretical perspective, our results were asymptotic and explicit bounds (revealing the dimension dependence) remain to be investigated.
Further, the finite sample properties of \ac{kgd} inherit the well-known limitations of \ac{ksd} (which is recovered in the case of linear $\cL$), such as insensitivity to differences between high-dimensional distributions and to differences in the weight afforded to distant mixture components \citep{wenliang2020blindness}.
The first of these limitations could be addressed through the development of a sliced version of \ac{kgd}, mirroring sliced \ac{ksd} \citep{gong2021sliced}, while the second limitation appears more fundamental.
An extension to distributions defined on spaces other than $\mathbb{R}^d$ should also be possible, extending discrete \citep{yang2018goodness,hodgkinson2020reproducing,shi2022gradient,matsubara2024generalized} and manifold \citep{liu2018riemannian,xu2021interpretable,barp2022riemann,qu2024kernel} versions of \ac{ksd}. 

\Cref{sec: new algs} illustrated the potential of \ac{kgd} for methodological development.
Building on this, one could develop the supporting theory for extensible sampling (\Cref{sec: extensible}), and parametric variational approximation (\Cref{sec: variational}), generalising existing results for \ac{ksd} \citep{chen2018stein,chen2019stein,fisher2021measure,chopin2021fast,benard2023kernel,koppel2024online,kirk2025low}. 
Though we presented convergence theory for \ac{vgd}, one could go further and attempt to mirror recent developments for \ac{svgd}, such as implementing regularised and higher-order numerical schemes \citep{detommaso2018stein,he2024regularized,stein2025towards}, extension to graphical models \citep{zhuo2018message}, low-rank approximation \citep{liu2022grassmann}, and a gradient-free implementation \citep{han2018stein}.
Going further, one can conceive of further classes of algorithm based on \ac{kgd}, for example based on re-weighting samples obtained using \ac{mfld} to eliminate the inherent bias of \ac{mfld}, which would represent a generalisation of the Stein importance sampling method \citep{hodgkinson2020reproducing,riabiz2022optimal,wang2023stein,anastasiou2023stein}.

To mitigate the $\Omega(d n^2)$ computational cost of evaluating \ac{kgd}, it would be interesting to extend some of the recent computational developments for \ac{ksd}, such as random feature approximations of the kernel \citep{huggins2018random}, stochastic approximation of the loss function \citep{gorham2020stochastic}, Nystr\"{o}m approximation of kernel matrices \citep{kalinke2025nystrom}, thinning methods \citep{dwivedi2024kernel}, and gradient-free implementations \citep{fisher2023gradient}.

Our contribution generalises elements of Stein's method beyond their traditional setting, in the sense that our analogue of the Stein operator, $\mathcal{T}_Q$, is now $Q$-dependent.
Undertaking a systematic generalisation of the powerful machinery of Stein’s method to nonlinear (mean-field / McKean–Vlasov) Markov processes would be an interesting direction for future work. 
In the opposite direction, recent advances in Stein discrepancies could potentially be carried over to the nonlinear setting; for example, the recent work of \citet{wynne2025fourier} extends Stein discrepancies to targets defined on separable Hilbert spaces; developing this generalisation in the nonlinear context would be of interest.

Finally, while our focus was on computation for post-Bayesian statistical methods, entropic regularisation is also widely used in areas such as optimal transport \citep{cuturi2013sinkhorn}, reinforcement learning \citep{neu2017unified}, and physical simulation \citep{shen2023entropy,carrillo2024relative}; the potential applications of \ac{kgd} to these areas remain to be explored.

\paragraph*{Acknowledgments}
HK, ZS and CJO were supported by EPSRC (EP/W019590/1). 
CJO was supported by The Alan Turing Institute and a Philip Leverhulme Prize (PLP-2023-004). 
CC and AK acknowledge the support of the Agence Nationale de la Recherche, through the PEPR PDE-AI project (ANR-23-PEIA-0004). 
The authors wish to thank the Editors and Reviewers for their constructive feedback on the manuscript.
The authors also wish to thank Sayan Banerjee, Alessandro Barp, Fran\c{c}ois-Xavier Briol, Zonghao Chen, Jeremias Knoblauch, Lester Mackey, Nikolas N\"usken, Qiang Liu, Marina Riabiz, Antonin Schrab and Łukasz Szpruch for discussion during the preparation of this manuscript.

\bibliographystyle{abbrvnat}
\bibliography{bibliography}

\begin{thebibliography}{110}
\providecommand{\natexlab}[1]{#1}
\providecommand{\url}[1]{\texttt{#1}}
\expandafter\ifx\csname urlstyle\endcsname\relax
  \providecommand{\doi}[1]{doi: #1}\else
  \providecommand{\doi}{doi: \begingroup \urlstyle{rm}\Url}\fi

\bibitem[Albeverio et~al.(1996)Albeverio, Kondratiev, and R{\"o}ckner]{albeverio1996differential}
S.~Albeverio, Y.~G. Kondratiev, and M.~R{\"o}ckner.
\newblock Differential geometry of {P}oisson spaces.
\newblock \emph{Comptes Rendus de l'Acad{\'e}mie des Sciences. S{\'e}rie 1, Math{\'e}matique}, 323\penalty0 (10):\penalty0 1129--1134, 1996.

\bibitem[Anastasiou et~al.(2023)Anastasiou, Barp, Briol, Ebner, Gaunt, Ghaderinezhad, Gorham, Gretton, Ley, Liu, Mackey, Oates, Reinert, and Swan]{anastasiou2023stein}
A.~Anastasiou, A.~Barp, F.-X. Briol, B.~Ebner, R.~E. Gaunt, F.~Ghaderinezhad, J.~Gorham, A.~Gretton, C.~Ley, Q.~Liu, L.~Mackey, C.~J. Oates, G.~Reinert, and Y.~Swan.
\newblock Stein’s method meets computational statistics: A review of some recent developments.
\newblock \emph{Statistical Science}, 38\penalty0 (1):\penalty0 120--139, 2023.

\bibitem[Banerjee et~al.(2025)Banerjee, Balasubramanian, and Ghosal]{banerjee2025improved}
S.~Banerjee, K.~Balasubramanian, and P.~Ghosal.
\newblock Improved finite-particle convergence rates for {S}tein variational gradient descent.
\newblock In \emph{The Thirteenth International Conference on Learning Representations}, 2025.

\bibitem[Barbour(1990)]{barbour1990stein}
A.~D. Barbour.
\newblock Stein's method for diffusion approximations.
\newblock \emph{Probability Theory and Related Fields}, 84\penalty0 (3):\penalty0 297--322, 1990.

\bibitem[Barp et~al.(2022)Barp, Oates, Porcu, and Girolami]{barp2022riemann}
A.~Barp, C.~J. Oates, E.~Porcu, and M.~Girolami.
\newblock A {R}iemann--{S}tein kernel method.
\newblock \emph{Bernoulli}, 28\penalty0 (4):\penalty0 2181--2208, 2022.

\bibitem[Barp et~al.(2024)Barp, Simon-Gabriel, Girolami, and Mackey]{barp2024targeted}
A.~Barp, C.-J. Simon-Gabriel, M.~Girolami, and L.~Mackey.
\newblock Targeted separation and convergence with kernel discrepancies.
\newblock \emph{Journal of Machine Learning Research}, 25\penalty0 (378):\penalty0 1--50, 2024.

\bibitem[Bauer(2001)]{Bauer2001}
H.~Bauer.
\newblock \emph{Measure and Integration Theory}.
\newblock De Gruyter, 2001.
\newblock ISBN 9783110866209.

\bibitem[B{\'e}nard et~al.(2023)B{\'e}nard, Staber, and Da~Veiga]{benard2023kernel}
C.~B{\'e}nard, B.~Staber, and S.~Da~Veiga.
\newblock Kernel {S}tein discrepancy thinning: {A} theoretical perspective of pathologies and a practical fix with regularization.
\newblock \emph{Advances in Neural Information Processing Systems}, 36:\penalty0 49281--49311, 2023.

\bibitem[Berg et~al.(1984)Berg, Christensen, and Ressel]{Berg_1984}
C.~Berg, J.~P.~R. Christensen, and P.~Ressel.
\newblock \emph{Harmonic Analysis on Semigroups}.
\newblock Springer New York, 1984.
\newblock \doi{10.1007/978-1-4612-1128-0}.

\bibitem[Bogachev(2007)]{Bogachev_2007}
V.~I. Bogachev.
\newblock \emph{Measure Theory}.
\newblock Springer Berlin Heidelberg, 2007.
\newblock ISBN 9783540345145.

\bibitem[Brooks et~al.(2011)Brooks, Gelman, Jones, and Meng]{brooks2011handbook}
S.~Brooks, A.~Gelman, G.~Jones, and X.-L. Meng.
\newblock \emph{Handbook of Markov Chain Monte Carlo}.
\newblock CRC Press, 2011.

\bibitem[Buck(1958)]{Buck_1958}
R.~C. Buck.
\newblock Bounded continuous functions on a locally compact space.
\newblock \emph{Michigan Mathematical Journal}, 5\penalty0 (2), 1958.
\newblock \doi{10.1307/mmj/1028998054}.

\bibitem[Cardaliaguet(2010)]{cardaliaguet2010notes}
P.~Cardaliaguet.
\newblock {Notes on Mean Field Games}.
\newblock Technical report, P.-L. Lions Lectures at College de France, 2010.

\bibitem[Carmeli et~al.(2010)Carmeli, De~Vito, Toigo, and Umanit{\'a}]{carmeli2010vector}
C.~Carmeli, E.~De~Vito, A.~Toigo, and V.~Umanit{\'a}.
\newblock Vector valued reproducing kernel {H}ilbert spaces and universality.
\newblock \emph{Analysis and Applications}, 8\penalty0 (01):\penalty0 19--61, 2010.

\bibitem[Carrillo et~al.(2024)Carrillo, Feng, Guo, Jabin, and Wang]{carrillo2024relative}
J.~A. Carrillo, X.~Feng, S.~Guo, P.-E. Jabin, and Z.~Wang.
\newblock Relative entropy method for particle approximation of the {L}andau equation for {M}axwellian molecules.
\newblock \emph{arXiv preprint arXiv:2408.15035}, 2024.

\bibitem[Chen et~al.(2025)Chen, Nott, and Tan]{chen2025weighted}
A.~Chen, D.~J. Nott, and L.~S. Tan.
\newblock Weighted {F}isher divergence for high-dimensional {G}aussian variational inference.
\newblock \emph{arXiv preprint arXiv:2503.04246}, 2025.

\bibitem[Chen et~al.(2022)Chen, Ren, and Wang]{chen2022uniform}
F.~Chen, Z.~Ren, and S.~Wang.
\newblock Uniform-in-time propagation of chaos for mean field {L}angevin dynamics.
\newblock \emph{arXiv preprint arXiv:2212.03050}, 2022.

\bibitem[Chen et~al.(2018)Chen, Mackey, Gorham, Briol, and Oates]{chen2018stein}
W.~Y. Chen, L.~Mackey, J.~Gorham, F.-X. Briol, and C.~J. Oates.
\newblock Stein points.
\newblock In \emph{International Conference on Machine Learning}, pages 844--853. PMLR, 2018.

\bibitem[Chen et~al.(2019)Chen, Barp, Briol, Gorham, Girolami, Mackey, and Oates]{chen2019stein}
W.~Y. Chen, A.~Barp, F.-X. Briol, J.~Gorham, M.~Girolami, L.~Mackey, and C.~J. Oates.
\newblock Stein point {M}arkov chain {M}onte {C}arlo.
\newblock In \emph{International Conference on Machine Learning}, pages 1011--1021. PMLR, 2019.

\bibitem[Ch{\'e}rief-Abdellatif and Alquier(2020)]{cherief2020mmd}
B.-E. Ch{\'e}rief-Abdellatif and P.~Alquier.
\newblock Mmd-bayes: Robust bayesian estimation via maximum mean discrepancy.
\newblock In \emph{Symposium on Advances in Approximate Bayesian Inference}, pages 1--21. PMLR, 2020.

\bibitem[Chizat(2022)]{chizat2022mean}
L.~Chizat.
\newblock Mean-field {L}angevin dynamics: {E}xponential convergence and annealing.
\newblock \emph{Transactions on Machine Learning Research}, 2022.

\bibitem[Chopin and Ducrocq(2021)]{chopin2021fast}
N.~Chopin and G.~Ducrocq.
\newblock Fast compression of {MCMC} output.
\newblock \emph{Entropy}, 23\penalty0 (8):\penalty0 1017, 2021.

\bibitem[Chwialkowski et~al.(2016)Chwialkowski, Strathmann, and Gretton]{chwialkowski2016kernel}
K.~Chwialkowski, H.~Strathmann, and A.~Gretton.
\newblock A kernel test of goodness of fit.
\newblock In \emph{International Conference on Machine Learning}, pages 2606--2615. PMLR, 2016.

\bibitem[Conway(1965)]{Conway1965}
J.~Conway.
\newblock \emph{The Strict Topology and Compactness in the Space of Measures.}
\newblock PhD thesis, Louisiana State University Libraries, 1965.

\bibitem[Cordero-Encinar et~al.(2025)Cordero-Encinar, Akyildiz, and Duncan]{cordero2025non}
P.~Cordero-Encinar, O.~D. Akyildiz, and A.~B. Duncan.
\newblock Non-asymptotic analysis of diffusion annealed langevin monte carlo for generative modelling.
\newblock \emph{arXiv preprint arXiv:2502.09306}, 2025.

\bibitem[Cuturi(2013)]{cuturi2013sinkhorn}
M.~Cuturi.
\newblock Sinkhorn distances: Lightspeed computation of optimal transport.
\newblock \emph{Advances in Neural Information Processing Systems}, 26, 2013.

\bibitem[Del~Moral(2013)]{del2013mean}
P.~Del~Moral.
\newblock Mean field simulation for {M}onte {C}arlo integration.
\newblock \emph{Monographs on Statistics and Applied Probability}, 126\penalty0 (26):\penalty0 6, 2013.

\bibitem[Detommaso et~al.(2018)Detommaso, Cui, Marzouk, Spantini, and Scheichl]{detommaso2018stein}
G.~Detommaso, T.~Cui, Y.~Marzouk, A.~Spantini, and R.~Scheichl.
\newblock A {S}tein variational {N}ewton method.
\newblock \emph{Advances in Neural Information Processing Systems}, 31, 2018.

\bibitem[Dinh et~al.(2016)Dinh, Sohl-Dickstein, and Bengio]{dinh2016density}
L.~Dinh, J.~Sohl-Dickstein, and S.~Bengio.
\newblock Density estimation using real nvp.
\newblock \emph{arXiv preprint arXiv:1605.08803}, 2016.

\bibitem[Dudley(1999)]{Dudley_1999}
R.~M. Dudley.
\newblock \emph{Uniform Central Limit Theorems}.
\newblock Cambridge University Press, 1999.
\newblock \doi{10.1017/cbo9780511665622}.

\bibitem[Dudley(2002)]{Dudley2002}
R.~M. Dudley.
\newblock \emph{Real Analysis and Probability}.
\newblock Cambridge University Press, 2002.
\newblock \doi{10.1017/cbo9780511755347}.

\bibitem[Dupuis and Ellis(2011)]{dupuis2011weak}
P.~Dupuis and R.~S. Ellis.
\newblock \emph{A Weak Convergence Approach to the Theory of Large Deviations}.
\newblock John Wiley \& Sons, 2011.

\bibitem[Dwivedi and Mackey(2024)]{dwivedi2024kernel}
R.~Dwivedi and L.~Mackey.
\newblock Kernel thinning.
\newblock \emph{Journal of Machine Learning Research}, 25\penalty0 (152):\penalty0 1--77, 2024.

\bibitem[Fisher and Oates(2023)]{fisher2023gradient}
M.~Fisher and C.~J. Oates.
\newblock Gradient-free kernel {S}tein discrepancy.
\newblock \emph{Advances in Neural Information Processing Systems}, 36:\penalty0 23855--23885, 2023.

\bibitem[Fisher et~al.(2021)Fisher, Nolan, Graham, Prangle, and Oates]{fisher2021measure}
M.~Fisher, T.~Nolan, M.~Graham, D.~Prangle, and C.~J. Oates.
\newblock Measure transport with kernel {S}tein discrepancy.
\newblock In \emph{International Conference on Artificial Intelligence and Statistics}, pages 1054--1062. PMLR, 2021.

\bibitem[Fox and Roberts(2012)]{fox2012tutorial}
C.~W. Fox and S.~J. Roberts.
\newblock A tutorial on variational {B}ayesian inference.
\newblock \emph{Artificial Intelligence Review}, 38:\penalty0 85--95, 2012.

\bibitem[Gneiting and Raftery(2007)]{gneiting2007strictly}
T.~Gneiting and A.~E. Raftery.
\newblock Strictly proper scoring rules, prediction, and estimation.
\newblock \emph{Journal of the American statistical Association}, 102\penalty0 (477):\penalty0 359--378, 2007.

\bibitem[Gong et~al.(2021)Gong, Li, and Hern{\'a}ndez-Lobato]{gong2021sliced}
W.~Gong, Y.~Li, and J.~M. Hern{\'a}ndez-Lobato.
\newblock Sliced kernelized {S}tein discrepancy.
\newblock In \emph{International Conference on Learning Representations}, 2021.

\bibitem[Goodfellow et~al.(2016)Goodfellow, Bengio, Courville, and Bengio]{goodfellow2016deep}
I.~Goodfellow, Y.~Bengio, A.~Courville, and Y.~Bengio.
\newblock \emph{Deep Learning}.
\newblock MIT Press Cambridge, 2016.

\bibitem[Gorham and Mackey(2015)]{gorham2015measuring}
J.~Gorham and L.~Mackey.
\newblock Measuring sample quality with {S}tein's method.
\newblock \emph{Advances in Neural Information Processing Systems}, 28, 2015.

\bibitem[Gorham and Mackey(2017)]{gorham2017measuring}
J.~Gorham and L.~Mackey.
\newblock Measuring sample quality with kernels.
\newblock In \emph{International Conference on Machine Learning}, pages 1292--1301. PMLR, 2017.

\bibitem[Gorham et~al.(2020)Gorham, Raj, and Mackey]{gorham2020stochastic}
J.~Gorham, A.~Raj, and L.~Mackey.
\newblock Stochastic {S}tein discrepancies.
\newblock \emph{Advances in Neural Information Processing Systems}, 33:\penalty0 17931--17942, 2020.

\bibitem[Han and Liu(2018)]{han2018stein}
J.~Han and Q.~Liu.
\newblock Stein variational gradient descent without gradient.
\newblock In \emph{International Conference on Machine Learning}, pages 1900--1908. PMLR, 2018.

\bibitem[Hartman(2002)]{hartman2002ordinary}
P.~Hartman.
\newblock \emph{Ordinary Differential Equations}.
\newblock SIAM, 2002.

\bibitem[He et~al.(2024)He, Balasubramanian, Sriperumbudur, and Lu]{he2024regularized}
Y.~He, K.~Balasubramanian, B.~K. Sriperumbudur, and J.~Lu.
\newblock Regularized {S}tein variational gradient flow.
\newblock \emph{Foundations of Computational Mathematics}, pages 1--59, 2024.

\bibitem[Hodgkinson et~al.(2020)Hodgkinson, Salomone, and Roosta]{hodgkinson2020reproducing}
L.~Hodgkinson, R.~Salomone, and F.~Roosta.
\newblock The reproducing {S}tein kernel approach for post-hoc corrected sampling.
\newblock \emph{arXiv preprint arXiv:2001.09266}, 2020.

\bibitem[Hoeffding(1948)]{hoeffding1992class}
W.~Hoeffding.
\newblock A class of statistics with asymptotically normal distribution.
\newblock \emph{The Annals of Mathematical Statistics}, 19\penalty0 (3):\penalty0 293--325, 1948.

\bibitem[Hu et~al.(2021)Hu, Ren, {\v{S}}i{\v{s}}ka, and Szpruch]{hu2021mean}
K.~Hu, Z.~Ren, D.~{\v{S}}i{\v{s}}ka, and {\L}.~Szpruch.
\newblock Mean-field {Langevin} dynamics and energy landscape of neural networks.
\newblock In \emph{Annales de l'Institut Henri Poincare (B) Probabilites et statistiques}, volume~57, pages 2043--2065, 2021.

\bibitem[Huang et~al.(2021)Huang, Geng, and Li]{huang2021importance}
R.~Huang, A.~Geng, and Y.~Li.
\newblock On the importance of gradients for detecting distributional shifts in the wild.
\newblock \emph{Advances in Neural Information Processing Systems}, 34:\penalty0 677--689, 2021.

\bibitem[Huggins and Mackey(2018)]{huggins2018random}
J.~Huggins and L.~Mackey.
\newblock Random feature {S}tein discrepancies.
\newblock \emph{Advances in Neural Information Processing Systems}, 31, 2018.

\bibitem[Hyv{\"a}rinen(2005)]{hyvarinen2005estimation}
A.~Hyv{\"a}rinen.
\newblock Estimation of non-normalized statistical models by score matching.
\newblock \emph{Journal of Machine Learning Research}, 6\penalty0 (4), 2005.

\bibitem[Jankowiak et~al.(2020{\natexlab{a}})Jankowiak, Pleiss, and Gardner]{jankowiak2020deep}
M.~Jankowiak, G.~Pleiss, and J.~Gardner.
\newblock Deep sigma point processes.
\newblock In \emph{Conference on Uncertainty in Artificial Intelligence}, pages 789--798. PMLR, 2020{\natexlab{a}}.

\bibitem[Jankowiak et~al.(2020{\natexlab{b}})Jankowiak, Pleiss, and Gardner]{jankowiak2020parametric}
M.~Jankowiak, G.~Pleiss, and J.~Gardner.
\newblock Parametric {G}aussian process regressors.
\newblock In \emph{International Conference on Machine Learning}, pages 4702--4712. PMLR, 2020{\natexlab{b}}.

\bibitem[Kalinke et~al.(2025)Kalinke, Szab{\'o}, and Sriperumbudur]{kalinke2025nystrom}
F.~Kalinke, Z.~Szab{\'o}, and B.~Sriperumbudur.
\newblock Nystr{\"o}m kernel {S}tein discrepancy.
\newblock In \emph{The 28th International Conference on Artificial Intelligence and Statistics}, 2025.

\bibitem[Kanagawa et~al.(2025)Kanagawa, Barp, Gretton, and Mackey]{KanBarGreMac2025}
H.~Kanagawa, A.~Barp, A.~Gretton, and L.~Mackey.
\newblock Controlling moments with kernel {Stein} discrepancies.
\newblock \emph{{Annals of Applied Probability}}, 35\penalty0 (6):\penalty0 3818--3843, 2025.
\newblock ISSN 1050-5164.
\newblock \doi{10.1214/25-AAP2206}.

\bibitem[Kirk et~al.(2025)Kirk, Rusch, Zech, and Rus]{kirk2025low}
N.~Kirk, T.~K. Rusch, J.~Zech, and D.~Rus.
\newblock Low {S}tein discrepancy via message-passing {M}onte {C}arlo.
\newblock In \emph{Frontiers in Probabilistic Inference: Learning meets Sampling}, 2025.

\bibitem[Knoblauch et~al.(2022)Knoblauch, Jewson, and Damoulas]{knoblauch2022optimization}
J.~Knoblauch, J.~Jewson, and T.~Damoulas.
\newblock An optimization-centric view on {B}ayes' rule: {R}eviewing and generalizing variational inference.
\newblock \emph{Journal of Machine Learning Research}, 23\penalty0 (132):\penalty0 1--109, 2022.

\bibitem[Kook et~al.(2024)Kook, Zhang, Chewi, Erdogdu, and Li]{kook2024sampling}
Y.~Kook, M.~S. Zhang, S.~Chewi, M.~A. Erdogdu, and M.~B. Li.
\newblock Sampling from the mean-field stationary distribution.
\newblock In \emph{The 37th Annual Conference on Learning Theory}, pages 3099--3136. PMLR, 2024.

\bibitem[Koppel et~al.(2024)Koppel, Eappen, Bhatt, Hawkins, and Ganesh]{koppel2024online}
A.~Koppel, J.~Eappen, S.~Bhatt, C.~Hawkins, and S.~Ganesh.
\newblock Online {MCMC} thinning with kernelized {S}tein discrepancy.
\newblock \emph{SIAM Journal on Mathematics of Data Science}, 6\penalty0 (1):\penalty0 51--75, 2024.

\bibitem[Korba et~al.(2021)Korba, Aubin-Frankowski, Majewski, and Ablin]{korba2021kernel}
A.~Korba, P.-C. Aubin-Frankowski, S.~Majewski, and P.~Ablin.
\newblock Kernel {Stein} discrepancy descent.
\newblock In \emph{International Conference on Machine Learning}, pages 5719--5730. PMLR, 2021.

\bibitem[Lai et~al.(2026)Lai, Linero, and Yao]{lai2024predictive}
J.~Lai, A.~R. Linero, and Y.~Yao.
\newblock Predictive variational inference: Learn the predictively optimal posterior distribution.
\newblock In \emph{International Conference on Machine Learning}. PMLR, 2026.

\bibitem[Lanzetti et~al.(2025)Lanzetti, Bolognani, and D{\"o}rfler]{lanzetti2025first}
N.~Lanzetti, S.~Bolognani, and F.~D{\"o}rfler.
\newblock First-order conditions for optimization in the {W}asserstein space.
\newblock \emph{SIAM Journal on Mathematics of Data Science}, 7\penalty0 (1):\penalty0 274--300, 2025.

\bibitem[Lascu et~al.(2024)Lascu, Majka, {\v{S}}i{\v{s}}ka, and Szpruch]{lascu2024linear}
R.-A. Lascu, M.~B. Majka, D.~{\v{S}}i{\v{s}}ka, and {\L}.~Szpruch.
\newblock Linear convergence of proximal descent schemes on the wasserstein space.
\newblock \emph{arXiv preprint arXiv:2411.15067}, 2024.

\bibitem[Liu and Zhu(2018)]{liu2018riemannian}
C.~Liu and J.~Zhu.
\newblock Riemannian {S}tein variational gradient descent for {B}ayesian inference.
\newblock In \emph{Proceedings of the AAAI Conference on Artificial Intelligence}, number~1, 2018.

\bibitem[Liu and Wang(2016)]{liu2016stein}
Q.~Liu and D.~Wang.
\newblock Stein variational gradient descent: A general purpose {B}ayesian inference algorithm.
\newblock \emph{Advances in Neural Information Processing Systems}, 29, 2016.

\bibitem[Liu et~al.(2016)Liu, Lee, and Jordan]{liu2016kernelized}
Q.~Liu, J.~Lee, and M.~Jordan.
\newblock A kernelized {S}tein discrepancy for goodness-of-fit tests.
\newblock In \emph{International Conference on Machine Learning}, pages 276--284. PMLR, 2016.

\bibitem[Liu et~al.(2026)Liu, Mackey, and Oates]{liu2025learning}
Q.~Liu, L.~Mackey, and C.~J. Oates.
\newblock Probabilistic inference and learning with {S}tein's method.
\newblock \emph{arXiv preprint arXiv:2603.07467}, 2026.

\bibitem[Liu et~al.(2022)Liu, Zhu, Ton, Wynne, and Duncan]{liu2022grassmann}
X.~Liu, H.~Zhu, J.-F. Ton, G.~Wynne, and A.~Duncan.
\newblock Grassmann {S}tein variational gradient descent.
\newblock In \emph{International Conference on Artificial Intelligence and Statistics}, pages 2002--2021. PMLR, 2022.

\bibitem[Lytras and Mertikopoulos(2024)]{lytras2024tamed}
I.~Lytras and P.~Mertikopoulos.
\newblock Tamed {L}angevin sampling under weaker conditions.
\newblock \emph{arXiv preprint arXiv:2405.17693}, 2024.

\bibitem[Lyu(2009)]{lyu2009interpretation}
S.~Lyu.
\newblock Interpretation and generalization of score matching.
\newblock In \emph{Proceedings of the Twenty-Fifth Conference on Uncertainty in Artificial Intelligence}, pages 359--366, 2009.

\bibitem[Masegosa(2020)]{masegosa2020learning}
A.~Masegosa.
\newblock Learning under model misspecification: {A}pplications to variational and ensemble methods.
\newblock \emph{Advances in Neural Information Processing Systems}, 33:\penalty0 5479--5491, 2020.

\bibitem[Matsubara et~al.(2024)Matsubara, Knoblauch, Briol, and Oates]{matsubara2024generalized}
T.~Matsubara, J.~Knoblauch, F.-X. Briol, and C.~J. Oates.
\newblock Generalized {B}ayesian inference for discrete intractable likelihood.
\newblock \emph{Journal of the American Statistical Association}, 119\penalty0 (547):\penalty0 2345--2355, 2024.

\bibitem[McLatchie et~al.(2025)McLatchie, Cherief-Abdellatif, Frazier, and Knoblauch]{mclatchie2025predictively}
Y.~McLatchie, B.-E. Cherief-Abdellatif, D.~T. Frazier, and J.~Knoblauch.
\newblock Predictively oriented posteriors.
\newblock \emph{arXiv preprint arXiv:2510.01915}, 2025.

\bibitem[Morningstar et~al.(2022)Morningstar, Alemi, and Dillon]{morningstar2022pacm}
W.~R. Morningstar, A.~Alemi, and J.~V. Dillon.
\newblock {PACm-Bayes}: {N}arrowing the empirical risk gap in the misspecified {B}ayesian regime.
\newblock In \emph{International Conference on Artificial Intelligence and Statistics}, pages 8270--8298. PMLR, 2022.

\bibitem[Neu et~al.(2017)Neu, Jonsson, and G{\'o}mez]{neu2017unified}
G.~Neu, A.~Jonsson, and V.~G{\'o}mez.
\newblock A unified view of entropy-regularized {M}arkov decision processes.
\newblock \emph{arXiv preprint arXiv:1705.07798}, 2017.

\bibitem[Nitanda and Suzuki(2017)]{nitanda2017stochastic}
A.~Nitanda and T.~Suzuki.
\newblock Stochastic particle gradient descent for infinite ensembles.
\newblock \emph{arXiv preprint arXiv:1712.05438}, 2017.

\bibitem[Nitanda et~al.(2022)Nitanda, Wu, and Suzuki]{nitanda2022convex}
A.~Nitanda, D.~Wu, and T.~Suzuki.
\newblock Convex analysis of the mean field {L}angevin dynamics.
\newblock In \emph{International Conference on Artificial Intelligence and Statistics}, pages 9741--9757. PMLR, 2022.

\bibitem[Nitanda et~al.(2025)Nitanda, Lee, Kai, Sakaguchi, and Suzuki]{nitanda2025propagation}
A.~Nitanda, A.~Lee, D.~T.~X. Kai, M.~Sakaguchi, and T.~Suzuki.
\newblock Propagation of chaos for mean-field {L}angevin dynamics and its application to model ensemble.
\newblock \emph{arXiv preprint arXiv:2502.05784}, 2025.

\bibitem[Oates et~al.(2017)Oates, Girolami, and Chopin]{oates2017control}
C.~J. Oates, M.~Girolami, and N.~Chopin.
\newblock Control functionals for {M}onte {C}arlo integration.
\newblock \emph{Journal of the Royal Statistical Society Series B: Statistical Methodology}, 79\penalty0 (3):\penalty0 695--718, 2017.

\bibitem[Otto(2001)]{Otto01t}
F.~Otto.
\newblock The geometry of dissipative evolution equations: The porous medium equation.
\newblock \emph{Communications in Partial Differential Equations}, 26\penalty0 (1-2):\penalty0 101--174, 2001.

\bibitem[Papamakarios et~al.(2021)Papamakarios, Nalisnick, Rezende, Mohamed, and Lakshminarayanan]{papamakarios2021normalizing}
G.~Papamakarios, E.~Nalisnick, D.~J. Rezende, S.~Mohamed, and B.~Lakshminarayanan.
\newblock Normalizing flows for probabilistic modeling and inference.
\newblock \emph{Journal of Machine Learning Research}, 22\penalty0 (57):\penalty0 1--64, 2021.

\bibitem[Qu et~al.(2024)Qu, Fan, and Vemuri]{qu2024kernel}
X.~Qu, X.~Fan, and B.~C. Vemuri.
\newblock Kernel {S}tein discrepancy on {L}ie groups: Theory and applications.
\newblock \emph{IEEE Transactions on Information Theory}, 70\penalty0 (12), 2024.

\bibitem[Riabiz et~al.(2022)Riabiz, Chen, Cockayne, Swietach, Niederer, Mackey, and Oates]{riabiz2022optimal}
M.~Riabiz, W.~Y. Chen, J.~Cockayne, P.~Swietach, S.~A. Niederer, L.~Mackey, and C.~J. Oates.
\newblock Optimal thinning of {MCMC} output.
\newblock \emph{Journal of the Royal Statistical Society Series B: Statistical Methodology}, 84\penalty0 (4):\penalty0 1059--1081, 2022.

\bibitem[Roberts and Tweedie(1996)]{roberts1996exponential}
G.~O. Roberts and R.~L. Tweedie.
\newblock Exponential convergence of {L}angevin distributions and their discrete approximations.
\newblock \emph{Bernoulli}, pages 341--363, 1996.

\bibitem[Schwartz(1978)]{Schwartz1978}
L.~Schwartz.
\newblock \emph{Théorie des Distributions}.
\newblock Hermann, Paris, 1978.
\newblock ISBN 2705655514.

\bibitem[Serfling(2009)]{serfling2009approximation}
R.~J. Serfling.
\newblock \emph{Approximation Theorems of Mathematical Statistics}.
\newblock John Wiley \& Sons, 2009.

\bibitem[Shen and Wang(2023)]{shen2023entropy}
Z.~Shen and Z.~Wang.
\newblock Entropy-dissipation informed neural network for {McKean--Vlasov} type {PDE}s.
\newblock \emph{Advances in Neural Information Processing Systems}, 36:\penalty0 59227--59238, 2023.

\bibitem[Shen et~al.(2025)Shen, Knoblauch, Power, and Oates]{shen2024prediction}
Z.~Shen, J.~Knoblauch, S.~Power, and C.~J. Oates.
\newblock Prediction-centric uncertainty quantification via {MMD}.
\newblock In \emph{International Conference on Artificial Intelligence and Statistics}, pages 649--657. PMLR, 2025.

\bibitem[Sheth and Khardon(2020)]{sheth2020pseudo}
R.~Sheth and R.~Khardon.
\newblock Pseudo-{B}ayesian learning via direct loss minimization with applications to sparse {G}aussian process models.
\newblock In \emph{Symposium on Advances in Approximate Bayesian Inference}, pages 1--18. PMLR, 2020.

\bibitem[Shi et~al.(2022)Shi, Zhou, Hwang, Titsias, and Mackey]{shi2022gradient}
J.~Shi, Y.~Zhou, J.~Hwang, M.~Titsias, and L.~Mackey.
\newblock Gradient estimation with discrete {S}tein operators.
\newblock \emph{Advances in Neural Information Processing Systems}, 35:\penalty0 25829--25841, 2022.

\bibitem[Simon-Gabriel et~al.(2023)Simon-Gabriel, Barp, Sch{\"{o}}lkopf, and Mackey]{SimonGabriel2020}
C.-J. Simon-Gabriel, A.~Barp, B.~Sch{\"{o}}lkopf, and L.~Mackey.
\newblock Metrizing weak convergence with maximum mean discrepancies.
\newblock \emph{Journal of Machine Learning Research}, 24\penalty0 (1), 2023.

\bibitem[Smola et~al.(2007)Smola, Gretton, Song, and Sch{\"o}lkopf]{smola2007hilbert}
A.~Smola, A.~Gretton, L.~Song, and B.~Sch{\"o}lkopf.
\newblock A {H}ilbert space embedding for distributions.
\newblock In \emph{International Conference on Algorithmic Learning Theory}, pages 13--31. Springer, 2007.

\bibitem[Sriperumbudur et~al.(2011)Sriperumbudur, Fukumizu, and Lanckriet]{sriperumbudur2011universality}
B.~K. Sriperumbudur, K.~Fukumizu, and G.~R. Lanckriet.
\newblock Universality, characteristic kernels and {RKHS} embedding of measures.
\newblock \emph{Journal of Machine Learning Research}, 12\penalty0 (7), 2011.

\bibitem[Stein(1972)]{stein1972bound}
C.~Stein.
\newblock A bound for the error in the normal approximation to the distribution of a sum of dependent random variables.
\newblock In \emph{Proceedings of the 6th Berkeley Symposium on Mathematical Statistics and Probability, Volume 2: Probability theory}, pages 583--603. University of California Press, 1972.

\bibitem[Stein and Li(2025)]{stein2025towards}
V.~Stein and W.~Li.
\newblock Towards understanding accelerated {S}tein variational gradient flow--analysis of generalized bilinear kernels for {G}aussian target distributions.
\newblock \emph{arXiv preprint arXiv:2509.04008}, 2025.

\bibitem[Suzuki et~al.(2023)Suzuki, Wu, and Nitanda]{suzuki2023convergence}
T.~Suzuki, D.~Wu, and A.~Nitanda.
\newblock Convergence of mean-field {L}angevin dynamics: Time-space discretization, stochastic gradient, and variance reduction.
\newblock \emph{Advances in Neural Information Processing Systems}, 36:\penalty0 15545--15577, 2023.

\bibitem[Teymur et~al.(2021)Teymur, Gorham, Riabiz, and Oates]{teymur2021optimal}
O.~Teymur, J.~Gorham, M.~Riabiz, and C.~J. Oates.
\newblock Optimal quantisation of probability measures using maximum mean discrepancy.
\newblock In \emph{International Conference on Artificial Intelligence and Statistics}, pages 1027--1035. PMLR, 2021.

\bibitem[Tr\`{e}ves(1967)]{Treves1967}
F.~Tr\`{e}ves.
\newblock \emph{Topological Vector Spaces, Distributions and Kernels}.
\newblock Academic Press, New York, 1967.
\newblock ISBN 9780080873374.

\bibitem[Wang et~al.(2023)Wang, Chen, Kanagawa, and Oates]{wang2023stein}
C.~Wang, Y.~Chen, H.~Kanagawa, and C.~J. Oates.
\newblock Stein {$\Pi$}-importance sampling.
\newblock \emph{Advances in Neural Information Processing Systems}, 36:\penalty0 71948--71994, 2023.

\bibitem[Wang and Liu(2019)]{wang2019nonlinear}
D.~Wang and Q.~Liu.
\newblock Nonlinear {S}tein variational gradient descent for learning diversified mixture models.
\newblock In \emph{International Conference on Machine Learning}, pages 6576--6585. PMLR, 2019.

\bibitem[Wang et~al.(2019)Wang, Tang, Bajaj, and Liu]{wang2019stein}
D.~Wang, Z.~Tang, C.~Bajaj, and Q.~Liu.
\newblock Stein variational gradient descent with matrix-valued kernels.
\newblock \emph{Advances in Neural Information Processing Systems}, 32, 2019.

\bibitem[Wang et~al.(2020)Wang, Cheng, Yueru, Zhu, and Zhang]{wang2020wasserstein}
Z.~Wang, S.~Cheng, L.~Yueru, J.~Zhu, and B.~Zhang.
\newblock A {W}asserstein minimum velocity approach to learning unnormalized models.
\newblock In \emph{International Conference on Artificial Intelligence and Statistics}, pages 3728--3738. PMLR, 2020.

\bibitem[Wendland(2004)]{wendland2004scattered}
H.~Wendland.
\newblock \emph{Scattered Data Approximation}.
\newblock Cambridge University Press, 2004.

\bibitem[Wenliang and Kanagawa(2020)]{wenliang2020blindness}
L.~K. Wenliang and H.~Kanagawa.
\newblock Blindness of score-based methods to isolated components and mixing proportions.
\newblock \emph{arXiv preprint arXiv:2008.10087}, 2020.

\bibitem[Wibisono(2018)]{wibisono2018sampling}
A.~Wibisono.
\newblock Sampling as optimization in the space of measures: {T}he {L}angevin dynamics as a composite optimization problem.
\newblock In \emph{Conference on Learning Theory}, pages 2093--3027. PMLR, 2018.

\bibitem[Wild et~al.(2023)Wild, Ghalebikesabi, Sejdinovic, and Knoblauch]{wild2023rigorous}
V.~D. Wild, S.~Ghalebikesabi, D.~Sejdinovic, and J.~Knoblauch.
\newblock A rigorous link between deep ensembles and (variational) {B}ayesian methods.
\newblock \emph{Advances in Neural Information Processing Systems}, 36:\penalty0 39782--39811, 2023.

\bibitem[Wynne et~al.(2025)Wynne, Kasprzak, and Duncan]{wynne2025fourier}
G.~Wynne, M.~J. Kasprzak, and A.~B. Duncan.
\newblock A {F}ourier representation of kernel {S}tein discrepancy with application to goodness-of-fit tests for measures on infinite dimensional {H}ilbert spaces.
\newblock \emph{Bernoulli}, 31\penalty0 (2):\penalty0 868--893, 2025.

\bibitem[Xu and Matsuda(2021)]{xu2021interpretable}
W.~Xu and T.~Matsuda.
\newblock Interpretable {S}tein goodness-of-fit tests on {R}iemannian manifold.
\newblock In \emph{International Conference on Machine Learning}, pages 11502--11513. PMLR, 2021.

\bibitem[Yang et~al.(2018)Yang, Liu, Rao, and Neville]{yang2018goodness}
J.~Yang, Q.~Liu, V.~Rao, and J.~Neville.
\newblock Goodness-of-fit testing for discrete distributions via {S}tein discrepancy.
\newblock In \emph{International Conference on Machine Learning}, pages 5561--5570. PMLR, 2018.

\bibitem[Zhuo et~al.(2018)Zhuo, Liu, Shi, Zhu, Chen, and Zhang]{zhuo2018message}
J.~Zhuo, C.~Liu, J.~Shi, J.~Zhu, N.~Chen, and B.~Zhang.
\newblock Message passing {S}tein variational gradient descent.
\newblock In \emph{International Conference on Machine Learning}, pages 6018--6027. PMLR, 2018.

\end{thebibliography}

\newpage
\appendix

\part{} 
\parttoc 

\section{Proofs of Theoretical Results}
\label{app: proofs}

\subsection{Additional Notation}\label{sec:add_notation}

Here we introduce additional notation that will be used in the sequel.

\paragraph*{Operations on $\mathbb{R}^d$}

For a matrix $M \in \mathbb{R}^{d \times d}$, we continue to let $\|M\|$ denote any valid matrix norm, but we distinguish the Frobenius norm as $\Verts{M}_{\mathrm{F}} \coloneqq \sum_{ij}M_{ij}^2$.
For a differentiable function $f : \mathbb{R}^d \rightarrow \mathbb{R}^d$, let $[\nabla f]_{i,j} \coloneqq \partial_j f_i$.
For a differentiable function $f : \mathbb{R}^d \rightarrow \mathbb{R}^{d \times d}$, let $\nabla \cdot f$ denote the column-wise divergence, i.e. $[\nabla \cdot f]_j \coloneqq \sum_{i=1}^d\partial_{i} f_{ij}$.

\paragraph*{Operations on $\mathcal{P}(\mathbb{R}^d)$}

A vector-valued function is said to be an element of $\cL^1(\mathbb{R}^d)$ if each coordinate function is an element of $\cL^1(\mathbb{R}^d)$.
The convention $\int f \dd Q^{\otimes 0} = f$ will be used.
Let $\mathcal{L}^1(\R^m, Q) := \{f : \mathbb{R}^d \rightarrow \mathbb{R}^m : \int \Verts{f} \; \dd Q(x) < \infty  \}$. 
We also use the functional notation $Q(f) \coloneqq \int f\; \dd Q$. 
For $\theta: \R^d \to [1, \infty)$, we denote the set of measures that finitely integrate $\theta$ by $\finitemomentspace{\theta} \coloneqq \{Q \in \finitemomentspace{}: \theta \in \cL^1(Q)\}$. 

\paragraph*{Generalised Score}
It will be convenient to introduce the shorthand $\genscore{P}{Q}(x) \coloneqq \nabla \log \genpdf{Q}(x) = \nabla \log q_0(x) - \vargrad  \mathcal{L}(Q)(x)$.
One can interpret $\genscore{P}{Q}(x)$ as a generalisation of the Stein score of $\target$, i.e. $\nabla \log p(x)$ where $p$ is a density for $\target$, which is recovered in the special case of a linear loss function $\cL$.

\subsection{Relaxing the First Variation Requirement}
\label{app: defn first var}

This appendix is dedicated to a technical discussion of the case where a function $\cF$ is not defined everywhere on $\finitemomentspace{}$, so that in particular the first variation $\cF'$ (in the sense used in \Cref{subsec: notation main}) is not well-defined. 
This provides justification for the informal derivations in \Cref{subsec: gd}, where $\mathcal{F}(Q) = \KLD(Q || Q_0)$, and indicates how \Cref{asm: diff loss}, where $\cF(Q) = \cL(Q)$, can be relaxed.

In this section we adopt the convention that a function $\cF: \cS \subset \finitemomentspace{} \to \R$ can be unambiguously extended to a function $\cF : \finitemomentspace{} \to (-\infty,\infty]$ by defining $\cF(Q) = \infty$ for $Q \notin \cS$,
and accordingly we let $\dom(\cF) \coloneqq \{Q \in \finitemomentspace{}: \cF(Q) < \infty \}$.

\subsubsection{Minimality and Stationarity}
\label{app: min and stat}

The aim of this section is to rigorously relate the minima of $\cF$ to the stationary points of $\cF$.
For this purpose, it suffices to define the first variation as follows:

\begin{definition}[Subderivative]
\label{def:subd-fv}
Let $\cF: \finitemomentspace{} \to (-\infty, \infty]$ be a function. 
A \emph{subderivative} of $\cF$ at $Q \in \dom(\cF) \subseteq \finitemomentspace{}$, is defined as an element $g \in\cL^1(Q)$ such that
\begin{equation}
    \liminf_{t \to 0^+}\frac{ \cF\bigl(Q + t (R - Q)\bigr)-\cF(Q)}{t} \geq \int g \; \dd (R - Q)  \label{eq:first-variation}
\end{equation}
in the extended real $[-\infty, \infty]$, for at least one element $R \in \finitemomentspace{}\setminus\{Q\}$.
Let $\cA_{\cF, \geq}(Q; g)$ denote the set of \emph{admissible directions} $R$ for which \eqref{eq:first-variation} is satisfied, and let $\cA_{\cF}(Q;g ) \subseteq \cA_{\cF, \geq}(Q; g)$ denote the set of $R$ for which the equality 
\begin{equation*}
    \lim_{t \to 0^+}\frac{ \cF\bigl(Q + t (R - Q)\bigr)-\cF(Q)}{t} = \int g \; \dd (R - Q)
\end{equation*}
holds in $\R$. 
We say the subderivative $g$ is \emph{fully admissible} at $Q$ if $\cA_{\cF, \geq}(Q;g ) \supseteq \dom(\cF) \setminus \{Q\}$. 
A subderivative $g$ of $\cF$ 
is \emph{compactly-bounded admissible} at $Q \in \dom(\cF)$ with respect to reference $Q_0 \in \finitemomentspace{}$
if $\cA_{\cF}(Q; g)$ contains $\cR_{b, c}(Q_0) \setminus \{Q\}$, 
where 
\[
\cR_{b, c}(Q_0) 
 \coloneqq \left\{hQ_0: h \geq 0, \int h\; \dd Q_0=1,  \Verts{h}_\infty <\infty,\ \mathrm{supp}(h) \text{ is compact}\right\}. 
\]
\end{definition}

\noindent If a first variation $\cF'(Q)$ exists, then it is also a subderivative of $\cF$ at $Q$.
Important functions which do not have a first variation, such as the \ac{kld}, do have a subderivative (see \Cref{lem:KL-fv}). 
Concrete examples are contained in \Cref{app: examples}.

In general, the subderivative need not be unique due to the restriction of the admissible directions $R$. 
Nonetheless, the following lemma shows that given a reference measure, a compactly-bounded admissible subderivative is uniquely determined by the functional, up to an additive constant, and it can thus be regarded as a generalisation of the first variation defined in the main text. 

\begin{lemma}[Uniqueness of compactly-bounded admissible subderivative]\label{lem:subderivative-unique}
Let $\cF:\finitemomentspace{}\to(-\infty,\infty]$ and 
$Q\in\dom(\cF)$. 
Then, up to an additive constant, a subderivative of $\cF$ at $Q$ that is compactly-bounded admissible with respect to reference $Q_0$ is unique $Q_0$-almost everywhere.
\end{lemma}
\begin{proof}
    If there are two choices $g_{1}$ and $g_{2}$, let $h = g_{1} - g_{2}$ and $c_Q = \int h\; \dd Q$. 
    Then, for any $R \in \cR_{b,c}(Q_0) \setminus \{Q\}$, 
    \begin{align*}
        \int g_1\,\dd(R-Q)
        =\lim_{t\to0^+}\frac{\cF((1-t)Q+tR)-\cF(Q)}{t}
        =\int g_2\,\dd(R-Q).
    \end{align*}
    In particular, $f = h - c_Q$ satisfies $\int_B f \dd R = 0$ for $R \propto 1_B Q_0$  and bounded measurable $B$ with $Q_0(B) > 0$.
    Concretely, we choose $B_m^+ = \{f > 0\} \cap \{ \Verts{x}_2  \leq m\}$ and 
    $B_m^- = \{f < 0\} \cap \{ \Verts{x}_2  \leq m\}$. 
    Then we have $\int_{\Verts{x}_2\leq m} \verts{f} \dd Q_0 = 0$. So we have $f = 0$ $Q_0$-a.e. by monotone convergence. 
\end{proof}

Our definition of stationarity (\Cref{def:stationary_point}) in the main text is based on the following result. 
For $h : \R^d \rightarrow \R$ and $Q \in \finitemomentspace{}$, we use the shorthand $h Q$ for the measure whose density is $h$ with respect to $Q$.
\begin{proposition}[Directional local minimum of $\eobj$ satisfies the self-consistency equation]
\label{cor:minimum-is-stationary-eobj}
Let $P \in \dom(\eobj)$. 
Suppose there exists a subderivative $g_\cL$ of $\cL$ at $P$ that is compactly-bounded admissible with respect to $Q_0$. 
Assume that, for each $R\in\cR_{b,c}(Q_0)$, the function
\[
t\longmapsto \eobj(P+t(R-P)),\qquad t\in[0,1],
\]
has a local minimum at \(t=0\).
Then, $P$ satisfies $\dd P / \dd Q_0 \propto \exp(-g_\cL)$, $Q_0$-a.e.
\end{proposition}

\begin{proof}
    Let $\cK(\cdot) = \KLD(\cdot || Q_0)$. 
    Take $R \in \cR_{b, c}$ and $M_t = (1-t)P + tR$ for $t \in (0,1]$. 
    First, 
    note that by convexity 
    \begin{align*}
        \cK(M_t) - \cK(P) \leq  t \{ \cK(R) - \cK(P) \}.
    \end{align*}
    Thus, by the local minimality of $P$, 
     \begin{align*}
        0\leq \frac{\eobj(M_t) - \eobj(P)}{t} \leq  \frac{\cL(M_t) - \cL(P)}{t}   + \cK(R) - \cK(P), 
    \end{align*}
    for sufficiently small $t > 0$; 
    taking $t \to 0$ leads to 
    \begin{equation}
         \int g_\cL\; \dd P + \cK(P) \leq \int g_\cL\; \dd R + \cK(R), \label{eq:eobj-var-ineq}
    \end{equation}
    where we have used the assumption on $\cA_{\cL}(P; g_\cL)$. 

    Now, we concretely consider 
    \[
        R = \frac{\exp(-g_\cL)}{Z_n} 1_{E_n} Q_0
    \]
where 
    $E_n = \{ \verts{g_\cL} \leq n \}\cap \{\Verts{x}_2 \leq n\}$ 
    and $Z_n = \int_{E_n} \exp(-g_\cL) \dd Q_0$. 
Then, \eqref{eq:eobj-var-ineq} implies 
    \begin{align*}
       c \eqqcolon \int g_\cL\; \dd P  + \cK(P) \leq - \log Z_n, 
    \end{align*}
    and hence $\log Z_n \leq -c < \infty$.  
    Moreover, the monotone convergence theorem implies 
    \begin{align*}
        \int g_\cL\; \dd P  + \cK(P) \leq  -\log Z
    \end{align*}
    with $Z = \int \exp(-g_\cL) \dd Q_0 \leq e^{-c}$. 
    Finally, take $Q \propto \exp(-g_\cL) Q_0$. Then 
    \begin{align*}
        0\leq \KLD(P || Q) 
        &= \int \left\{\log \dv[]{P}{Q_0}  - \log \dv{Q_0}{Q}\right\}\; \dd P \\
        &= \cK(P)   + \int g_\cL\; \dd P + \log Z  \leq 0,
    \end{align*}
    which concludes $P = Q$. 
\end{proof}
\noindent Note that in the above proof, the compactness of the support of $h$ is not required, 
but it allows us to consider a weaker assumption on $g_\cL$, for example, by imposing a weaker integrability condition. 
Essentially the assumed local minimality of the original objective implies that $P$ is the minimiser of the entropy-regularised objective obtained by linearising the loss at $P$, and hence the Gibbs measure provides a concrete characterisation of the solution. 

We also make the following observation: 
\begin{proposition}[A minimum has a constant subderivative] 
\label{prop:min-is-stat}
Let $\cF: \finitemomentspace{} \to (-\infty, \infty]$. 
Assume $\cF$ has a global minimum $Q_*$. 
Suppose that there exists $g:\R^d \to \R$ such that 
$hQ_* \in \cA_\cF(Q_*; g)$ for any bounded measurable $h \geq 0$ with $\int h\; \dd Q_* =1$. 
Then, $Q_*$-a.e., 
we have $g = \int g \; \dd Q_*$. 
\end{proposition}
\begin{proof}
Since $Q_*$ is a minimum, we have $Q_* \in \dom(\cF)$, and for any $Q \in \cA_{\cF}(Q_*; g)$, 
\begin{equation}
\lim_{t \to 0^+} \frac{\cF \bigl(Q_* + t(Q-Q_*)\bigr)-\cF(Q_*)}{t} = \int g\; \dd (Q - Q_*)\geq 0. \label{eq:optimum-fv}
\end{equation}
We now specify $Q$ in the form $hQ_*$. 
Specifically, for a bounded measurable function $h_0$, we can take $\eta > 0$ sufficiently small such that 
\begin{align*}
    h_1 & := 1+\eta(h_0 - Q_*(h_0)) \geq 0 , \\
    h_2 & := 1-\eta(h_0- Q_*(h_0)) \geq 0 .
\end{align*}
Substituting $Q = h_i Q_*$ into the right hand side of \eqref{eq:optimum-fv} for $i \in \{1,2\}$ yields the inequalities 
\begin{align*}
     \eta \left\{ Q_*(gh_0) - Q_*(g)Q_*(h_0)\right\} & \geq 0 , \\
     - \eta \left\{ Q_*(gh_0) - Q_*(g)Q_*(h_0)\right\} & \geq 0 
\end{align*}
and thus $Q_*\bigl( \{g-Q_*(g)\} h_0 \bigr) = 0$.
Let $\tilde{g} = g - Q_*(g)$. 
Since the above holds for any bounded measurable $h_0$, 
taking $h_0 = 1\{\tilde{g} \geq 0\} - 1\{\tilde{g} < 0\}$, we have $\int \verts{\tilde{g}} \; \dd Q_* = 0$ 
and thus, $Q_*$-a.e., 
$g=  \int g \; \dd Q_*$. 
\end{proof}

\begin{remark}[Local minimality]
The above proof also works if $Q_*$ is instead a minimum in the set $\{thQ_* + (1-t)Q_*: hQ_* \in \dom(\cF),\ \text{bounded measurable }h\geq0,\ t\in [0,1] \}$. 
\end{remark}

\begin{remark}[Alternative route to self consistency]
If we additionally assume that $\cL$ is convex and $\dom(\eobj)$ contains a measure equivalent to $Q_0$, 
then \cref{prop:min-is-stat} implies the self consistency condition as in \Cref{cor:minimum-is-stationary-eobj}, as long as $\cA_\cL(P; g_\cL)$ contains bounded reweightings of $P$.  
Note that \Cref{cor:minimum-is-stationary-eobj} holds without the convexity of $\cL$. 
\end{remark}

The above result motivates us to consider a refined notion of stationarity compared to \Cref{def:stationary_point} in the main text:

\begin{definition}[$g$-stationarity]
\label{def:alt-stationarity}
Let $\cF: \finitemomentspace{} \to (-\infty, \infty]$. 
A point $Q \in \finitemomentspace{}$ is called $g$-\emph{stationary} with respect to $\cF$
if there exists a subderivative $g:\R^d \to \R$ such that $g=c$ for some $c \in \R$, $Q$-a.e., and $\cA_{\cF, \geq}(Q; g)$ is not empty. 
\end{definition}

This alternative definition of stationarity can also be sufficient for minimality for convex functions: 

\begin{lemma}[$g$-stationary implies minimum for a convex function] \label{lem:stat-suff-min}
Let $\cF: \finitemomentspace{} \to (-\infty, \infty]$ be a convex function.
Suppose there exists $g$-stationary $Q \in \dom(\cF)$ such that 
$\cA_{\cF, \geq}(Q, g) \cap \dom(\cF) = \{R \in \dom(\cF) \setminus\{Q\}: R \ll Q\}$. 
Then, $Q$ is a minimiser of $\cF$ in $\cA_{\cF, \geq}(Q; g)$. 
Consequently, if any $R \in \dom(\cF)$ satisfies $R \ll Q$ (i.e. $g$ is fully admissible at $Q$), then $Q$ is a global minimum of $\cF$. 
\end{lemma}
\begin{proof}
By assumption, for any $R \in \cA_{\cF, \geq}(Q; g) \cap \dom(\cF)$, there exists $f \geq 0$ such that $R = fQ$; 
we also have $g = c$, $Q$-a.e., for some $c \in \R$. 
Now, by the convexity of $\cF$, the function $F: t\in[0,1] \mapsto \cF\bigl( Q + t(R-Q)\bigr)$ is also convex, 
which implies 
\begin{equation*}
    F(1) - F(0) \geq \lim_{t \to 0^+} \frac{F(t) - F(0)}{t}.
\end{equation*}
In terms of $\cF$, this inequality is expressed as 
\begin{align*}
   \cF(R) - \cF(Q) 
   &\geq \int g \dd (R-Q) = \int g(f-1) \dd Q\\
   &= c\int (f-1)\; \dd Q = 0.
\end{align*}
Thus, $Q$ is a minimum in $\cA_{\cF, \geq}(Q; g)$. 
\end{proof}

\begin{proposition}[Sufficient condition for minimality for entropy-regularised objective]
\label{prop:stationarity-eobj-minimality}
Let $\mathcal{K}(\cdot) \coloneq \KLD(\cdot || Q_0)$ in shorthand, so that the entropy-regularised objective $\eobj$ in \eqref{eq: objective} can be written as $\eobj = \cL + \mathcal{K}$. 
Suppose  $\dom(\eobj) = \dom(\cL) \cap \dom(\cK) \not= \emptyset$, 
and that there exists $Q \in \dom(\eobj)$ satisfying $\dd Q / \dd Q_0 \propto \exp(-g_\cL)$, $Q_0$-a.e. for some subderivative $g_\cL$ of $\cL$.   
Take $g = g_\cL + g_\cK$, where $g_\cK \coloneqq \log \dd Q/ \dd Q_0$. 
Then, $Q$ is $g$-stationary with respect to $\eobj$. 
Further, assume the following: 
(a) $\cL$ is convex, and 
(b) $\{R \in \dom(\eobj)\setminus\{Q\}: R\ll Q\} \subset \cA_{\cL, \geq}(Q; g_\cL)$. 
Then $Q$ is a minimiser of $\eobj$. 
\end{proposition}
\begin{proof}
First, the $Q_0$-a.e. relation $\dd Q / \dd Q_0 \propto \exp(-g_\cL)$ shows both the equivalence between $Q$ and $Q_0$ and therefore the $g$-stationarity. The rest of the proof is devoted to the second claim, for which \Cref{lem:stat-suff-min} will be used.

First, for $R \in \dom(\eobj) = \dom(\cK) \cap \dom(\cL)$,
the equivalence between $Q$ and $Q_0$ implies that  $R \ll Q$, since $R\in \dom(\cK)$ satisfies $R \ll Q_0$; and \Cref{lem:KL-fv} leads to $R \in \cA_{\cK, \geq}(Q; g_\cK)$. 
From this observation and the assumed condition (b), we have $R \in \cA_{\cL, \geq}(Q; g_\cL)$, 
and consequently $R \in \cA_{\cL, \geq}(Q; g_\cL) \cap \cA_{\cK, \geq}(Q; g_\cK)$. 
Note that for $R \in \dom(\cL)$, we have $\int g_\cL\; \dd (R-Q) < \infty$, as otherwise $\cL(R) = \infty$ by the same convexity argument as in the proof of \Cref{lem:stat-suff-min} (the same goes for $g_\cK$ and $\cK$). 
Now, 
for $R \in \cA_{\cL, \geq}(Q; g_\cL) \cap \cA_{\cK, \geq}(Q; g_\cK)$, 
\begin{align*}
    \liminf_{t \to 0^+}\frac{ \eobj\bigl(Q + t (R - Q)\bigr)-\eobj(Q)}{t}
    &\geq \int g_\cL \dd(R-Q) + \int g_\cK \dd(R-Q)\\
    &= \int g\; \dd (R-Q).
\end{align*}
Thus, we have $R \in \cA_{\eobj, \geq}(Q; g)$ whenever $R \in \dom(\eobj) \subset \cA_{\cL, \geq}(Q; g_\cL) \cap \cA_{\cK, \geq}(Q; g_\cK)$. 
Applying \Cref{lem:stat-suff-min} (with convexity of $\eobj)$ concludes the proof. 
\end{proof}

\begin{remark}[On the definition of stationarity]
\Cref{cor:minimum-is-stationary-eobj} shows that $\dd Q/ \dd Q_0 \propto \exp(-g_\cL)$ is a necessary condition for minimality with respect to an entropy-regularised objective (even when $\cL$ is not convex), 
whereas \Cref{prop:stationarity-eobj-minimality} shows that $\dd Q/ \dd Q_0 \propto \exp(-g_\cL)$ is also sufficient when a convex loss is used. 
\end{remark}

\subsubsection{Concrete Examples} \label{app: examples}
This appendix provides concrete examples of subderivatives and their admissible directions, illustrating the concepts introduced in \Cref{app: min and stat}. 

\begin{lemma}[Linear loss]
Let $\mathcal{U}(Q) = \int u\; \dd Q$, where $u: \R^d \to \R$. 
For $Q \in \dom(\mathcal{U}) = \{Q\in \finitemomentspace{}: u \in \cL^1(Q) \}$, 
$u$ is a subderivative with $\cA_{\mathcal{U}}(Q; g) = \dom(\mathcal{U})$
and $\cA_{\mathcal{U}, \geq}(Q; g) \supset \{Q\ \in \finitemomentspace{}: u_{-} \in \cL^1(Q)\}$. 
If $u$ is lower-bounded, then 
$\cA_{\mathcal{U}, \geq}(Q; g) = \finitemomentspace{}$. 
\end{lemma}

\begin{lemma}[Interaction energy] 
\label{lem:fv-interact}
    For $r \geq 2$, let $w: \prod_{i=1}^r \R^d \to \R$ and 
    \[
    \mathcal{W}(Q)  
    = \int w(x_1,\dots, x_r)\; \dd Q^{\otimes r}(x_1, \dots, x_r)
    \]
    with $\dom (\mathcal{W}) = \{ Q \in \finitemomentspace{}: w \in \cL^1(Q^{\otimes r})\}$. 
    For $\theta: \R^d \to [1,\infty)$, take 
    $\finitemomentspace{\theta} \coloneqq \{Q \in \finitemomentspace{}: \theta \in \cL^1(Q)\}$. 
    Suppose $-w_{-}(x_1, \dots, x_r) \lesssim \prod_{i=1}^r\theta(x_i)$. 
    Then, the function
    \[
        g(x) = \begin{cases}
            &\tilde{g}(x)  \text{ if } \tilde{g} (x) < \infty,\\
            & 0 \text{ otherwise},
        \end{cases}
    \]
    where 
    \[
        \tilde{g}(x) = \sum_{i=1}^{r} \int w(y_1, \dots, y_r)\; \dd Q_{x, i}(y_1, \dots, y_r), 
    \]
    and 
    \[
        Q_{x, i} \coloneqq Q\otimes \dots \otimes Q \otimes \underbrace{\delta_{x}}_{i\mathrm{th\ component}} \otimes Q \otimes \dots \otimes Q, 
    \]
    is a subderivative at $Q \in \finitemomentspace{\theta } \cap \dom(\mathcal{W})$ 
    with $\finitemomentspace{\theta} \setminus \{Q\} \subseteq \cA_{\mathcal{W}, \geq}(Q; g)$; 
    in particular, if $w$ is lower-bounded, then $\finitemomentspace{} \setminus \{Q\} \subseteq \cA_{\mathcal{W}, \geq}(Q; g)$.
    If additionally $w(x_1, \dots, x_r) \lesssim \prod_{i=1}^r\theta(x_i)$, then $g=\tilde{g}$ and $\finitemomentspace{\theta} \setminus\{Q\} \subseteq \cA_{\mathcal{W}}(Q; g)$
\end{lemma}
\begin{proof}
    Formally, with $M_t = (1-t)Q + tR$, we have 
    \begin{align*}
        \mathcal{W}(M_t) = \sum_{j=0}^r (1-t)^{r-j} t^j \sum_{S \subseteq \{1, \dots, r\}, \verts{S}=j} \int w\; \dd \left( R^S \otimes Q^{S^c}\right),
    \end{align*}
    where $R^S \otimes \dd Q^{S^c}$ denotes the product measure whose $i$th component is $R$ is $i \in S$ and $Q$ otherwise for $i \in \{1, \dots, r\}$ (under the convention $S = \emptyset$ if $\verts{S} = 0$). 
    This expansion is valid if $w_{-}$ is finitely integrable with respect to $R^S \otimes \dd Q^{S^c}$ (and hence the RHS is well-defined), which holds if $Q, R \in \finitemomentspace{\theta}$. 
    If $\int\tilde{g}\,\dd R=\infty$, the first-order mixed terms make the $\mathcal{W}(M_t)$ infinite, so the desired inequality holds trivially. Otherwise, $g=\tilde{g}$ $R$-a.e. and the following expansion applies: 
     \begin{align*}
        \frac{ \mathcal{W}(M_t) - \mathcal{W}(Q)} {t}
        &= (1-t)^{r-1} \left(\int g\; \dd (R-Q) + r \mathcal{W}(Q)\right)+  \frac{(1-t)^{r} - 1}{t} \cdot \mathcal{W}(Q)  \\
        & \quad+
        \sum_{j=2}^r (1-t)^{r-j} t^{j-1} \sum_{S \subseteq \{1, \dots, r\}, \verts{S}=j} \int w\; \dd \left( R^S \otimes Q^{S^c}\right).
    \end{align*}
    Although the third term can be $\infty$, it becomes finite if $w(x_1, \dots, x_r) \lesssim \prod_{i=1}^r \theta(x_i)$, in which case vanishes as $t \to 0^+$. 
    In consequence, we have arrived at the evaluation
    \begin{align*}
        \liminf_{t\to 0^+} \frac{ \mathcal{W}(M_t) - \mathcal{W}(Q)} {t}
        &\geq \int g\; \dd (R-Q),
    \end{align*}
    where the equality holds (and the limit exists) if $R \in \finitemomentspace{\theta}$ and $w(x_1, \dots, x_r) \lesssim \prod_{i=1}^r \theta(x_i)$. 
\end{proof}

\begin{lemma}[KL-divergence]
\label{lem:KL-fv}
Let $\cK(\cdot) = \KLD(\cdot || Q_0)$
and $Q \in \dom(\cK)$. 
Then at $Q$, 
take a subderivative 
\begin{equation*}
    g = \log \dv[]{Q}{Q_0},
\end{equation*}
where $\dd Q / \dd Q_0$ is a version of the Radon-Nikodym derivative, and 
$g$ is defined to be zero on $\{\dd Q/ \dd Q_0=0\}$.
Then, 
\[
\{R \in \dom(\cK):  R \ll Q\} \subset 
\left\{R \in \finitemomentspace{}:  R \ll Q, \int g\, \dd R \text{ is well defined} \right\} \subset \cA_{\cK, \geq}(Q; g)
\]
and $\{R \in \dom(\cK): \KLD(R||Q)<\infty\} \subset \cA_{\cK}(Q; g)$. 
\end{lemma}
\begin{proof}
For $R \ll Q$, 
consider the mixture $M_t = Q + t(R-Q)$, having density 
$\dd M_t / \dd Q_0 = m_t=  (1-t) q + t r$, where $r= \dd R/ \dd Q_0$ and $q = \dd Q/\dd Q_0$. 
Let $\phi(z) = z \log z$ for $z > 0$ and $\phi(0) = 0$. 
Then, for $t > 0$, 
\begin{equation}
    \frac{\KLD(M_t || Q_0) - \KLD(Q || Q_0)}{t}
    = \int \frac{\phi\bigl(m_t(x)\bigr) - \phi\bigl(q(x)\bigr)}{t}\; \dd Q_0(x). 
    \label{eq:KL-quotient}
\end{equation}    
Note that due to $R \ll Q$, we have $r=0$ wherever $q = 0$, and thus the quotient in \eqref{eq:KL-quotient} vanishes in $\{q=0\}$. 
Since $\phi$ is convex, we have $\phi(z) - \phi(z') \geq \phi'(z')(z - z')$ for $z, z'\ \in (0, \infty)$ 
and thus 
\begin{align*}
    \eqref{eq:KL-quotient} 
    \geq \int \left(\log q(x) + 1 \right)(r(x) - q(x))\; \dd Q_0(x)
    = \int \log q\; \dd (R-Q),
\end{align*}
where the inequality holds whenever the RHS is well defined. 
If $R \in \dom(\mathcal{K})$, then we have 
\begin{align*}
    \int g_+\; \dd R 
    &\leq  \KLD(R || Q_0) + \log \left( \int e^{g_+}\; \dd Q_0 \right)\\
    &\leq \KLD(R || Q_0)  + \log \left\{ \int \left(\dv{Q}{Q_0} \lor 1\right) \dd Q_0\right\}\\
    &\leq \KLD(R || Q_0)  + \log 2, 
\end{align*}
where $g_+(x) = g(x) \lor 0$ and the first inequality follows from evaluating the KL between $R$ and the Gibbs measure $\exp(g_+) Q_0 / \int e^{g_+} \dd Q_0$  (using truncation appropriately). Thus, the positive part of $g$ is always finitely integrable. 

We next show that the equality holds under the additional assumption $\KLD(R||Q) < \infty$. 
Define the Bregman divergence 
\[
    d_\phi(z, w) \coloneqq \phi(z) - \phi(w) - \phi'(w)(z - w) \geq 0
\]
for $z, w \in (0,\infty)$. 
Note that if $r(x) > 0$ and $q(x) > 0$, 
\begin{align*}
    d_\phi(m_t(x), q(x)) %
    &= m_t(x) \log \frac{m_t(x)}{q(x)} - m_t(x) + q(x). 
\end{align*}
We thus have 
\begin{align*}
     \eqref{eq:KL-quotient}
     &= \int \phi'(q(x)) \dd (R-Q) + \int d_\phi(m_t(x), q(x))\; \dd Q_0 (x)\\
     &= \int \log \dv[]{Q}{Q_0}\dd (R-Q) + \frac{1}{t}\int \phi\left( \frac{m_t}{q} \right)\; \dd Q.
\end{align*}
We show that the second term vanishes as $t \to 0^+$. 

Let $\tilde{r} = r(x)/q(x)$. 
Then, 
\begin{align*}
    &\frac{1}{t}\int \phi\left( \frac{m_t}{q} \right)\; \dd Q
    = t^{-1}\int \bigl( \bigl(1 + t(\tilde{r}-1)\bigr) \log  \bigl(1 + t(\tilde{r}-1)\bigr) -\underbrace{t(\tilde{r}-1)}_{\text{zero integral}} \bigr) \; \dd Q.
\end{align*}
Since 
\begin{align*}
    \frac{  \verts{\bigl(1 + t(\tilde{r}-1)\bigr) \log  \bigl(1 + t(\tilde{r}-1)\bigr) -t(\tilde{r}-1) }}{t(1-\tilde{r})} \leq 1 - \tilde{r} \leq 1
\end{align*}
for any $\tilde{r} \leq 1$. 
By the dominated convergence theorem, we have 
\begin{align*}
    &\lim_{t \to 0^+} \frac{1}{t}\int_{\tilde{r} \leq 1} \bigl( \bigl(1 + t(\tilde{r}-1)\bigr) \log  \bigl(1 + t(\tilde{r}-1)\bigr) -t(\tilde{r}-1) \bigr)  \; \dd Q\\
    &=\int_{\tilde{r} \leq 1} \Bigl( \lim_{t \to 0^+} \frac{1}{t} \bigl(1 + t(\tilde{r}-1)\bigr) \log  \bigl(1 + t(\tilde{r}-1)\bigr) -(\tilde{r}-1) \Bigr)  \; \dd Q\\
    &=\int_{\tilde{r} \leq 1} \bigl( (\tilde{r}-1) - (\tilde{r}-1) \bigr)  \; \dd Q = 0.
\end{align*}
We also have 
\begin{align*}
    0&\leq \lim_{t\to 0^+}\frac{1}{t}\int_{\tilde{r} > 1} \bigl(  \bigl(1 + t(\tilde{r}-1)\bigr) \log  \bigl(1 + t(\tilde{r}-1)\bigr) -t(\tilde{r}-1)  \bigr)  \; \dd Q\\
    &\leq \int_{\tilde{r} > 1}  \lim_{t \to 0^+} (\tilde{r}-1) \log  \bigl(1 + t(\tilde{r}-1)\bigr) \; \dd Q = 0,
\end{align*}
where we have used the dominated convergence theorem along with the pair of inequalities $\log(1+t\tilde{r}-1) \leq \log (1+ \tilde{r}-1)$ and $\log(1+x) \leq x$ for $x\geq 0$, 
and the assumption $\int \tilde{r}\log{\tilde{r}}\dd Q = \KLD(R || Q) < \infty$. 
\end{proof}

\subsection{Variational Gradient Descent}
\label{app: svgd}

This appendix contains the proof of \Cref{prop: svgd converge}, reported in \Cref{sec: svgd}.
The strategy is based on Theorem 1 of \citet{banerjee2025improved}:

\begin{proof}[Proof of \Cref{prop: svgd converge}]
Introduce the shorthand $\mathbf{x} \coloneqq (x_1 , \dots , x_n) \in \mathbb{R}^{d \times n}$ and 
\begin{align*}
\Phi_{\mathbf{x}}(x_i,x_j) \coloneqq k(x_i , x_j) \underbrace{ (\nabla \log q_0 - \vargrad  \cL(Q_n))(x_j) }_{ =: \genscore{P}{Q_n}(x_j) } + \nabla_1 k(x_j , x_i) , \qquad Q_n \coloneqq \frac{1}{n} \sum_{j=1}^n \delta_{x_j}  ,
\end{align*}
where for convenience we have suppressed the $t$-dependence, i.e. $x_i \equiv x_i(t)$ and $Q_n \equiv Q_n(\mathbf{x})$.
Recall from \Cref{lem: computable} that $\genpdf{Q}(x) \coloneqq q_0(x) \exp(-\cL'(Q)(x))$, so that $\genscore{P}{Q_n} = \nabla \log \genpdf{Q_n}$; in what follows only the function $\genscore{P}{Q_n}$ will be required.

\smallskip
\noindent \textit{Flow of the particle system:}
Let $\mathbf{V} = (V_1,\dots,V_n) : \mathbb{R}^{d \times n} \rightarrow \mathbb{R}^{d \times n}$ with component functions $V_i(\mathbf{x}) \coloneqq \frac{1}{n} \sum_{j=1}^n \Phi_{\mathbf{x}}(x_i,x_j)$, so that \eqref{eq: gen svgd} reads $\dot{\mathbf{x}} = \mathbf{V}(\mathbf{x})$.
Under our assumptions $\mathbf{V}$ is $C^1(\mathbb{R}^{d \times n})$: indeed, $\nabla \log q_0$ is $C^1(\mathbb{R}^d)$, $\mathbf{x} \mapsto \vargrad \cL(Q_n)(x_j)$ is $C^1(\mathbb{R}^{d \times n})$ by \Cref{item: loss reg asm}, and $k$ and $\nabla_1 k$ are $C^1$ since $\phi \in C^2(\mathbb{R}^d)$.
Moreover, $\mathbf{V}$ has at most linear growth.
To see this, note from Regularisation (a) that $\|\nabla \log q_0(x_j)\| \lesssim 1 + \|x_j\|$, while from \eqref{eq: vargrad growth} in the proof of \Cref{lem:initial-entropy} (which uses only \Cref{item: loss reg asm}(a,b)) we have $\|\vargrad \cL(Q_n)(x_j)\|_2 \leq \mathfrak{b}_0 + \mathfrak{M} \|x_j\|_2$.
Thus
\begin{align*}
\|(\nabla \log q_0 - \vargrad  \cL(Q_n))(x_j)\| \lesssim 1 + \|\mathbf{x}\|
\end{align*}
for each $j \in \{1,\dots,n\}$.
Since $k$ is a positive-definite translation-invariant kernel, by Bochner's theorem $\phi$ is the Fourier transform of a finite non-negative measure, which has finite second moments since $\phi \in C^2(\mathbb{R}^d)$; hence $k$ and its partial derivatives up to second order are bounded.
It follows that $\|\Phi_{\mathbf{x}}(x_i,x_j)\| \lesssim 1 + \|\mathbf{x}\|$ for each $i,j \in \{1 , \dots , n\}$, and hence $\|\mathbf{V}(\mathbf{x})\| \lesssim 1 + \|\mathbf{x}\|$.
Since $\mathbf{V}$ is $C^1$ with at most linear growth, solutions of $\dot{\mathbf{x}} = \mathbf{V}(\mathbf{x})$ exist, are unique, and are defined for all $t \in \mathbb{R}$, and the associated flow $\mathbf{F}_t : \mathbb{R}^{d \times n} \rightarrow \mathbb{R}^{d \times n}$, $\mathbf{F}_0 = \mathrm{Id}$, is a $C^1$ diffeomorphism for each $t$ (with inverse $\mathbf{F}_{-t}$).
Further, $(t,\mathbf{y}) \mapsto \mathbf{F}_t(\mathbf{y})$ and $(t,\mathbf{y}) \mapsto D\mathbf{F}_t(\mathbf{y})$ are continuous, and the Jacobian matrix satisfies the variational equation $\partial_t D\mathbf{F}_t(\mathbf{y}) = D\mathbf{V}(\mathbf{F}_t(\mathbf{y})) \, D\mathbf{F}_t(\mathbf{y})$ \citep[][Chapter 5, Theorem 3.1]{hartman2002ordinary}.
Let $\Lambda_t(\mathbf{y}) \coloneqq \det D\mathbf{F}_t(\mathbf{y})$; from the variational equation and Jacobi's formula,
\begin{align}
    \partial_t \log \Lambda_t(\mathbf{y}) = \mathrm{tr}\big( D\mathbf{V}(\mathbf{F}_t(\mathbf{y})) \big) = (\nabla \cdot \mathbf{V})(\mathbf{F}_t(\mathbf{y})) , \qquad \Lambda_0(\mathbf{y}) = 1 , \label{eq: Jacobi}
\end{align}
so that $\Lambda_t(\mathbf{y}) = \exp ( \int_0^t (\nabla \cdot \mathbf{V})(\mathbf{F}_s(\mathbf{y})) \, \mathrm{d}s ) > 0$.
Since $\mathbf{x}(t) = \mathbf{F}_t(\mathbf{x}(0))$ and $\mathbf{x}(0)$ has density $p_n(0,\cdot) = \mu_0^{\otimes n}$, the change of variables formula shows that $\mathbf{x}(t)$ has a density $p_n(t,\cdot)$, characterised by
\begin{align}
    p_n(t,\mathbf{F}_t(\mathbf{y})) \, \Lambda_t(\mathbf{y}) = p_n(0,\mathbf{y}) , \qquad \mathbf{y} \in \mathbb{R}^{d \times n} . \label{eq: change of variables}
\end{align}
Let $S_0 \coloneqq \operatorname{supp} p_n(0,\cdot) = (\operatorname{supp} \mu_0)^n$, which is compact.
Then $p_n(t,\cdot)$ is supported in $\mathbf{F}_t(S_0)$ and, for $t \in [0,T]$, these sets are contained in the compact set $\{\mathbf{F}_s(\mathbf{y}) : s \in [0,T], \, \mathbf{y} \in S_0\}$.

\smallskip
\noindent \textit{Main part of the proof:}
Consider the function $\pi_n : \mathbb{R}^{d \times n} \rightarrow (0,\infty)$ defined as
\begin{align}
    \pi_n(\mathbf{x}) \coloneqq q_0(x_1) \cdots q_0(x_n) \exp\left( - n \cL(Q_n) \right) , \qquad Q_n \equiv Q_n(\mathbf{x}) , \label{eq: n particle Gibbs}
\end{align}
which is (up to normalisation) the $n$-particle Gibbs measure associated with the objective \eqref{eq: objective} \citep[cf.][]{chen2022uniform}.
Crucially, $\pi_n$ is a fixed function on $\mathbb{R}^{d \times n}$; it does not depend on $t$.
From the chain rule in \Cref{item: loss reg asm}, $\mathbf{x} \mapsto \log \pi_n(\mathbf{x})$ is $C^1(\mathbb{R}^{d \times n})$ and its gradient can be computed exactly, with all dependence of $Q_n$ on $\mathbf{x}$ accounted for, as
\begin{align}
    \nabla_{x_i} \log \pi_n(\mathbf{x}) = \nabla \log q_0(x_i) - n \nabla_{x_i} \cL(Q_n) = \nabla \log q_0(x_i) - \vargrad \cL(Q_n)(x_i) = \genscore{P}{Q_n}(x_i) . \label{eq: grad log pin}
\end{align}
Let
\begin{align}
    H(t) & \coloneqq \int \log \left( \frac{p_n(t,\mathbf{x})}{ \pi_n(\mathbf{x}) } \right) p_n(t,\mathbf{x}) \; \mathrm{d}\mathbf{x} = \KLD(p_n(t,\cdot) || Q_0^{\otimes n}) + n \, \mathbb{E}[\cL(Q_n)] , \label{eq: dummy x}
\end{align}
where the expectation is with respect to $\mathbf{x} \sim p_n(t,\cdot)$.
Since $\pi_n$ does not depend on $t$, the $t$ dependence of $H(t)$ enters only through the term $p_n(t,\mathbf{x})$ in the integrand.
Write $p_0 \equiv p_n(0,\cdot)$.
Here $H(0) = n \KLD(\mu_0 || Q_0) + n \mathbb{E}[\cL(Q_n)]$ is finite, since $\KLD(\mu_0 || Q_0) < \infty$ (see the proof of \Cref{lem:initial-entropy}) and $\mathbf{x} \mapsto \cL(Q_n)$ is continuous, hence bounded on $S_0$.
The change of variables $\mathbf{x} = \mathbf{F}_t(\mathbf{y})$, together with \eqref{eq: change of variables}, gives
\begin{align}
    H(t) & = \int \log \left( \frac{ p_0(\mathbf{y}) }{ \pi_n(\mathbf{F}_t(\mathbf{y})) \, \Lambda_t(\mathbf{y}) } \right) p_0(\mathbf{y}) \; \mathrm{d}\mathbf{y} \nonumber \\
    & = H(0) - \int \left[ \log \pi_n(\mathbf{F}_t(\mathbf{y})) - \log \pi_n(\mathbf{y}) + \log \Lambda_t(\mathbf{y}) \right] p_0(\mathbf{y}) \; \mathrm{d}\mathbf{y} . \label{eq: H change of variables}
\end{align}
From the chain rule, $\partial_t \log \pi_n(\mathbf{F}_t(\mathbf{y})) = \nabla \log \pi_n(\mathbf{F}_t(\mathbf{y})) \cdot \mathbf{V}(\mathbf{F}_t(\mathbf{y}))$, which together with \eqref{eq: Jacobi} shows that the $t$-derivative of the integrand in \eqref{eq: H change of variables} is $p_0(\mathbf{y})$ multiplied by a function that is continuous in $(t,\mathbf{y})$, and hence bounded on the compact set $[0,T] \times S_0$.
Differentiation under the integral is therefore justified by the dominated convergence theorem and, transforming back using \eqref{eq: change of variables},
\begin{align*}
    H'(t) & = - \int \left[ (\nabla \cdot \mathbf{V})(\mathbf{F}_t(\mathbf{y})) + \nabla \log \pi_n(\mathbf{F}_t(\mathbf{y})) \cdot \mathbf{V}(\mathbf{F}_t(\mathbf{y})) \right] p_0(\mathbf{y}) \; \mathrm{d}\mathbf{y} \\
    & = - \int \left[ (\nabla \cdot \mathbf{V})(\mathbf{x}) + \nabla \log \pi_n(\mathbf{x}) \cdot \mathbf{V}(\mathbf{x}) \right] p_n(t,\mathbf{x}) \; \mathrm{d}\mathbf{x} .
\end{align*}
Finally, $(\nabla \cdot \mathbf{V})(\mathbf{x}) = \sum_{i=1}^n \nabla_{x_i} \cdot V_i(\mathbf{x}) = \frac{1}{n} \sum_{i=1}^n \sum_{j=1}^n \nabla_{x_i} \cdot \Phi_{\mathbf{x}}(x_i,x_j)$, where $\nabla_{x_i}$ acts on all dependence on $x_i$, and from \eqref{eq: grad log pin} we have $\nabla \log \pi_n(\mathbf{x}) \cdot \mathbf{V}(\mathbf{x}) = \frac{1}{n} \sum_{i=1}^n \sum_{j=1}^n \genscore{P}{Q_n}(x_i) \cdot \Phi_{\mathbf{x}}(x_i,x_j)$.
Thus
\begin{align*}
    H'(t) & = - \frac{1}{n} \sum_{i=1}^n \sum_{j=1}^n \int \left(  \nabla_{x_i} \cdot \Phi_{\mathbf{x}}(x_i,x_j) + \genscore{P}{Q_n}(x_i) \cdot \Phi_{\mathbf{x}}(x_i,x_j) \right) p_n(t,\mathbf{x}) \; \mathrm{d}\mathbf{x} .
\end{align*}

\smallskip
\noindent \textit{Algebraic simplification:}
Now, since $k$ is translation-invariant we have $\nabla_1 k(x,x) = 0$ for all $x \in \mathbb{R}^d$, and
\begin{align*}
    \nabla_{x_i} \cdot \Phi_{\mathbf{x}}(x_i,x_j) & = \nabla_{x_i} \cdot [ k(x_i , x_j) \genscore{P}{Q_n}(x_j)] + \nabla_{x_i} \cdot [ \nabla_1 k(x_j , x_i) ] \\
    & = \nabla_1 k(x_i,x_j) \cdot \genscore{P}{Q_n}(x_j) + k(x_i,x_j) \nabla_{x_i} \cdot \genscore{P}{Q_n}(x_j) + \nabla_1 \cdot \nabla_2 k(x_i,x_j) \\
    & \qquad +  \{ \underbrace{ \nabla_1 k(x_i,x_i) }_{=0} \cdot \genscore{P}{Q_n}(x_i) + \Delta_1 k(x_i,x_i)  \} \mathbbm{1}_{i=j}  \\
    \genscore{P}{Q_n}(x_i) \cdot \Phi_{\mathbf{x}}(x_i,x_j) & = k(x_i,x_j) \genscore{P}{Q_n}(x_i) \cdot \genscore{P}{Q_n}(x_j) + \genscore{P}{Q_n}(x_i) \cdot \nabla_1 k(x_j,x_i) .
\end{align*}
In addition, let $\nabla_{x_i} \cdot b_{Q_n}(x_j)$ be a shorthand for $\nabla_i \cdot f(x_1,\dots,x_n)$ where $f(x_1,\dots,x_n) \coloneq b_{Q_n}(x_j)$.
Let $\Psi(\mathbf{x},y) \coloneqq \vargrad \cL(Q_n)(y)$, which is $C^1(\mathbb{R}^{d \times n} \times \mathbb{R}^d)$ by \Cref{item: loss reg asm}.
From the multivariate chain rule,
$\nabla_{x_i} \cdot [\Psi(\mathbf{x},x_j)] = \{ (\nabla_y \cdot \Psi)(\mathbf{x},x_i) \} \mathbbm{1}_{i=j} + [ \nabla_{x_i} \cdot \Psi(\mathbf{x},y) ]_{y = x_j}$,
where the first term is $\nabla \cdot \vargrad \cL(Q_n)(x_i)$ and, from the second chain rule in \Cref{item: loss reg asm}, the second term is $\frac{1}{n} \vargrad \cdot \vargrad \cL(Q_n)(x_j)(x_i)$.
Thus
\begin{align}
    \nabla_{x_i} \cdot \genscore{P}{Q_n}(x_j) & =  \{ \nabla \cdot ( \nabla \log q_0 - \vargrad  \mathcal{L}(Q_n) )(x_i)  \} \mathbbm{1}_{i=j}  -  \frac{1}{n} \vargrad  \cdot \vargrad  \mathcal{L}(Q_n)(x_j)(x_i)  \label{eq: div score}
\end{align}
where we note in passing that the final term in \eqref{eq: div score} vanishes for linear $\cL$ and was therefore not included in the analysis of \citet{banerjee2025improved} for \ac{svgd}.

\smallskip
\noindent \textit{Overall bound:}
The above calculations have led us to the overall expression
\begin{align*}
    H'(t) & = - n \mathbb{E} [ \KGD_K^2(Q_n) ]  \\
    & \qquad + \mathbb{E} \Bigg[ \frac{1}{n^2}  \sum_{i=1}^n \sum_{j=1}^n  k(x_i,x_j) \vargrad  \cdot \vargrad  \mathcal{L}(Q_n)(x_j)(x_i)   \\
    & \qquad \qquad \qquad  - \frac{1}{n}  \sum_{i=1}^n  \Big\{ k(x_i,x_i) [ \Delta \log q_0(x_i) - \nabla \cdot \vargrad  \mathcal{L}(Q_n)(x_i) ] 
    + \Delta_1 k(x_i,x_i) \Big\}  \Bigg] 
\end{align*}
where the expectation is with respect to the (random) initialisation of the particle set.
Thus, under our assumptions, specifically \Cref{item: loss reg asm}(b,c) (noting that $|\nabla \cdot \vargrad \cL(Q_n)(x_i)| \leq d \, \mathfrak{M}$), Regularisation (a), and boundedness of $k$ and its partial derivatives up to second order,
\begin{align*}
    H'(t) & \leq - n \mathbb{E}[ \KGD_K^2(Q_n)  ]  + C_K   
\end{align*}
for some finite, $K$-dependent constant $C_K$.
Integrating both sides from $0$ to $T$ and rearranging yields
\begin{align*}
    \frac{1}{T} \int_0^T \mathbb{E}[ \KGD_K^2(Q_n) ] \; \mathrm{d}t \leq \frac{H(0) - H(T)}{n T} + \frac{C_K}{n} .
\end{align*}
The result follows since $H(0) - H(T) \leq n C_0$ for an $n$- and $T$-independent constant $C_0$, from \Cref{lem:initial-entropy}.
\end{proof}

\begin{lemma}[Control of $H(0) - H(T)$]
\label{lem:initial-entropy}
In the setting of \Cref{prop: svgd converge}, for all $n \in \mathbb{N}$ and $T \geq 0$,
\[
H(0) - H(T) \leq n C_0 , \qquad C_0 \coloneqq \KLD(\mu_0 || Q_0) + \mathfrak{b}_0 R_S + \tfrac{\mathfrak{M}}{2} R_S^2 + \log \bar{Z} < \infty ,
\]
where $R_S \coloneqq \sup_{x \in S} \|x\|_2$, $S \coloneqq \operatorname{supp} \mu_0$, and $\bar{Z} \coloneqq \int q_0(x) \exp\!\left( \tfrac{\mathfrak{M}}{2}\|x\|_2^2 + \mathfrak{b}_0 \|x\|_2 \right) \mathrm{d}x$.
\end{lemma}
\begin{proof}
\noindent \textit{Growth of the loss along particle configurations:}
Fix $Q \in \dom(\cL)$ and $y \in \mathbb{R}^d$.
By the fundamental theorem of calculus applied to $r \mapsto \vargrad \cL(Q)(ry)$,
\[
\vargrad \cL(Q)(y) = \vargrad \cL(Q)(0) + \int_0^1 \nabla \vargrad \cL(Q)(ry) \, y \; \mathrm{d}r ,
\]
and hence, from \Cref{item: loss reg asm}(a,b) and since $\|Ay\|_2 \leq \|A\|_{\mathrm{op}} \|y\|_2$,
\begin{align}
\|\vargrad \cL(Q)(y)\|_2 \leq \mathfrak{b}_0 + \mathfrak{M} \|y\|_2 . \label{eq: vargrad growth}
\end{align}
Fix $\mathbf{x} = (x_1,\dots,x_n) \in \mathbb{R}^{d \times n}$ and consider the path $s \mapsto s \mathbf{x}$, $s \in [0,1]$, along which the empirical distribution $Q_n(s \mathbf{x})$ moves from $\delta_0$ (at $s = 0$) to $Q_n(\mathbf{x})$ (at $s = 1$); note that $Q_n(s\mathbf{x}) \in \dom(\cL)$, since $\cL(Q_n)$ is real-valued by \Cref{item: loss reg asm}.
From the chain rule in \Cref{item: loss reg asm},
\[
\frac{\mathrm{d}}{\mathrm{d}s} \cL(Q_n(s\mathbf{x})) = \sum_{i=1}^n \nabla_{x_i} \cL(Q_n)\big|_{s \mathbf{x}} \cdot x_i = \frac{1}{n} \sum_{i=1}^n \vargrad \cL(Q_n(s\mathbf{x}))(s x_i) \cdot x_i ,
\]
whose absolute value, by the Cauchy--Schwarz inequality, is at most $\frac{1}{n} \sum_{i=1}^n ( \mathfrak{b}_0 \|x_i\|_2 + \mathfrak{M} s \|x_i\|_2^2 )$.
Integrating over $s \in [0,1]$,
\begin{align}
    \left| \cL(Q_n(\mathbf{x})) - \cL(\delta_0) \right| \leq \frac{1}{n} \sum_{i=1}^n \zeta(x_i) , \qquad \zeta(x) \coloneqq \mathfrak{b}_0 \|x\|_2 + \tfrac{\mathfrak{M}}{2} \|x\|_2^2 . \label{eq: L growth particles}
\end{align}

\smallskip
\noindent \textit{Upper bound on $H(0)$:}
At $t = 0$ we have $p_n(0,\cdot) = \mu_0^{\otimes n}$ and, by the tensorisation of the \ac{kld} over product measures together with \eqref{eq: L growth particles} and $x_i \in S$,
\begin{align*}
H(0) &= n \KLD(\mu_0 || Q_0) + n \, \mathbb{E}_{x_1,\dots,x_n \sim \mu_0}[\cL(Q_n)] \\
    &\leq n \left\{ \KLD(\mu_0 || Q_0) + \cL(\delta_0) + \mathfrak{b}_0 R_S + \tfrac{\mathfrak{M}}{2} R_S^2 \right\} .
\end{align*}
Here $\KLD(\mu_0 || Q_0) = \int \mu_0 \log \mu_0 - \int \mu_0 \log q_0 < \infty$, since the density $\mu_0$ is bounded, so that $\int \mu_0 \log \mu_0 \leq \log \sup_x \mu_0(x)$, while $\log q_0$ is continuous, and hence bounded on the compact set $S$.

\smallskip
\noindent \textit{Lower bound on $H(T)$:}
From \eqref{eq: L growth particles}, $n \cL(Q_n) \geq n \cL(\delta_0) - \sum_{i=1}^n \zeta(x_i)$, and thus, writing $p_n \equiv p_n(T,\cdot)$,
\[
H(T) \geq n \cL(\delta_0) + \int \log \left( \frac{ p_n(\mathbf{x}) }{ \prod_{i=1}^n q_0(x_i) \exp(\zeta(x_i)) } \right) p_n(\mathbf{x}) \; \mathrm{d}\mathbf{x} \geq n \cL(\delta_0) - n \log \bar{Z} ,
\]
where the final inequality holds since $\prod_{i=1}^n q_0(x_i) \exp(\zeta(x_i))$ has total mass $\bar{Z}^n$, which is finite by the Regularisation assumption in \Cref{prop: svgd converge}, and the \ac{kld} between $p_n$ and the normalised version of this measure is non-negative.

\smallskip
\noindent Combining the upper bound on $H(0)$ with the lower bound on $H(T)$, the terms involving $\cL(\delta_0)$ cancel and $H(0) - H(T) \leq n C_0$, as claimed.
\end{proof}

\subsection{Relaxing the Definition of Gradient Discrepancies}
\label{app: relax def gd}

This appendix explains how the \ac{gd} in \Cref{def: grad disc} of the main text, which assumes the existence of the variational gradient $\vargrad \cL(Q)$ at all $Q \in \finitemomentspace{}$, can be relaxed to accommodate situations where $\vargrad \cL(Q)$ exists only in a restricted sense (cf. the examples in \Cref{app: examples}).

As in the main text, we can define a \ac{gd} that allows us to quantify the discrepancy between $Q$ and $\exp(-g_Q)Q_0$, where $g_Q: \R^d \to \R$ is a $Q$-dependent function. 
For the \ac{gd} to represent a measure of suboptimality for an entropy-regularised objective \eqref{eq: objective}, $g_Q$ must be chosen appropriately. 
\Cref{cor:minimum-is-stationary-eobj} states that the self-consistency condition \eqref{eq: self-consistency} is necessary for $Q$ to be a minimum provided $g_Q$ is a subderivative of $\cL$ at $Q$ (and satisfies the admissibility below). 
In light of this result, we require loss functions to be \emph{sufficiently regular} in the following sense:

\begin{definition}[Functional differentiability]
\label{def: func-differentiablity}
Suppose $\cL: \finitemomentspace{} \to (-\infty, \infty]$ satisfies the following:  
    At any $Q \in \dom(\cL)$, there exists 
    a subderivative $g_Q$ of $\cL$ 
    that is differentiable and 
   compactly-bounded admissible with respect to $Q_0$ (cf. \Cref{def:subd-fv}). 
Then we say that $\cL$ is \emph{functionally differentiable} and 
denote $\mathfrak{D}\cL(Q) \coloneqq g_Q$ for $Q \in \dom(\cL)$ and $\mathfrak{D}_{\mathrm{V}} \cL(Q) \coloneqq \nabla g_Q$. 
\end{definition}

\noindent While $\mathfrak{D} \cL$  need not be unique, any two choices only differ by an additive constant $Q_0$-almost everywhere, according to \Cref{lem:subderivative-unique}. 
Consequently, if $Q_0$ is absolutely continuous with respect to Lebesgue measure, $\mathfrak{ D}_{\mathrm{V}}L(Q)$ is unique $Q_0$-almost everywhere; 
in particular if $Q_0$ has full support, then it is unique everywhere. 
If $\cL$ admits a first variation $\cL'$ as defined in \Cref{subsec: notation main}, then $\mathfrak{D}\cL$ here agrees with $\cL'$ and is uniquely defined (up to an additive constant). 
Since $\mathfrak{D} \cL(Q)$ is not defined outside the domain of $\cL$, we can refine the definition of the \ac{gd} in \Cref{def: grad disc} to take $\infty$ outside $\dom(\cL)$:

\begin{definitionprime}{def: grad disc}[Gradient Discrepancy for entropy-regularised objectives]
    Let $\Op{Q} v(x) = [(\nabla \log q_0)(x) - \mathfrak{D}_{\mathrm{V}} \cL(Q)(x)] \cdot v(x) + (\nabla \cdot v)(x)$.
    For a given set $\mathcal{V}$ of differentiable vector fields on $\mathbb{R}^d$, we define the \acl{gd} as
    \begin{align}
    	\GD(Q) \coloneqq 
            \sup_{\substack{v \in \mathcal{V} \;  \mathrm{s.t.} \\ (\Op{Q} v)_{-} \in \cL^1(Q)}} \left| \int \Op{Q} v(x) \; \dd Q(x) \right| 
    \end{align}
    for $Q \in \dom(\cL)$ and $\GD(Q) \coloneqq \infty$ if $Q \not\in \dom(\cL)$. 
\end{definitionprime}

\noindent The theoretical results that we state in \Cref{sec: theory} are applicable to this more general notion of \ac{gd}, but for clarity we do not emphasise this in their statement or proof; the reader interested in our theoretical arguments should simply replace all occurrences of $\vargrad \cL$ with $\mathfrak{D}_{\mathrm{V}} \cL$ and 
and $\cL'$ with $\mathfrak{D}\cL$. 
Note that for a \ac{gd} to characterise a minimum (not just the self-consistency condition), we need additional assumptions on the subderivative beyond those used in \Cref{def: func-differentiablity}, see \Cref{prop:stationarity-eobj-minimality}. 

\subsection{Preliminary Results on \ac{kgd}}
\label{app: prelim KGD}

This appendix contains the proof of \Cref{lem: computable} and related preliminary results concerning \ac{kgd} that will be required in the sequel.

\begin{lemma}[RKHS]
\label{lem: our rkhs}
    Let $Q_0$ have a density $q_0 > 0$ and let $K$ be a kernel for which the elements of $\genpdf{Q} \rkhs{K}=\{\genpdf{Q}v: v\in\rkhs{K}\}$ are partially differentiable functions.
    Then $\Op{Q} (\rkhs{K})$ is an \ac{rkhs} with kernel $k_K^Q$ in \eqref{eq: PQ kernel}.
\end{lemma}
\begin{proof}
    The following proof is analogous to Appendix C.5 of \citet{barp2024targeted}.
    First note that $\genpdf{Q} \mathcal{H}_K$ is an \ac{rkhs} whose kernel is $K_\rho(x,y) \coloneqq \genpdf{Q}(x) K(x,y) \genpdf{Q}(y)$ \citep[see Appendix B of][]{barp2024targeted}.
    Since $q_0 > 0$ we have $\genpdf{Q} > 0$ as well.
    From the differential reproducing property \citep[][Lemma 4]{barp2024targeted},
    \begin{align*}
        \Op{Q} v(x) = \frac{1}{\genpdf{Q}(x)} \nabla_x \cdot (\genpdf{Q} v)
        & = \frac{1}{\genpdf{Q}(x)} \sum_{i=1}^d \partial_{x_i} (\genpdf{Q} v_i) \\
        & = \frac{1}{\genpdf{Q}(x)} \sum_{i=1}^d \left\langle \partial_{x_i} (\genpdf{Q}(\cdot) K(\cdot,x) \genpdf{Q}(x) e_i ) , \genpdf{Q} v \right\rangle_{K_\rho}
    \end{align*}
    and from \citet[][Proposition 1]{carmeli2010vector} this is
    \begin{align*}
        = \frac{1}{\genpdf{Q}(x)} \sum_{i=1}^d \left\langle \partial_{x_i} ( K(\cdot,x) \genpdf{Q}(x) e_i ) , v \right\rangle_K
        = \left\langle \frac{1}{\genpdf{Q}(x)} \sum_{i=1}^d \partial_{x_i} ( K(\cdot,x) \genpdf{Q}(x) e_i )  , v \right\rangle_K .
    \end{align*}
    This shows that $\Op{Q}$ is the feature operator associated to 
    \begin{align*}
        \xi_P^Q(x) \coloneqq \frac{1}{\genpdf{Q}(x)} \sum_{i=1}^d \partial_{x_i} ( K(\cdot,x) \genpdf{Q}(x) e_i ) 
    \end{align*}
    and thus by \citet[][Proposition 1]{carmeli2010vector}, $\Op{Q} \mathcal{H}_K$ is an \ac{rkhs} with kernel 
    \begin{align*}
        \langle \xi_P^Q(x) , \xi_P^Q(x') \rangle_K
        = \sum_{i=1}^d \sum_{j=1}^d \frac{1}{\genpdf{Q}( x) \genpdf{Q}(x')} \partial_{x_j'} \partial_{x_i} \left( \genpdf{Q}( x)  K_{i,j}( x, x') \genpdf{Q}( x') \right)
        = k_K^Q(x,x') ,
    \end{align*}
    as claimed.
\end{proof}

Now we can prove \Cref{lem: computable}:

\begin{proof}[Proof of \Cref{lem: computable}]
    From \citet[][Lemma 5]{barp2024targeted} we have that $\Op{Q} \mathcal{B}_K$ is equal to the unit ball in $\Op{Q} \mathcal{H}_K$, and from \Cref{lem: our rkhs} this is $\mathcal{B}_{k_K^Q}$.
    Thus, from a standard argument \citep[see e.g.][Lemma 2]{barp2024targeted},
    \begin{align*}
	\KGD_K(Q) = \sup_{h \in \mathcal{B}_{k_K^Q}} \left| \int h(x) \; \dd Q(x) \right|
    = \left( \iint k_K^Q(x,x') \; \dd Q(x) \dd Q(x') \right)^{1/2},
    \end{align*}
    as required.
\end{proof}

\begin{remark}[A simple instance of $k_K^Q$]
\label{rem: explicit formula kgd}
    Taking $K(x,x') = k(x,x') I_{d \times d}$ for $k$ a scalar-valued kernel, one has the simple formula
    \begin{align*}
        k_K^Q(x,x') & = \nabla_1 \cdot \nabla_2 k(x,x') + \nabla_1 k(x,x') \cdot \genscore{P}{Q}(x') \\
        & \qquad \qquad + \nabla_2 k(x,x') \cdot \genscore{P}{Q}(x) + k(x,x') \genscore{P}{Q}(x) \cdot \genscore{P}{Q}(x') ,
    \end{align*}
    which is reminiscent of the Stein kernel \citep{oates2017control} with the generalised score function $\genscore{P}{Q}$ now appearing in place of the Stein score $\nabla \log p$; the two score functions coincide in the special case of linear $\cL$.
\end{remark}

\subsection{Stationarity Characterisation}\label{appendix:separation}
This appendix is dedicated to establishing conditions for the \ac{kgd} to be separate non-stationary points, meaning that $\KGD_K(Q) = 0$ iff $Q$ is stationary, for all $Q \in \mathcal{P}(\mathbb{R}^d)$.
In \Cref{app: identity}, we establish that the \ac{kgd} provides a necessary condition for a distribution to be a stationary point of the objective $\eobj$, which we call the \emph{identity property}: $\KGD_K(Q) = 0$ if $Q$ is a stationary point. 
We then prove a general result (\Cref{thm:general-separation}) for the other direction: $\KGD_K(Q) = 0$ only if $Q$ is a stationary point in \Cref{app:gen-separation}. 
This result yields \Cref{thm:separation} and \Cref{prop: sufficient KSD separating} 
in the main text as corollaries, which are proved in \Cref{app:proof-of-separation-maintext} and \Cref{app:proof-of-separation-maintext-TI}. 

\subsubsection{Background Material and Definitions} 

This appendix is dedicated to introducing the concepts required for stating the results in \Cref{appendix:separation}. 

Define 
(i) $C_c(\mathbb{R}^d)$ 
(ii) $C_0(\mathbb{R}^d)$ 
and 
(iii) $C_b(\mathbb{R}^d)$ 
as the sets of continuous functions $f$ on $\mathbb{R}^d$ respectively satisfying the following: 
(i) $f$ has compact support, 
(ii) $f$ vanishes at infinity, and 
(iii) $f$ is bounded. 
For $F \in \{C_c(\mathbb{R}^d), C_0(\mathbb{R}^d), C_b(\mathbb{R}^d)\}$ and $m \geq 0$, 
we define 
\[
F^m \coloneqq \left\{
f \in C^m(\mathbb{R}^d) :
\partial^\alpha f \in F
\text{ for all multi-indices } \alpha
\text{ with } |\alpha| \leq m
\right\}.
\]
the set of $m$ times continuously differentiable functions where all $m$th partial derivatives belong to $F$. 
Note that these function spaces form vector spaces, and become \acp{tvs} with appropriate topological structures. 

Consider the case $m=1$; we will make our choices of \ac{tvs} structures explicit. 
Recall from \Cref{def: universal kernel} that we can equip $C_0^1(\mathbb{R}^d)$ with the norm 
    \begin{align*}
        \|f\|_{C_0^1} \coloneqq \sup_{x \in \mathbb{R}^d} (\|f(x)\| + \|\nabla f(x)\|).
    \end{align*}
For $C_b^1(\mathbb{R}^d)$, we will make use of the strict topology $\beta$~\citep{Buck_1958} and denote the corresponding \ac{tvs} by $C_b^1(\mathbb{R}^d)_\beta$ to emphasise the \ac{tvs} structure. 
For our purposes, it is convenient to first introduce the following generalisation of $C_b^1(\mathbb{R}^d)_\beta$. 
For continuous $\theta: \mathbb{R}^d \to [c, \infty)$ with $c > 0$, define a family 
\begin{equation*}
    \left\{
    f\mapsto \sup_{x\in \mathbb{R}^d} \verts{\gamma(x) \theta(x)f(x)},\ 
    f\mapsto \sup_{x\in \mathbb{R}^d} \verts{\gamma(x) \partial_i f(x)}: 
    \gamma \in C_0(\mathbb{R}^d),\ 1 \leq i \leq d
    \right\}
\end{equation*}
of seminorms. 
We define $C_{b,\theta}^1(\mathbb{R}^d)_\beta$ as the \ac{tvs} 
of functions $f:\mathbb{R}^d \to \mathbb{R}$ such that $\theta f \in C_b(\mathbb{R}^d)$ and $\partial_i f \in C_b(\mathbb{R}^d)$ for each $i \in \{1,\dots, d\}$, equipped with the topology generated by the above family of seminorms. 
Note that choosing a constant function $\theta$ recovers $C_b^1(\mathbb{R}^d)_\beta$. 
Similarly, we define 
$C_{0}^1(\mathbb{R}^d, \mathbb{R}^r)$ and 
$C_{b,\theta}^1(\mathbb{R}^d, \mathbb{R}^r)_\beta$, 
the space of $\mathbb{R}^r$-valued functions, by the respective $r$-fold product space 
of $C_{0}^1(\mathbb{R}^d)$ and $C_{b,\theta}^1(\mathbb{R}^d)_\beta$. 

In our proof, we work with the topological dual space of the above function spaces. 
For a given \ac{tvs} $\mathcal{V}$, we let $\mathcal{V}^*$ denote the topological dual space (consisting of continuous linear functionals on $\mathcal{V}$). 
A Radon measure $\mu$ is a continuous linear functional on $C_c(\mathbb{R}^d)$ (with respect to the inductive limit topology),\footnote{On a locally compact Polish space, all finite Borel measures are Radon (in a measure theoretic sense), and the linear functional definition here coincides and includes our problem set $\finitemomentspace{}$~\citep[see, e.g.][Chapter IV]{Bauer2001}.} 
and viewing it as a Schwartz distribution~\citep[see, e.g.][Part II]{Treves1967}, we denote its distributional derivative with respect to the $i$th coordinate by $\partial_i \mu: C_c^\infty(\mathbb{R}^d) \to \mathbb{R}$, 
which acts on a function $f \in C_c^{\infty}(\mathbb{R}^d)$ by $\partial_i \mu (f) = -\mu(\partial_i f)$.
For a continuous $g: \R^d \to \R$, we also define the left multiplication $g\mu \in C_c(\R^d)^*$ of a Radon measure $\mu$ with $g$ by 
$g \mu (f) = \mu(g f)$ for $f \in C_c(\R^d)$. 

Our particular interest is in the following space: 

\begin{definition}[Order $\leq1$ distributions given by finite Radon measures; {\citealt[Definition 8]{barp2024targeted}}] 
For $r \geq 1$, we let $\DL(\mathbb{R}^r)$ denote the vector space of continuous linear functionals on 
$C_{0}^1(\mathbb{R}^d, \mathbb{R}^r)$ or equivalently $C_{b}^1(\mathbb{R}^d, \mathbb{R}^r)_\beta$. 
We use the shorthand $\DL$ for $\DL(\mathbb{R})$. 
\end{definition}
\noindent
A functional $D$ belongs to $\DL$ iff 
it can be represented as $ \mu_0 + \sum_{j=1}^d \partial_i \mu_i$, where $\mu_0, \dots, \mu_d$ are finite Radon measures on $\mathbb{R}^d$. 
The motivation for considering $C_b^1(\mathbb{R}^d)_\beta$ is because of the dual characterisation given above (the dual of $C_b(\mathbb{R}^d)_\beta$ is given by the space of finite Radon measures, \citealt{Conway1965}), while allowing us to consider bounded functions; 
cf. \citep[][the proof of Theorem 24.4]{Treves1967}. 
Likewise, 
a functional $D$ belongs to $\DL(\mathbb{R}^r)$ iff it is given as $\sum_{j=1}^r D_j e^j$, 
where $D_j \in \mathcal{D}_{L^1}^1$, and
$(e^1, \dots, e^r)$ is the dual basis of the standard basis $(e_1, \dots, e_r)$ of $\mathbb{R}^r$. 

Our result requires characterising functionals using an \ac{rkhs}. 
The following definition formalises this idea: 
\begin{definition}[Embedding into \ac{rkhs} and characteristic kernels;  cf. {\citealt[Appendix C]{barp2024targeted}}]
Let $\rkhs{K}$ be the \ac{rkhs} associated to a kernel $K: \mathbb{R}^d \times \mathbb{R}^d \to \mathbb{R}^{d\times d}$
and let $\mathcal{V}$ be a \ac{tvs} of functions on $\mathbb{R}^d$ containing $\rkhs{K}$. 
We say that a linear functional $D$ on $\mathcal{V}$ \emph{embeds} into $\rkhs{K}$ if $D\rvert_{\rkhs{K}}$ is continuous on $\rkhs{K}$, 
i.e., if there is a function $\Phi_K(D) \in \rkhs{K}$ such that $D(v) = \la \Phi_K(D), v\ra_{\rkhs{K}}$. 
For a set $\mathcal{D}$ of embeddable linear functionals on $\mathcal{V}$, 
we say $K$ is \emph{$\mathcal{D}$-characteristic} or (\emph{characteristic to }$\cal{D}$) if 
the embedding map $\Phi_K: \mathcal{D} \to \rkhs{K}$ is injective, that is, for $D_1, D_2 \in \mathcal{D}$, we have $\Phi_K(D_1) \neq \Phi_K(D_2)$ whenever $D_1 \neq D_2$ (as functionals on $\mathcal{V}$). 
\end{definition}
\noindent A related notion is universality: 
\begin{definition}[Univesality]
Let $\mathcal{V}$ be a \ac{tvs} of functions on $\mathbb{R}^d$. 
An \ac{rkhs} $\rkhs{K}$ (or kernel $K)$ is universal to $\mathcal{V}$ if $\rkhs{K} \cap \mathcal{V}$ is dense in $\mathcal{V}$. 
\end{definition}

\subsubsection{Identity Property}
\label{app: identity}
This appendix is dedicated to establishing the identity property: $\KGD_K(Q) = 0$ whenever $Q$ is a stationary point. 

The following lemma provides sufficient conditions for the identity property:

\begin{lemma}[Identity property]\label{lem:identity}
    Let $Q$ be a stationary point of $\eobj$ with density $q$, 
    and $v: \R^d \to \R^d$ be differentiable. 
    Suppose $v \in \cL^1(\R^d, Q)$  and $\Op{Q}v \in \cL^1(Q)$. 
    Then, 
    \[
        \int \Op{Q}v(x)\; \dd Q(x) = 0. 
    \]
\end{lemma}
\begin{proof}
    By the self-consistency of $Q$, 
    we have $\nabla \log q(x) = \nabla\log q_0(x) - \vargrad \cL(Q)(x)$. 
    Thus, 
    \begin{align*}
        \int \Op{Q}v(x)\; \dd Q(x)
        &= \int \{\nabla \log q(x) \cdot v(x) + \nabla \cdot v(x)\}\; \dd Q(x)\\
        &= \int \frac{\nabla \cdot \bigl(q(x)v(x)\bigr)}{q(x)}\; \dd Q(x)\\
        &= \int \nabla \cdot \bigl(q(x)v(x)\bigr)\; \dd x.
    \end{align*}
    The conclusion follows from the divergence theorem \citep[][Theorem 3.2]{liu2025learning}. 
\end{proof}

\begin{coroll} [Sufficient condition]
\label{cor:suffcondsforzeromean}
Let $\theta: \R^d \to [c,\infty)$ for some $c > 0$, and $v: \R^d \to \R^d$ be differentiable. 
Assume $\nabla \Verts{\log q_0(x)} \lesssim \theta(x)$ and $\Verts{\vargrad \cL(Q)(x)} \lesssim \theta(x)$ for each $Q \in \dom(\cL)$. 
Then, if $\sup_{x\in \R^d} \Verts{\theta v} < \infty$, $\sup_{x\in \R^d} \Verts{\nabla v} < \infty$, 
then the integrability assumptions in \Cref{lem:identity} are satisfied for any stationary point $Q$ of $\eobj$. 
\end{coroll}

\begin{coroll}[Identity property for KGD]\label{cor:zero-KGD}
\label{cor: zeroKGD}
Let \Cref{asm:uniq,asm: diff loss,asm: q0 support} hold.
Let $K: \R^d \times \R^d \to \R^{d \times d}$ such that $\partial_{1,i}\partial_{2,i}K$ exists for each $i \in \{1,\dots, d\}$. 
Let $Q$ be a stationary point of $\eobj$. 
If $\rkhs{K} \subset \cL^1(\R^d, Q)$ and $\Op{Q}(\rkhs{K})\subset \cL^1(Q)$, 
then $\KGD_K(Q) = 0$. 
\end{coroll}

\subsubsection{General Sufficient Conditions for Separating Stationary Points}
\label{app:gen-separation}
In this section we aim to establish the following result: 

\begin{theorem}\label{thm:general-separation}
Let $K$ be a kernel as in \Cref{def:kgd}. 
Suppose the following conditions: 
(i) $\nabla \log q_0$ is continuous, and $\vargrad \cL(Q)$ continuous for each $Q \in \dom(\cL)$; 
(ii) For continuous $\theta: \mathbb{R}^d \to [c, \infty)$ with some constant $c > 0$, 
we have $\Verts{\nabla \log q_0(x)} \lesssim \theta(x) $ 
and $\Verts{\vargrad  \cL(Q)(x)}\lesssim \theta(x)$ for each $Q \in \dom(\cL)$; 
(iii) There exists an \ac{rkhs} $\rkhs{L}\subset \rkhs{K}$ with $L$ being $C^1_{b, \theta}(\mathbb{R}^d, \mathbb{R}^d)_\beta^*$-characteristic and $\rkhs{L}\subset {C^1_{b, \theta}(\R^d, \R^d)}$.
Then $\KGD_K(Q) = 0$ implies $Q$ is a stationary point. 
\end{theorem}

\paragraph{Proof Outline}
Consider a linear functional  
\begin{equation}
    D_Q \coloneqq Q \circ \Op{Q}: v \in (\mathbb{R}^d \to \mathbb{R}^d) \mapsto \int \Op{Q} (v) \dd Q \in \mathbb{R}, \label{eq:functional}
\end{equation}
defined on an appropriate set on differentiable vector fields $v$. 
The $\KGD$ can be seen as the dual norm 
of the functional $D_Q$ with respect to the \ac{rkhs} $\rkhs{K}$ (up to the integrability condition): 
\begin{align*}
    \KGD_P(Q) = \sup_{\|v\|_{\mathcal{H}_K} \leq 1}  | D_Q (v)|
\end{align*}
Thus, the separation of stationary points via \ac{kgd} can be cast as the separation of $D_Q$ from zero. 
In the following, we therefore consider the following question: `Does $D_Q = 0$ on the \ac{rkhs} $\rkhs{K}$ imply that $Q$ is a stationary point?'
We can address this problem by first identifying the class of functionals that contains $D_Q$, 
and $D_Q$ being a zero functional implies the stationarity of $Q$ (\Cref{lem:Dq-properties}). 
We then show that an appropriate \ac{rkhs} can uniquely embed the functionals in that class, and the \ac{kgd} therefore returns zero only when $Q$ is a stationary point (\Cref{lem:separation-equiv}). 
Our proof essentially follows the separation result for the \ac{ksd} by \citet{barp2024targeted}. 

\paragraph{Preliminary Results}

We begin with a lemma characterising the functional $D_Q$:

\begin{lemma}[Properties of $D_Q$]
\label{lem:Dq-properties}
Let $\theta: \mathbb{R}^d \to [c, \infty)$ be a continuous function where $c > 0$ 
and take $Q \in \dom(\cL)$. 
Suppose the following: 
(i) $\nabla \log q_0$ is continuous and $\Verts{\nabla \log q_0(x)} \lesssim \theta(x)$; 
(ii) $\vargrad \cL(Q)$ is continuous 
and $\Verts{\vargrad  \cL(Q)(x)}\lesssim \theta(x)$. %
Then, $D_Q$ in \eqref{eq:functional} is an element of $C_{b, \theta}^1(\mathbb{R}^d, \mathbb{R}^d)_\beta^*$. 
Moreover, $D_Q\rvert_{C_{b, \theta}^1(\mathbb{R}^d, \mathbb{R}^d)} \equiv 0$ implies $Q$ is a stationary point of $\eobj$. 
\end{lemma}
\begin{proof}
First note that the functional $D_Q$ is expressed as 
\begin{equation*}
    D_Q = \sum_{j=1}^d \{ (\genscore{P}{Q})_j Q - \partial_j Q \} e^j, 
\end{equation*}
where $(\genscore{P}{Q})_j$ denotes the $j$th component of $\genscore{P}{Q} = \nabla\log \genpdf{Q}$. 
For fixed $Q$, this expression is the same as the Langevin Stein operator, where $\genscore{P}{Q}$ is given by $\nabla \log p$ of some density function $p$. 
The first claim follows from the corresponding result for the Langevin Stein operator; see \citet[proof of Proposition 12]{barp2024targeted}. 

We now address the second claim. 
Suppose $D_Q = 0$  on $C_{b,\theta}^1(\mathbb{R}^d, \mathbb{R}^d)$. 
Since $C_c^{\infty}(\mathbb{R}^d)$ can be continuously embedded into $C_{b, \theta}^1(\mathbb{R}^d)_\beta$, 
by evaluating $D_Q$ at $v_{\varphi, i} = \varphi e_i$ for each $i \in \{1, \dots, d\}$ and arbitrary $\varphi \in C_c^\infty(\mathbb{R}^d)$, 
we obtain the distributional differential equation 
\begin{equation*}
    (\genscore{P}{Q})_i Q - \partial_i Q = 0,\ 1 \leq i \leq d. 
\end{equation*}
This implies 
\begin{align*}
    \partial_i \left( \genpdf{Q}^{-1}Q \right) = -\genpdf{Q}^{-1} \bigl( (\genscore{P}{Q})_i Q - \partial_i Q \bigr) = 0
\end{align*}
for each $i \in \{1, \dots, d\}$, viewed as a Schwartz distribution. 
Thus, the distribution $\genpdf{Q}^{-1}Q$ is translation invariant, and hence proportional to the Lebesgue measure $L$~\cite[Theorem VI of Chapter II]{Schwartz1978}, that is $Q \propto \genpdf{Q} L$. 
\end{proof}

\begin{lemma}[Zero $D_Q$ characterisation with \ac{rkhs}]
\label{lem:separation-equiv}
With $\rkhs{K} \subset C^1_{b,\theta}(\R^d, \R^d)$, 
in the setting of \Cref{lem:Dq-properties}, 
the following statements are equivalent: 
(i) $\KGD_K(Q)=0$ implies that $Q$ is a stationary point, and 
(ii) $K$ separates $0$ from non-zero functionals in $\{D_Q: Q \in \dom(\cL) \} \subset C_{b,\theta}^1(\R^d, \R^d)_\beta^*$. 
\end{lemma}

\begin{proof}
Note that by \Cref{lem:identity}, $D_Q  = 0 $ on $C_{b,\theta}^1(\R^d, \R^d)$ if $Q$ is a stationary point. 
As a result, $\{D_Q: Q \in \dom(\cL) \}$ contains the zero functional due to the assumed existence of a stationary point. 
Since $\rkhs{K} \subset C_{b,\theta}^1(\mathbb{R}^d, \mathbb{R}^d)$ implies $\rkhs{k_K^Q} \subset C_b({\mathbb{R}^d})$, 
by \citet[Proposition 6]{barp2024targeted}, any $Q \in \dom(\cL)$ embeds into $\rkhs{k_K^Q}$, and thus 
$D_Q$ embeds into $\rkhs{K}$, which yields $\KGD_K(Q) = \Verts{\Phi_K(D_Q)}_{\rkhs{K}}$.

Now we examine the equivalence between the two statements. 
Suppose $D_Q$ is mapped to $\Phi_K(D_Q) = 0$. 
From the above discussion, this implies $\KGD_K(Q)=0$, and $Q$ is a stationary point under the hypothesis~(i), hence necessarily $D_Q = 0$. 
Suppose the statement~(ii) is true. Then, $\KGD_K(Q) = 0$ implies $D_Q=0$ since under the separation assumption, only the zero functional is mapped to the zero in the \ac{rkhs} via the embedding $\Phi_K$. 
By \Cref{lem:Dq-properties}, it holds that $Q$ is a stationary point. 
\end{proof}

\begin{coroll}[Stationarity characterisation with $C_{b,\theta}^1$-kernels]\label{cor:KGD-stationarity}
In the setting of \Cref{lem:separation-equiv}, if $K$ is $C_{b,\theta}^1(\R^d, \R^d)_\beta^*$-characteristic, then
$\KGD_K(Q)=0$ iff $Q$ is a stationary point of $\eobj$ in \eqref{eq: objective}. 
\end{coroll}
\begin{proof}
    The `if' direction can be shown by combining \Cref{cor:suffcondsforzeromean} with $\rkhs{K} \subset C^1_{b,\theta}(\R^d, \R^d)$. 
    For the `only if' direction, since by definition $\KGD_K(Q) = \infty$ for $Q \notin \dom(\cL)$, we consider distributions in $\dom(\cL)$. 
    The conclusion follows from \Cref{lem:separation-equiv}  and $C_{b,\theta}^1(\R^d, \R^d)_\beta^*$-characteristicity. 
\end{proof}

\paragraph{Proof of \Cref{thm:general-separation}}
We are ready to prove \Cref{thm:general-separation}:

\begin{proof}[Proof of \Cref{thm:general-separation}]
Suppose $Q$ is not a stationary point. 
Since $L$ is assumed to be $C_{b,\theta}^1(\R^d, \R^d)_\beta^*$-characteristic, 
from \Cref{lem:separation-equiv}, we have $\KGD_L(Q) > 0$. 
This implies there exists $v \in \rkhs{L}$ such that $D_Q(v) \neq 0$. 
By the assumption $\rkhs{L} \subset \rkhs{K}$, we have $v \in \rkhs{K}$, and therefore $\KGD_K(Q) > 0$. 
\end{proof}

\subsubsection{Proof of \Cref{thm:separation}}
\label{app:proof-of-separation-maintext}
\begin{proof}[Proof of \Cref{thm:separation}]
\Cref{thm:separation} is an instantiation of \Cref{cor:KGD-stationarity} (of \Cref{lem:separation-equiv}) above, which requires a $C_{b,\theta}^1(\R^d, \R^d)_\beta^*$-characteristic kernel. 
The result follows from \Cref{lem:tilted-kernel-universality} and 
the equivalence between $\DL(\R^d)$-characteriscity and $C_0^1(\R^d, \R^d)$-universality in \Cref{lem:DL1-characteristic-equiv}. 
\end{proof}

\begin{lemma}[{Lemma 11 of \citealt{barp2024targeted}}]\label{lem:tilted-kernel-universality}
Let $\theta \in C^1(\mathbb{R}^d): \mathbb{R}^d \to [c, \infty)$ for $c > 0$  
such that $1/\theta$, $\nabla (1/\theta)$ are bounded. 
If $K$ is universal to $C_b^1(\mathbb{R}^d, \mathbb{R}^d)_\beta$, 
then the tilted kernel $x,y \mapsto K(x, y)/\{ \theta(x) \theta(y)\}$ is universal to $C_{b,\theta}^1(\mathbb{R}^d, \mathbb{R}^d)_\beta$. 
\end{lemma}

\begin{lemma}[$\DL(\R^d)$-characteristic kernels]
\label{lem:DL1-characteristic-equiv}
Let $K$ be a kernel with $\rkhs{K} \subset C_{b}^1(\mathbb{R}^d, \mathbb{R}^d)$. 
Then, $K$ is universal to 
$C_b^1(\mathbb{R}^d, \mathbb{R}^d)_\beta$ iff K is $\DL(\mathbb{R}^d)$-characteristic. 
Moreover, for a kernel $K$ with $\rkhs{K} \subset C_0^1(\mathbb{R}^d, \mathbb{R}^d)$, 
$C_0^1(\mathbb{R}^d, \mathbb{R}^d)$-universality is equivalent to $\DL(\mathbb{R}^d)$-characteristicity. 
\end{lemma}
\begin{proof}
Note that $\rkhs{K}$ is continuously embedded into $C_{b}^1(\mathbb{R}^d, \mathbb{R}^d)_\beta$, 
since for any $\gamma \in C_0(\mathbb{R}^d)$ and $v \in \rkhs{K}$, we have 
\begin{align*}
    \sup_{x\in \mathbb{R}^d} \Verts{ \gamma(x) v(x)} \leq \left(\sup_{x\in\mathbb{R}} \sqrt{\Verts{K(x, x)}}\verts{\gamma(x)}\right)\Verts{v}_{\rkhs{K}}
\end{align*}
and 
\begin{align*}
    \sup_{x\in \mathbb{R}^d} 
    \Verts{ \gamma(x) \partial_i v(x)} \leq 
    \left(\sup_{x\in\mathbb{R}} \sqrt{\Verts{\partial_{1, i}\partial_{2, i}K(x, x)}}\verts{\gamma(x)}\right)\Verts{v}_{\rkhs{K}}; 
\end{align*}
and $\sup_{x\in\mathbb{R}^d}\Verts{K(x, x)} < \infty$ and $\sup_{x\in \mathbb{R}^d}\Verts{\partial_{1, i}\partial_{2, i}K(x, x)} < \infty$ for each $i \in \{1,\dots,d\}$ from $\rkhs{K} \subset C_{b}^1(\mathbb{R}^d, \mathbb{R}^d)$~\citep[Lemma 3]{barp2024targeted}. 
By \citet[Theorem 6]{SimonGabriel2020}, 
the $C_b^1(\mathbb{R}^d, \mathbb{R}^d)_\beta$-universality of $\rkhs{K}$ is then equivalent to $\DL(\mathbb{R}^d)$-characteristicity,
and $\DL(\mathbb{R}^d)$-characteristicity is equivalent to $K$ being universal to $C_0^1(\mathbb{R}^d, \mathbb{R}^d)$ for $\rkhs{K} \subset C_0^1(\mathbb{R}^d, \mathbb{R}^d)$ by the same theorem and \citet[Lemma 8]{barp2024targeted}. 
\end{proof}

\subsubsection{Proof of \Cref{prop: sufficient KSD separating}}
\label{app:proof-of-separation-maintext-TI}
\begin{proof}[Proof of \Cref{prop: sufficient KSD separating}]
    The first claim is an application of \Cref{cor: zeroKGD}. 
    We address the second claim. 
    \citet[Theorem 3.8]{barp2024targeted} states that for our choice of kernel $k$, 
    there exist a translation invariant $\DL$-characteristic kernel $k_s$ with $\rkhs{k_s} \subset C_b^1(\R^d)$ and for each $t > 0$, 
    a positive definite function $f$ with $1/f \in C^1(\R^d)$, $\max(\verts{f(x)}, \Verts{\nabla f(x)}) \lesssim \exp(-t \sum_{i=1}^d \sqrt{\verts{x_i}})$ 
    such that the kernel $k_f(x, y) = f(x) k_s(x, y) f(y)$ satisfies $\rkhs{k_f} \subset \rkhs{k}$. 
    Note that, since $\rkhs{k_f} \subset C_0^1(\R^d)$, $k_f$ is $C_0^1(\R^d)$-universal for any such $f$ by \Cref{lem:tilted-kernel-universality}, \citet[Lemma 8]{barp2024targeted}, and \citet[Theorem 6]{SimonGabriel2020}. 
    \citet[Proposition 11]{barp2024targeted} in turn shows that $k_f \idmat$ is $C_0^1(\R^d, \R^d)$-universal. 
    The claim follows from \Cref{thm:general-separation} with $L = k_f \idmat$, along with \Cref{lem:tilted-kernel-universality} and \Cref{lem:DL1-characteristic-equiv}, by taking appropriately large $t$ to dominate the growth of $\theta$ with $f$. 
\end{proof}

\subsection{Continuity}
\label{app: continuity proofs}

This appendix is dedicated to establishing continuity of \ac{kgd}, proving the proofs for theoretical results appearing in \Cref{sec:continuity}, as well as the sample complexity result stated in \Cref{prop: sample complexity}.

\subsubsection{Preliminary Results}
\label{sec: interaction energy proof}

First we present preliminary results concerning the integral form of the loss function $\cL$ and continuity of the \ac{kgd}:

\begin{coroll}[Variational gradient of interaction functional]\label{cor:wgrad4intenergy}
Let $\polyorder\in [0,\infty)$. 
In the setting of \Cref{lem:fv-interact} with $r \geq 2$, 
suppose $w \in C^1(\prod_{i=1}^r \R^d)$, 
$\verts{w(x_1 \dots, x_r)} \lesssim \prod_{i=1}^r (1 + \Verts{x_i}^\polyorder)$, 
and $\Verts{\nabla w(x_1 \dots, x_r)} \lesssim \prod_{i=1}^r (1 + \Verts{x_i}^\polyorder)$. 
Assume $Q \in \finitemomentspace{\polyorder}$. 
Then 
     \[
        \mathfrak{D}_{\mathrm{V}} \mathcal{W}(Q)(x) =  \sum_{i=1}^{r} \int \nabla_i w(y_1, \dots, y_r) \; \dd Q_{x, i}(y_1, \dots, y_r). 
    \]
for any $Q \in \finitemomentspace{\polyorder}$. 
If $\polyorder= 0$, that is $w \in C_b^1(\prod_{i=1}^r \R^d)$, then $\vargrad \mathcal{W}(Q)$ is given as above for any $Q \in \finitemomentspace{}$. 
\end{coroll}
\begin{proof}
    Apply \citet[Corollary A.4]{Dudley_1999} so that the integral-derivative interchange 
    \[
        \nabla \int w(y_1, \dots, y_r) \; \dd Q_{x, i}(y_1, \dots, y_r)
        = \int \nabla_i w(y_1, \dots, y_{r}) \; \dd Q_{x, i}(y_1, \dots, y_{r}),
    \]
    is justified. 
\end{proof}

For presentational purposes, it is convenient to first state a more general assumption and then to show (in \Cref{subsec:proof-continuity}) that this assumption is implied by \Cref{asm: q0 support,asm: for continuity}.
Let 
\begin{align*}
\pi_i: (x_1, \dots, x_r) \mapsto (x_1, \dots, \underbrace{x_r}_{i \text{th component}}, x_{i}, \dots, x_{r-1})
\end{align*}
for $r \geq 2$ and $1 \leq i \leq r$, and $\pi_i: x\mapsto x$ if $r= 1$. 

\begin{assumption}[General assumptions for continuity of \ac{kgd}]\label{asm:continuity}
The following hold:
\begin{enumerate}[label=(\roman*)]
    \item \emph{(Loss gradient)} The loss $\cL(Q)$ has a variational gradient of the form 
    \begin{equation*}
        \vargrad  \cL(Q)(x) =  \sum_{i=1}^{r} \int \nabla_i w(y_1, \dots, y_r) \; \dd Q_{x, i}(y_1, \dots, y_r)
    \end{equation*}
    for some continuously differentiable $w:\prod_{i=1}^r\mathbb{R}^{d} \to \mathbb{R}$, 
    where  
    \[ 
    Q_{x, i} \coloneqq Q\otimes \dots \otimes Q \otimes \underbrace{\delta_{x}}_{i\mathrm{th\ component}} \otimes Q \otimes \dots \otimes Q.
    \] \label{item:vargrad-form-cont}
    \item \emph{(Kernel)} $(x,y) \mapsto \partial_{1,i}\partial_{2, i}K(x,y)$ is continuous for each $i \in \{1,\dots,d\}$. 
    \end{enumerate}
    In addition, let
    \begin{align*}
        f_1^{(i)}(x_1,\ldots,x_{r+1})
        &\coloneqq-\nabla_i w(\pi_i(x_1,\ldots,x_r))\\
        &\qquad\cdot\left[
        \nabla_2\cdot K(x_r,x_{r+1})
        +K(x_r,x_{r+1})\nabla \log q_0(x_{r+1})
        \right]\\
        &\hphantom{\coloneqq}\quad-\nabla_i w(\pi_i(x_1,\ldots,x_{r-1},x_{r+1}))\\
        &\qquad\cdot\left[
        \nabla_1\cdot K(x_r,x_{r+1})
        +K(x_{r+1},x_r)\nabla \log q_0(x_r)
        \right]\\
        f_2^{(i,j)}(x_1, \dots, x_{r}, y_1, \dots, y_{r}) & \coloneqq  (\nabla_i w) (\pi_i(x_1,\dots, x_r)) \cdot K(x_r, y_r)(\nabla_j w)(\pi_j(y_1, \dots, y_r)),\\
        h_{q_0}(x, y) &\coloneqq \nabla \log q_0(x) \cdot K(x, y) \nabla \log q_0(y) + \nabla \log q_0(x) \cdot \nabla _2 \cdot K(x, y)\\
        &\hphantom{h_{q_0}(x, y) \coloneqq} + \nabla \log q_0(y) \cdot \nabla _1 \cdot K(x, y) + \nabla_1\cdot \nabla_2 \cdot  K(x, y).
    \end{align*}
    and let the following growth conditions hold: 
    \begin{enumerate}
    \setcounter{enumi}{2}
        \item[(iii)] $\verts{f_1^{(i)}(x_1, \dots, x_{r+1})}\lesssim  \prod_{i=1}^{r+1} (1+\Verts{x_i}^\polyorder)$
        \item[(iv)] 
        $\verts{f_2^{(i,j)}(x_1,\dots, x_{2r})} \lesssim  \prod_{i=1}^{2r} (1+\Verts{x_i}^\polyorder)$
        \item[(v)] 
        $\sqrt{h_{q_0}(x, x)} \lesssim 1 + \Verts{x}^\polyorder$
    \end{enumerate}
\end{assumption}

\begin{lemma}[KGD continuity]\label{prop:kgd-cont-abst}
    Let \Cref{asm: q0 support,asm:contdiff,asm:continuity} hold. 
    Then $\KGD_K(Q_\seqidx) \to \KGD_K(Q)$ whenever $Q_\seqidx \toL{\polyorder}Q$. 
\end{lemma}
\begin{proof}[Proof of \Cref{prop:kgd-cont-abst}]
First note that $Q_\seqidx \toL{\polyorder} Q$ implies $\polyorder$th moment uniform integrability. 
Uniform integrability allows us to assume $Q_\seqidx \in \finitemomentspace{\polyorder}$ without loss of generality for any $\seqidx \geq 1$, since $Q_\seqidx \in \finitemomentspace{\polyorder}$ eventually for sufficiently large $n$. 
Let us denote $\tilde{\mu}(\dd x) = (1+\Verts{x}^\polyorder) \mu (\dd x)$ for $\mu \in \finitemomentspace{\polyorder}$. 
In the proof below, we make use of weak convergence of finite measures $\tilde{Q}_\seqidx$ to $\tilde{Q}$, which is equivalent to $Q_\seqidx \toL{\polyorder} Q$. 

Suppose for the moment that $k_K^Q$ is given by an integral 
\begin{equation}
k_K^Q(x, y) = \int f(x_1,...,x_{\ell+2}) \ \dd (Q^{\otimes \ell}\otimes \delta_x\otimes \delta_y)(x_1,...,x_{\ell+2})  \label{eq: integral form of kPQ}
\end{equation}
of some continuous function  $f:\prod_{i=1}^{\ell+2}\mathbb{R}^{d} \to \mathbb{R}$ and some $\ell \in \mathbb{N}$ such that 
\begin{equation}
    \verts{f(x_1, \dots, x_{\ell+2})} \lesssim \prod_{i=1}^{\ell+2} \bigl(1+\Verts{x_i}^\polyorder\bigr),
    \label{eq:required-growth-inproof}
\end{equation}
which in particular implies $\Op{Q}(\mathcal{H}_K) \subset\cL^1(Q)$ for any $Q \in \finitemomentspace{\polyorder}$. 
From \Cref{lem: computable}, we could then write \ac{kgd} as  
\begin{equation}
\KGD_K(Q)^2 = \iint k_K^Q(x,x') \ \dd Q(x) \dd Q(x') 
=  \int f\ \dd Q^{\otimes (\ell+2)}, \label{eq:kgd-asint-intermediate}
\end{equation}
from which the continuity of \ac{kgd} would follow since, in the limit $\seqidx \to \infty$, 
\begin{align*}
\KGD_K(Q_\seqidx)^2 & = \int f\ \dd Q_\seqidx^{\otimes (\ell+2)}\\
&= \int \prod_{i=1}^{\ell+2} \bigl(1+\Verts{x_i}^\polyorder\bigr)^{-1} f(x_1, \dots, x_{\ell+2})\; \dd \tilde{Q}_\seqidx^{\otimes (\ell+2)}(x_1,\dots,x_{\ell+2})\\
& \hspace{-10pt} \to \int \prod_{i=1}^{\ell+2} \bigl(1+\Verts{x_i}^\polyorder\bigr)^{-1} f(x_1,\dots, x_{\ell+2})\; \dd \tilde{Q}^{\otimes (\ell+2)} (x_1,\dots, x_{\ell+2}) = \KGD_K(Q)^2. 
\end{align*}
The above convergence is due to the boundedness and the continuity of the integrand and because the finite measure $\tilde{Q}_\seqidx$ converges weakly to $\tilde{Q}$, 
along with the fact that the operation ($Q\mapsto Q^{\otimes (\ell+2)}$) of taking the product measure of (finitely many) finite measures is weakly continuous~\citep[Theorem 3.3, p.47]{Berg_1984}. 
Therefore, the rest of the proof is devoted to obtaining the expression \eqref{eq: integral form of kPQ}. 

According to \eqref{eq: PQ kernel}, 
\begin{equation}
    \begin{aligned}
    k_K^Q(x, y) &= \genscore{P}{Q}(x)\cdot K(x,y)\genscore{P}{Q}(y) +  \genscore{P}{Q}(y)\cdot \bigl( \nabla_{1} \cdot K(x,y)\bigr)\\
        & \qquad +\genscore{P}{Q}(x) \cdot \bigl(\nabla_{2} \cdot K(x,y)\bigr) + \nabla_{2}\cdot \nabla_{1} \cdot K(x,y).
    \end{aligned}\label{eq:gen-Stein-kernel-alt}
\end{equation}
Substituting $\genscore{P}{Q} =\nabla \log q_0 - \vargrad  \mathcal{L}(Q)$ into $k_K^Q$ gives 
\begin{align*} 
k_K^Q(x,y) 
&= \vargrad \cL(Q)(x)\cdot K(x,y)\vargrad \cL(Q)(y) - \vargrad \cL(Q)(y)\cdot \bigl( \nabla_{1} \cdot K(x,y)\bigr)\\
& \qquad -\vargrad \cL(Q)(x) \cdot \bigl(\nabla_{2} \cdot K(x,y)\bigr) + h_{q_0}(x, y) \\
&\qquad 
-\vargrad \cL(Q)(x)\cdot K(x,y)\nabla \log q_0(y)
-\vargrad \cL(Q)(y)\cdot K(y,x)\nabla \log q_0(x).
\end{align*}
where $h_{q_0}$ is the Langevin Stein kernel of $Q_0$~\citep[Eq. 5]{barp2024targeted}. 
Using the assumed form of $\vargrad  \cL(Q)$, 
we obtain
\begin{align*}
   k_K^Q(x,y) 
   &= h_{q_0}(x,y) +  \sum_{i=1}^r \int f_1^{(i)} \; \dd (Q^{\otimes (r-1)}\otimes\delta_x\otimes \delta_y )\\
   &\qquad + \sum_{i,j=1}^r \int f_2^{(i,j)}\; \dd (Q^{\otimes (r-1)}\otimes\delta_x\otimes Q^{\otimes (r-1)} \otimes \delta_y) .
\end{align*}
By \Cref{asm:continuity}, the functions $h_{q_0}$, $f^{(i)}_1$ and  $f^{(i,j)}_2$ are continuous; 
moreover, with $h_{q_0}(x,y) \leq \sqrt{h_{q_0}(x,x)}\sqrt{h_{q_0}(y, y)}$, 
the function 
\begin{equation}
\begin{aligned}
& \hspace{-40pt} f(x_1,\dots, x_r, y_1, \dots, y_r) \\
&= h_{q_0}(x_r, y_r) + \sum_{i=1}^rf_1^{(i)}(x_1, \dots, x_r, y_r) + \sum_{i,j=1}^r f_2^{(i,j)}(x_1,\dots, x_r, y_1, \dots, y_r)
\end{aligned}\label{eq:continuity-f-form}
\end{equation}
meets the required growth condition \eqref{eq:required-growth-inproof} with $\ell = 2(r-1)$. 
Thus, we have shown that $k_K^Q$ can be expressed as in \eqref{eq: integral form of kPQ}, as desired. 
\end{proof}

\begin{coroll}[Continuity for composite loss]\label{cor:continuity-composite}
Fix $\polyorder\in [0,\infty)$. 
Suppose \Cref{asm:continuity} holds with \Cref{item:vargrad-form-cont} replaced by the following form of loss: 
\begin{equation}
    \cL(Q) = 
    \begin{cases}
        &\psi \left( \mathcal{W}(Q) \right) \text{ if } Q \in \finitemomentspace{\polyorder}\\
        & \infty \text{ otherwise }, 
    \end{cases}\label{eq:composite-loss}
\end{equation}
where $\psi \in C^1(U, \R)$ with open $U \subset \R$, 
and 
$\mathcal{W}(Q) = \int w\; \dd Q^{\otimes r}$ is defined as in \Cref{cor:wgrad4intenergy} (cf. \Cref{lem:fv-interact}), assumed to take values in $U$ on $\finitemomentspace{\polyorder}$. 
Then, together with \Cref{asm: q0 support}, 
we have 
$\KGD_K(Q_\seqidx) \to \KGD_K(Q)$ whenever $Q_\seqidx \toL{\polyorder}Q$. 
\end{coroll}
\begin{proof}[Proof of \Cref{cor:continuity-composite}]
Note that by \Cref{cor:wgrad4intenergy} and the chain rule, 
\begin{align*}
    \mathfrak{D}_{\mathrm{V}}(Q)(x) = \psi'\left(\mathcal{W}(Q)\right) \cdot  \sum_{i=1}^{r} \int \nabla_i w(y_1, \dots, y_r) \; \dd Q_{x, i}(y_1, \dots, y_r)
\end{align*}
for $Q \in \finitemomentspace{\polyorder}$. 
That is, the variational gradient is in the same form as \Cref{item:vargrad-form-cont} in \Cref{asm:continuity} up to the multiplicative factor 
$\psi'\left(\mathcal{W}(Q)\right)$. 
By the continuity of derivative $\psi'$ and the growth assumption on $w$, we have $\psi'(\mathcal{W}(Q_n)) \to \psi'(\mathcal{W}(Q))$ whenever $Q_n \toL{\polyorder} Q$. 
We can then proceed as in the proof of \Cref{prop:kgd-cont-abst}, 
since $k_K^Q$ has a similar form, with the only difference being multiplicative factors given by $Q \mapsto \psi'\left(\mathcal{W}(Q)\right)$. 
\end{proof}

\begin{remark}[Relation to \Cref{prop:continuity}]\label{rem:composite-loss}
The above corollary states that 
\Cref{prop:continuity} holds if the loss $\cL$ is given in the composite form \eqref{eq:composite-loss}, if the following additional assumptions hold: 
(a) $\beta \leq \polyorder$ (see \Cref{asm: for continuity}) and (b) $w(x_1, \dots, x_r) \lesssim \prod_{i=1}^r (1 + \Verts{x_i}^\polyorder)$. 
\end{remark}

\subsubsection{Proof of \Cref{prop:continuity}}\label{subsec:proof-continuity}
\begin{proof}[Proof of \Cref{prop:continuity}]
We make use of \Cref{prop:kgd-cont-abst}. 
Since \Cref{asm: q0 support} is supposed, we just need to check that \Cref{asm:continuity} is satisfied.
Of the items in \Cref{asm:continuity}, the first two are shared with \Cref{asm: for continuity},
and thus it remains to check the growth conditions in \Cref{asm:continuity}. 
In the following, we examine the growth of each term of $f$ in \eqref{eq:continuity-f-form} under \Cref{asm: for continuity}. 

First we address $f^{(i)}_1$. 
For this purpose, we define $K_1:\mathbb{R}^{d}\times \mathbb{R}^{d} \to \mathbb{R}^{d\times d }$ by 
\begin{align}
    [K_1(a, b)]_{i,j} \coloneqq \partial_{1,i}K_{ij}(a,b)
    \leq \sqrt{ e_i \cdot \partial_{i, 1}\partial_{i,2}K(a, a) e_i } \sqrt{ e_j \cdot K(b,b)e_j  }
    \label{eq:K1-upperbound}
\end{align}
where $e_i$ denotes the $i$th canonical basis vector in $\mathbb{R}^d$, 
and the inequality is due to the derivative reproducing property~\citep[Lemma 4]{barp2024targeted}. 
Now note that 
\begin{align*}
 & \verts{f_1^{(i)}(x_1,\dots, x_{r+1})} \\
 &\leq d \Verts{(\nabla_i w)(\pi_i(x_1,\dots, x_r))}_2 \sqrt{\Verts{K(x_r, x_r)}_{\mathrm{op}, 2}} \cdot \sum_{i=1}^d \sqrt{\Verts{\partial_{1,i}\partial_{2,i}K(x_{r+1}, x_{r+1})}_{\mathrm{op}, 2}}\\ 
 &\hphantom{\leq} + d \Verts{(\nabla_i w)(\pi_i(x_1,\dots, x_{r+1}))}_2 \sqrt{\Verts{K(x_{r+1}, x_{r+1})}_{\mathrm{op}, 2}} \cdot \sum_{i=1}^d \sqrt{\Verts{\partial_{1,i}\partial_{2,i}K(x_r, x_r)}_{\mathrm{op}, 2}}\\
  &\hphantom{\leq} + \Verts{(\nabla_i w)(\pi_i ((x_1,\dots, x_r))}_2 \sqrt{\Verts{K(x_r, x_r)}_{\mathrm{op}}}  \sqrt{\Verts{K(x_{r+1}, x_{r+1})}_{\mathrm{op}}}  \Verts{\nabla \log q_0(x_{r+1})}_2\\
  &\hphantom{\leq} + \Verts{(\nabla_i w)(\pi_i ((x_1,\dots, x_{r+1}))}_2 \sqrt{\Verts{K(x_r, x_r)}_{\mathrm{op}}}  \sqrt{\Verts{K(x_{r+1}, x_{r+1})}_{\mathrm{op}}}  \Verts{\nabla \log q_0(x_{r})}_2,
\end{align*}
where the operator norm $\Verts{\cdot}_{\mathrm{op}, 2}$ is taken with respect to the 2-norm. 
Since $\Verts{\cdot}_2 \lesssim \Verts{\cdot}$ and $\Verts{\cdot}_{\mathrm{op}, 2} \lesssim \Verts{\cdot}_{\mathrm{op}}$, 
the claim for $f_1^{(i)}$ then follows by substituting the growth conditions in \Cref{asm: for continuity} into this estimate. 
Similarly, by the reproducing property, the function $f^{(i,j)}_2$ is evaluated as 
\begin{align*}
    &\verts{f_2^{(i,j)}(x_1, \dots, x_r, y_1, \dots y_r)}\\
    &\leq \sqrt{\Verts{K(x_r, x_r)}_{\mathrm{op}, 2}} \Verts{(\nabla_i w)(\pi_i (x_1, \dots, x_r))}_2 
    \cdot \sqrt{\Verts{K(y_r, y_r)}_{\mathrm{op}, 2}}  \Verts{(\nabla_j w)(\pi_j (y_1, \dots, y_r))}_2
\end{align*}
and the Stein kernel $h_{q_0}$ can also be evaluated as 
\begin{align*}
&h_{q_0}(x,x)\\
&\leq \Verts{\nabla \log q_0(x)}^2_2 \Verts{K(x,x)}_{\mathrm{op}, 2} +  d^2 \sum_{i=1}^{d} \Verts{\partial_{1, i}\partial_{2, i}K(x, x)}_{\mathrm{op}, 2}\\
&\hphantom{\leq} \quad + 2d^2 \Verts{\nabla \log q_0(x)}_2
\sqrt{\Verts{K(x,x)}_{\mathrm{op}, 2}}
\sum_{i=1}^{d} \sqrt{\Verts{ \partial_{1,i}\partial_{2,i}K(x, x) }_{\mathrm{op}, 2}},
\end{align*}
so that the result follows by again substituting the growth conditions in \Cref{asm: for continuity} into these estimates. 
\end{proof}

\subsubsection{Proof of \Cref{prop:kernel-conds-for-continuity}}
\label{app: sufficient kernel proof}

\begin{proof}[Proof of \Cref{prop:kernel-conds-for-continuity}]
\Cref{asm: kernel condition 2} is obvious. 
Since 
    \begin{align*}
     K(x,y)=\weight_{\polyorder-\beta}(x)\weight_{\polyorder-\beta}(y)L(x,y)+\weight_{\polyorder-\beta}(x) \weight_{\polyorder-\beta}(y)\bar{k}_{\mathrm{lin}}(x,y)\idmat,
    \end{align*}
we just need to check each term in the sum satisfies \Cref{asm: kernel condition 1,asm: kernel condition 3}. 
\Cref{asm: kernel condition 1} follows from $(\polyorder-\beta)$-growth of $\weight_{\polyorder-\beta}$ and the boundedness of $x \mapsto \Verts{L(x,x)}_\mathrm{op}$ and $x \mapsto \bar{k}_{\mathrm{lin}}(x,x)$. 
For \Cref{asm: kernel condition 3}, note that we have 
\begin{align*}
\Verts{\partial_{i, 1} L(x, y)}_{\mathrm{op}}
&= \sup_{\Verts{u} = \Verts{v} = 1} \verts{u\cdot \partial_{1,i}L(x, y)v}\\ %
&\leq \sup_{\Verts{u} = \Verts{v} = 1} \sqrt{ u \cdot \partial_{1, i}\partial_{2, i}L(x, x) u} \sqrt{ v \cdot L(y, y) v  } \\
&= \sqrt{\Verts{\partial_{1, i}\partial_{2, i}L(x, x)}_{\mathrm{op}}} \sqrt{\Verts{L(y, y)}_{\mathrm{op}}}. 
\end{align*}
Thus, for any $z \in \mathbb{R}^d$, 
\begin{align*}
     &\left. \Verts{ \partial_{1, i} \partial_{2, i}\left\{\weight_{\polyorder-\beta}(x)\weight_{\polyorder-\beta}(y)L(x,y)\right\} }_{\mathrm{op}}\right\rvert_{x=z,\ y=z}\\ 
     &\leq \{\partial_{i} \weight_{\polyorder-\beta}(z)\}^2 \Verts{L(z, z)}_{\mathrm{op}}
     + 2\left(\partial_{i} \weight_{\polyorder-\beta}(z)\right)\weight_{\polyorder-\beta}(z)    
    \sqrt{\Verts{\partial_{1, i}\partial_{2, i}L(z, z)}_{\mathrm{op}}} \sqrt{\Verts{L(z, z)}_{\mathrm{op}}} \\ 
     &\hphantom{=}\ 
     + \weight_{\polyorder-\beta}(z)^2\Verts{ \partial_{1, i} \partial_{2, i}L(z,z)}_{\mathrm{op}}.
    \end{align*}
Since $\verts{\partial_i \weight_{\polyorder-\beta}(x)} \leq \verts{\polyorder-\beta} / (2c)^{-1} \cdot \weight_{\polyorder-\beta}(x)$, using the boundedness of 
$x \mapsto \Verts{L(x,x)}_\mathrm{op}$ and 
$x \mapsto \Verts{\partial_{1, i}\partial_{2, i}L(x,x)}_\mathrm{op}$, 
each term is shown to be of $2(\polyorder-\beta)$-growth, and hence \Cref{asm: kernel condition 3} holds. The derivation for the second term is identical with $L$ replaced with $\bar{k}_{\mathrm{lin}}\idmat$. 
\end{proof}

\subsubsection{Proof of \Cref{prop: sample complexity}}
\label{app: sample complexity}

\begin{proof}[Proof of \Cref{prop: sample complexity}]
    From \eqref{eq:kgd-asint-intermediate}, the \ac{kgd} writes as
    \begin{align*}
        \KGD_K^2(Q_n) = \int f \ \dd Q_n^{\otimes (2r)} = \frac{1}{n^{\ell+2}} \sum_{i_1 = 1}^n \cdots \sum_{i_{2r} = 1}^n f(x_{i_1} , \dots , x_{i_{\ell+2}}) ,
    \end{align*}
    where the form of $f$ is given in \eqref{eq:continuity-f-form}.
    This we recognise as a non-symmetric $V$-statistic; see \Cref{app: v-stats}.
    As explained in the proof of \Cref{prop:continuity}, the growth conditions in \Cref{asm:continuity} hold, so that $f$ satisfies the pointwise bound $\verts{f(x_1,\dots,x_{2r})} \leq C \prod_{j=1}^{2r} (1 + \Verts{x_j}^\polyorder)$ for some $C < \infty$ and all $(x_1,\dots,x_{2r}) \in \prod_{j=1}^{2r} \R^d$, including tuples with repeated arguments.
    Thus the envelope condition \eqref{eq: v-stat envelope} of \Cref{cor: nonsymmetric V} in \Cref{app: v-stats} holds with $h(x) = 1 + \Verts{x}^\polyorder$, and $\int h^2 \; \dd Q < \infty$ since $Q \in \finitemomentspace{2\polyorder}$.
    Thus asymptotic normality is established. 
\end{proof}

\subsubsection{Asymptotic Normality of Non-Symmetric $V$-Statistics}
\label{app: v-stats}

This appendix is dedicated to the asymptotic behaviour of non-symmetric $V$-statistics, establishing \Cref{cor: nonsymmetric V} which was used in the proof of \Cref{prop: sample complexity}.
Though we suspect these arguments are standard, we have not found a self-contained reference that can be easily cited, so for completeness both the statements and proofs will be presented.
The strategy is to relate the $V$-statistic to a $U$-statistic, whose asymptotic normality is classical.

To state the following in a compact manner, let $x_{1:m} = (x_1,\dots,x_m)$ in shorthand.
Let $\Pi(n,p)$ denote the collection of all $p$-tuples $(i_1 , \dots , i_p) \in \{1,\dots,n\}^p$ with all components distinct.
Interpret $\mathcal{N}(0,0)$ as a point mass $\delta_0$ at $0$.
Let $x_{\pi} = (x_{\pi(1)}, \dots , x_{\pi(p)})$ in shorthand for $\pi \in \Pi(p,p)$.
Throughout, let $Q$ be a distribution on a measurable space $\mathcal{X}$, let $(x_n)_{n \in \mathbb{N}}$ be a sequence of independent draws from $Q$, and for a measurable $f : \mathcal{X}^p \rightarrow \mathbb{R}$ and $n \geq p$ let
\begin{align*}
V_n(f) \coloneqq \frac{1}{n^p} \sum_{i_1=1}^n \cdots \sum_{i_p=1}^n f(x_{i_1} , \dots ,x_{i_p}), \quad U_n(f) \coloneqq \frac{(n-p)!}{n!} \sum_{(i_1, \dots , i_p) \in \Pi(n,p)} f(x_{i_1}, \dots , x_{i_p})
\end{align*}
denote the associated $V$- and $U$-statistics.

\begin{lemma}[Equivalence of $U$- and $V$-statistics] \label{lem: u v equivalence}
Let $p \geq 2$ and let $f : \mathcal{X}^p \rightarrow \mathbb{R}$ be measurable.
Suppose there exist $C < \infty$ and a measurable $h : \mathcal{X} \rightarrow [0,\infty)$ with $\int h^2 \; \dd Q < \infty$, such that
\begin{align}
\verts{f(x_{1:p})} \leq C \prod_{j=1}^p h(x_j) \qquad \text{for all } x_{1:p} \in \mathcal{X}^p . \label{eq: v-stat envelope}
\end{align}
Then $V_n(f) - U_n(f) = O(n^{-1})$ almost surely; in particular $\sqrt{n} ( V_n(f) - U_n(f) ) \rightarrow 0$ almost surely.
\end{lemma}
\begin{proof}
Let $R_n$ denote the contribution to $V_n(f)$ from tuples $(i_1,\dots,i_p) \notin \Pi(n,p)$, i.e. those having at least two equal indices, so that
\begin{align}
V_n(f) - U_n(f) = R_n + \left( \frac{n!}{(n-p)! \, n^p} - 1 \right) U_n(f) . \label{eq: V minus U identity}
\end{align}
Bounding the indicator that a tuple has a repeated index by $\sum_{1 \leq a < b \leq p} 1_{\{i_a = i_b\}}$, and using \eqref{eq: v-stat envelope} and non-negativity of $h$,
\begin{align*}
\verts{R_n} & \leq \frac{C}{n^p} \sum_{1 \leq a < b \leq p} \sum_{i_1=1}^n \cdots \sum_{i_p=1}^n 1_{\{i_a = i_b\}} \prod_{j=1}^p h(x_{i_j}) \\
& = \frac{C}{n} \binom{p}{2} \left( \frac{1}{n} \sum_{i=1}^n h(x_i)^2 \right) \left( \frac{1}{n} \sum_{i=1}^n h(x_i) \right)^{p-2} ,
\end{align*}
where the equality holds since fixing $i_a = i_b$ produces a factor $\sum_i h(x_i)^2$, while the remaining $p-2$ indices are summed independently.
Since $\int h^2 \; \dd Q < \infty$, and hence also $\int h \; \dd Q < \infty$, the strong law of large numbers implies that both empirical averages are $O(1)$ almost surely, and thus $R_n = O(n^{-1})$ almost surely.
Similarly, from \eqref{eq: v-stat envelope} and non-negativity of $h$,
\begin{align*}
\verts{U_n(f)} \leq C \frac{(n-p)! \, n^p}{n!} \left( \frac{1}{n} \sum_{i=1}^n h(x_i) \right)^p = O(1)
\end{align*}
almost surely, while $n! / ((n-p)! \, n^p) - 1 = O(n^{-1})$ for fixed $p$.
The result follows from \eqref{eq: V minus U identity}.
\end{proof}

\begin{proposition}[Asymptotic normality of $U$-statistics] \label{thm: pfister C7}
Let $f : \mathcal{X}^p \rightarrow \mathbb{R}$ be a symmetric measurable function with $\int f(x_{1:p})^2 \; \mathrm{d}Q^{\otimes p}(x_{1:p}) < \infty$.
Then
$$
\sqrt{n}\left( U_n(f) - \int f(x_{1:p}) \; \mathrm{d}Q^{\otimes p}(x_{1:p}) \right) \stackrel{\mathrm{d}}{\rightarrow} \mathcal{N}(0, p^2 \sigma^2) 
$$
as $n \rightarrow \infty$, where $\sigma^2 \coloneqq \mathbb{V}_{X_1 \sim Q}[ \int f(X_1,x_{2:p}) \; \dd Q^{\otimes (p-1)}(x_{2:p})  ]$.
\end{proposition}
\begin{proof}
This is the classical central limit theorem for $U$-statistics of fixed order, due to \citet{hoeffding1992class} and obtained using Hoeffding's projection; see also Section 5.5.1 of \citet{serfling2009approximation}.
\end{proof}

\begin{coroll}[Asymptotic normality of non-symmetric $V$-statistics]
\label{cor: nonsymmetric V}
Let $f : \mathcal{X}^p \rightarrow \mathbb{R}$ be measurable, but not necessarily symmetric, and suppose that the envelope condition \eqref{eq: v-stat envelope} holds for some $C < \infty$ and some measurable $h : \mathcal{X} \rightarrow [0,\infty)$ with $\int h^2 \; \dd Q < \infty$.
Then there exists $\sigma^2 \in [0 , C^2 ( \int h^2 \; \dd Q )^p ]$ such that
$$
\sqrt{n}\left( \frac{1}{n^p} \sum_{i_1=1}^n \dots \sum_{i_p=1}^n f(x_{i_1}, \dots , x_{i_p}) - \int f(x_{1:p}) \; \mathrm{d}Q^{\otimes p}(x_{1:p}) \right) \stackrel{\mathrm{d}}{\rightarrow} \mathcal{N}(0, p^2 \sigma^2) .
$$
\end{coroll}
\begin{proof}
    For $p = 1$ the result is the classical central limit theorem, so suppose $p \geq 2$.
    Set
    \begin{align*}
        \tilde{f}(x_{1:p}) \coloneqq \frac{1}{|\Pi(p,p)|} \sum_{\pi \in \Pi(p,p)} f(x_\pi)
    \end{align*}
    so that $\tilde{f}$ is symmetric and, since the right-hand side of \eqref{eq: v-stat envelope} is invariant to permutation of its arguments, $\tilde{f}$ also satisfies \eqref{eq: v-stat envelope}.
    In particular, $\int \tilde{f}(x_{1:p})^2 \; \mathrm{d}Q^{\otimes p}(x_{1:p}) \leq C^2 ( \int h^2 \; \dd Q )^p < \infty$.
    Since $\Pi(n,p)$ and $Q^{\otimes p}$ are invariant to permutation of coordinates, $U_n(f) = U_n(\tilde{f})$ and $\int f \; \dd Q^{\otimes p} = \int \tilde{f} \; \dd Q^{\otimes p}$, so \Cref{thm: pfister C7} establishes asymptotic normality of $\sqrt{n}( U_n(f) - \int f \; \dd Q^{\otimes p})$, with variance $p^2 \sigma^2$ where, by Jensen's inequality, $\sigma^2 = \mathbb{V}_{X_1 \sim Q}[ \int \tilde{f}(X_1,x_{2:p}) \; \mathrm{d}Q^{\otimes (p-1)}(x_{2:p})  ] \leq \int \tilde{f}(x_{1:p})^2 \; \mathrm{d}Q^{\otimes p}(x_{1:p}) \leq C^2 ( \int h^2 \; \dd Q )^p$.
    Finally, $\sqrt{n} ( V_n(f) - U_n(f) ) \rightarrow 0$ almost surely from \Cref{lem: u v equivalence}, and the result follows from Slutsky's theorem.
\end{proof}

\subsection{Existence and Uniqueness of Stationary Points}
\label{app: exist unique stationary}

This appendix presents the short proof of \Cref{prop: exist unique}, which ensures the existence and uniqueness of a stationary point for $\eobj$ in \eqref{eq: objective}:

\begin{proof}[Proof of \Cref{prop: exist unique}]
    First we note that \Cref{cor:minimum-is-stationary-eobj} together with \Cref{prop:stationarity-eobj-minimality} imply that the minima of $\eobj$ can be identified with the stationary points of $\eobj$ in the sense of \Cref{def:stationary_point}.
    It therefore remains to show existence of a unique minimum of $\eobj$, and we achieve this by following the argument in \citet[][Proposition 2.5]{hu2021mean}.
    As a sublevel set of the relative entropy,
    $$
    S \coloneq \left\{ Q \in \finitemomentspace{} : \KLD(Q || Q_0) \leq \eobj(Q_0) - \inf_{Q' \in \finitemomentspace{}} \cL(Q') \right\}
    $$
    is non-empty and weakly compact \citep[][Lemma 1.4.3]{dupuis2011weak}.
    Since $\eobj$ is weakly lower semicontinuous, the minimum of $\eobj$ on $S$ is attained.
   Moreover, the global minimum of $\eobj$ is contained in $S$. Indeed, $Q_0 \in S$, since $\KLD(Q_0 || Q_0) = 0 \leq \cL(Q_0) - \inf_{Q' \in \finitemomentspace{}} \cL(Q') = \eobj(Q_0) - \inf_{Q' \in \finitemomentspace{}} \cL(Q')$, while for any $Q \notin S$ we have
    $$
    \eobj(Q) = \cL(Q) + \KLD(Q || Q_0) > \inf_{Q' \in \finitemomentspace{}} \cL(Q') + \eobj(Q_0) - \inf_{Q' \in \finitemomentspace{}} \cL(Q') = \eobj(Q_0) ,
    $$
    so that no $Q \notin S$ can be a global minimiser of $\eobj$.
    Further, since $\eobj$ is strictly convex, the minimum is unique, completing the argument.
\end{proof}

\subsection{Convergence Control}
\label{app: convergence control}

This appendix establishes the sufficient conditions for the \ac{kgd} to control $\polyorder$-convergence of probability measures, as stated in \Cref{thm: convergence control}. 
This is the most technical of the results that we present, and we therefore indicate the high-level structure of the argument in \Cref{app: sketch proof control}, with the finer details presented in the following subsections. 

\subsubsection{Proof of \Cref{thm: convergence control}}
\label{app: sketch proof control}

\begin{proof}
From \Cref{prop: moments exist} in \Cref{sec: diss implies moments} we have $\target \in \finitemomentspace{\polyorder}$. 

To prove the convergence control claim, we make use of two results: 
The first result (\Cref{thm:kgd-convergence-UIseq}, proved in \Cref{subsec:KGDcontrolsUIconv}) states when the \ac{kgd} is continuous and $\target$-separating, 
then vanishing \ac{kgd} convergence indeed implies $\polyorder$-convergence to $\target$ 
provided that the sequence has uniformly integrable $\polyorder$th moments, a condition formalised in \Cref{def:UI}. 
The second result (\Cref{lem:kgd-enforces-UI}, proved in \Cref{subsec:enforceUI}) provides sufficient conditions for a vanishing \ac{kgd} to enforce uniform integrability, which allows us to lift the uniform integrability assumption on the first result. 
Thus, combining these two results, vanishing \ac{kgd} implies $\polyorder$-convergence. 

We now examine how the assumptions and the kernel choice made in \Cref{thm: convergence control} guarantee the preconditions for these two results. 
First, to apply \Cref{thm:kgd-convergence-UIseq}, we need the \ac{kgd}'s separability and continuity. 
The separability of the \ac{kgd} is due to \Cref{thm:general-separation}. 
The recommended kernel choice is a sum of two kernels, where the \ac{rkhs} of $x, y \mapsto \weight_{\polyorder-\beta}(x)L(x,y)\weight_{\polyorder-\beta}(y)$ contains that of an $C_{b,\theta}^1(\R^d, \R^d)_\beta^*$-characteristic kernel with $\theta(x) = (1+ \Verts{x}^2_2)^{\beta/2}$, which can be shown as in the proof of \Cref{prop: sufficient KSD separating} in \Cref{app:proof-of-separation-maintext-TI}. 
The continuity of the \ac{kgd} holds due to \Cref{prop:continuity} and the choice of $\ell$, as translation-invariant kernel are bounded and have bounded diagonal derivatives.

Next, to apply \Cref{lem:kgd-enforces-UI}, we make use of \Cref{lem: growth-approx-general} with $\mathcal{Q} = \finitemomentspace{}$ and so check if its required conditions are satisfied. 
Applying \Cref{lem:sum-dissipativity-justgrowth} implies, due to the uniform bound on $\Verts{\vargrad \cL(Q)(x)}$, that $\genscore{P}{Q}(x)$ is $\finitemomentspace{}$-uniformly dissipative (\Cref{def:dist-dependent-vector-field-dissipativity}). 
Moreover, by \Cref{item:asm-convcontrol-growth} in \Cref{asm:enforce_tightness} and \Cref{item:asm for continuity growth prior} in \Cref{asm: for continuity}, we have  
\begin{align*}
    \Verts{\genscore{P}{Q}(x)} 
    &\leq \Verts{\nabla \log q_0} + \Verts{\vargrad \cL(Q)(x)} \\
    &\leq \left( C \lor  \sup_{Q \in \finitemomentspace{}, x \in \mathbb{R}^d}  \Verts{\vargrad \cL(Q)(x)}\right)  (1 + \Verts{x}^\beta).
\end{align*}
for all $x \in \mathbb{R}^d$ and $Q\in \finitemomentspace{}$ with some $Q$-independent constant $C>0$. 
The kernel form in \Cref{def: recommended kernel} is the same as that of \Cref{lem: growth-approx-general} and $L(x, y) = \ell(x,y) \idmat$ is $\mathcal{C}^1_0(\mathbb{R}^d, \mathbb{R}^d)$-universal by assumption. Thus, all the preconditions of \Cref{lem: growth-approx-general} are met, and so \Cref{lem:kgd-enforces-UI} applies. 

\end{proof}

\subsubsection{Dissipativity Implies Finite Moments}
\label{sec: diss implies moments}

The following proposition explains how the dissipativity assumption in \Cref{asm:enforce_tightness} ensures the existence of moments of $\target$:

\begin{proposition}  \label{prop: moments exist}
    Let \Cref{asm: q0 support} and \Cref{item:asm-concontrol-prior,{item:asm-convcontrol-growth}} in \Cref{asm:enforce_tightness} hold.
    Then a stationary $Q$ has a finite $\polyorder$th moment for any $\polyorder \in (0,\infty)$, that is, $Q \in \finitemomentspace{\polyorder}$. 
\end{proposition}
\begin{proof}
    As in \Cref{def: gen diss}, let $r_1 > 0$ and $r_2 \geq 0$ be the constants 
    that make the assumed dissipativity hold. 
    Now take $R_0 = 1\lor(2r_2/r_1)^{1/(2\diss)}$. Then for $\Verts{x} > R_0$, we have 
    \[
        - \nabla \log q_0(x) \cdot x > \frac{r_1}{2}\Verts{x}^{2\diss}.
    \]
    Now for such $x\in \R^d$, take $\lambda_0 = R_0 / \Verts{x} < 1$. 
    Since $\log q_0$ is $C^1(\mathbb{R}^d)$, from the fundamental theorem of calculus, 
    \begin{align*}
        \log q_0(x) 
        &= \log q_0(\lambda_0 x) +  (1-\lambda_0) \int_0^1 \nabla \log q_0\bigl( \{(1-\lambda_0)t + \lambda_0\} x\bigr) \cdot x \;\dd t \\
        &\leq \max_{\Verts{y}\leq R_0} \log q_0(y) - \frac{r_1}{2}\Verts{x}^{2\diss} (1-\lambda_0)\int_0^1 \bigl( (1-\lambda_0)t +\lambda_0)^{2\diss- 1}\;\dd t\\
        &= \max_{\Verts{y}\leq R_0} \log q_0(y) - \frac{r_1}{4\diss}  \bigl( \Verts{x}^{2\diss} - R_0^{2\diss}\bigr). 
    \end{align*}
    This evaluation implies that $q_0(x) \lesssim \exp\{-r_1\Verts{x}^{2\diss} / (4\diss)\}$. 
    By the self-consistency, 
    \begin{equation*}
    q(x) \propto \exp\{-\cL'(Q)(x)\}q_0(x) \lesssim \exp(-C\Verts{x}^{2\diss})
    \end{equation*}
    for sufficiently large $\Verts{x}$ and some constant $C>0$, where the inequality is due to the assumption $\cL'(Q)(x) = o(\Verts{x}^{2\diss})$ (\Cref{item:asm-convcontrol-growth} in \Cref{asm:enforce_tightness}),  combined with the established decay of $q_0(x)$.  
    The moment finiteness follows from this estimate. 
\end{proof}

\subsubsection{\ac{kgd} Controls Tight $\polyorder$-Convergence}\label{subsec:KGDcontrolsUIconv}

To prove a convergence control property, 
we first restrict ourselves to probability measures with well-controlled tails, formalised as follows: 
\begin{definition}[Uniform integrability]\label{def:UI}
    For $\polyorder\in [0, \infty)$, a sequence $(Q_\seqidx)_{\seqidx \in \mathbb{N}} \subset \mathcal{P}(\mathbb{R}^d)$ is said to have \emph{uniformly integrable $\polyorder$th moments} 
    if 
    $$
    \lim_{r\to \infty} \limsup_{\seqidx\to\infty} \int_{\Verts{x} > r}\Verts{x}^\polyorder \; \dd Q_\seqidx(x) = 0.
    $$
\end{definition}
\noindent 
Informally, uniform integrability requires the tail decay of the members of the sequence to remain at a polynomial rate such that the $\polyorder$th moment exists even in the limit. 
A related notion that will be used in our proof below is \emph{uniform tightness}: 
\begin{definition}[Tightness]\label{def:tightness}
    A set $\mathcal{Q}$ of nonnegative Borel measures is called \emph{uniformly tight} when for each $\varepsilon > 0$ there exists a compact set $S \subset \mathbb{R}^d$ such that $$Q(S^c) \leq \varepsilon$$ for all $Q \in \mathcal{Q}$. 
\end{definition}
\noindent Note that a sequence in $\finitemomentspace{}$ is tight if it satisfies \Cref{def:UI} with $\polyorder=0$, since a closed-ball is compact in $\mathbb{R}^d$. 
Thus, our definition of uniform integrability subsumes the case of tightness. 
In what follows, whenever we refer to the case $\polyorder=0$ in \Cref{def:UI}, this should be interpreted as the sequence being tight. 
Uniform integrability can be understood as a form of tightness, since the definition above states that we can take sufficiently large $N \in \mathbb{N}$ such that $\{ \Verts{x}^\polyorder Q_\seqidx(\dd x)\}_{n \geq N}$ is tight. 

The following theorem shows that \ac{kgd} controls $\polyorder$-convergence when the sequence has uniformly integrable $\polyorder$th moments:

\begin{theorem}[\ac{kgd} controls tight $\polyorder$-convergence]\label{thm:kgd-convergence-UIseq}
Fix $\polyorder \in [0, \infty)$ and let $P \in \finitemomentspace{\polyorder}$. 
Assume the following conditions on the \ac{kgd}:
\begin{itemize}
\item \emph{(Separability)} $\KGD_K(Q) = 0$ iff $Q = P$ for any $Q\in \finitemomentspace{\polyorder}$.
\item \emph{(Continuity)} $\KGD_K(Q_\seqidx) \to \KGD_K(Q) $ if $Q_\seqidx \toL{\polyorder} Q$. 
\end{itemize}
Now, for a sequence $(Q_\seqidx)_{\seqidx \in \mathbb{N}} \subset \finitemomentspace{}$, 
suppose that
\begin{itemize}
\item $(Q_\seqidx)_{\seqidx \in \mathbb{N}}$ have uniformly integrable $\polyorder$th moments. 
\item $\KGD_K(Q_\seqidx) \rightarrow 0$.
\end{itemize}
Then $Q_\seqidx \toL{\polyorder} P$.
\end{theorem}

\begin{proof}
By taking sufficiently large $\seqidx \in \mathbb{N}$ if necessary, we may assume that $Q_\seqidx \in \finitemomentspace{\polyorder}$ for any $\seqidx \in \mathbb{N}$. 
For $\mu \in \finitemomentspace{\polyorder}$, let $\tilde{\mu}$ denote the tilted measure $\tilde{\mu}(\dd x) = (1+\Verts{x}^\polyorder) \mu (\dd x)$. 
Since $(Q_\seqidx)_{\seqidx \in \mathbb{N}}$ has uniformly integrable $\polyorder$th moments, 
so does any subsequence $(Q_{\seqidx_i})_{i \in \mathbb{N}}$. 
In particular, $(\tilde{Q}_{\seqidx_i})_{i \in \mathbb{N}}$ is uniformly tight and 
$(\tilde{Q}_{\seqidx_i}(\mathbb{R}^d))_{i \in \mathbb{N}}$ is uniformly bounded. 
By \citet[Theorem 8.6.2]{Bogachev_2007}, 
$(\tilde{Q}_{\seqidx_i})_{i \in \mathbb{N}}$ contains a subsequence $(\tilde{Q}_{\seqidx_{i_j}})_{j \in \mathbb{N}}$ that converges weakly to some nonnegative measure $Q'$. 
Thus, for any continuous function $f$ of $\polyorder$-growth, we have 
\[
\lim_{j\to\infty} \int f \; \dd Q_{\seqidx_{i_j}} = 
\lim_{j\to\infty} \int \frac{f(x)}{1+\Verts{x}^\polyorder} \; \dd \tilde{Q}_{\seqidx_{i_j}}(x) =
\int \frac{f(x)}{1+\Verts{x}^\polyorder} \; \dd Q'(x) = 
\int f \; \dd Q,
\]
where we have defined $Q(\dd x) \coloneqq (1+\Verts{x}^\polyorder)^{-1} Q'(\dd x)$. 
As a consequence, by letting $f\equiv 1$ above, we obtain $Q \in \finitemomentspace{\polyorder}$, 
and we have shown $Q_{\seqidx_{i_j}} \toL{\polyorder} Q$ as $j \to \infty$. 

Since the \ac{kgd} is continuous with respect to $\polyorder$-convergence, $\KGD_K(Q_{n_{i_j}}) \rightarrow \KGD_K(Q)$. 
This limit is $0$ by the convergence assumption, and we must have $Q=P$ by the separability of the \ac{kgd}. 
Thus, we have shown that, out of any subsequence of $(Q_\seqidx)_{\seqidx \in \mathbb{N}}$, we can extract a (sub)subsequence that $\polyorder$-converges to $P$. 
By a classical argument~\citep[see, e.g.,][Proposition 9.3.1]{Dudley2002}, the original sequence $(Q_\seqidx)_{\seqidx \in \mathbb{N}}$ must also $\polyorder$-converge to $P \in \finitemomentspace{\polyorder}$. 
\end{proof}
\begin{remark}
    The above proof works for any functional on $\finitemomentspace{}$ and is not limited to the \ac{kgd}. 
\end{remark}

\Cref{thm:kgd-convergence-UIseq} is not sufficient to deduce \Cref{thm: convergence control}, because it could occur that $\KGD_K(Q_n) \rightarrow 0$ and yet the uniform integrability of $(Q_n)_{n \in \mathbb{N}}$ is violated.
This problem is addressed in the next section, where we establish conditions under which convergence in \ac{kgd} enforces the uniform integrability requirement. 

\subsubsection{\ac{kgd} Enforces Uniform Integrability}\label{subsec:enforceUI}

This appendix establishes sufficient conditions under which convergence in \ac{kgd} implies uniformly integrability. 
To begin we present a general sufficient condition:

\begin{definition}[Uniform $\polyorder$-growth approximation]
\label{def: uni-q-growth}
Let $ \mathcal{Q} \subset \finitemomentspace{}$ and $\polyorder \in [0, \infty)$.
The family $\{\Op{Q}\rkhs{K}\}_{Q \in \mathcal{Q}}$ is said to \emph{approximate $\polyorder$-growth $\mathcal{Q}$-uniformly} if for any $\varepsilon > 0$, there exists $r_\varepsilon > 0$ and $v_\varepsilon \in \rkhs{K}$
such that
\begin{align}
\Op{Q} v_\varepsilon(x) \geq \Verts{x}^\polyorder 1\{\Verts{x}>r_\varepsilon\} - \varepsilon, \label{eq: dist-uni-growth-approx}
\end{align}
for all $Q \in \mathcal{Q}$. 
\end{definition}

\begin{lemma}[\ac{kgd} enforces uniform integrability]
\label{lem:kgd-enforces-UI}
Fix $\polyorder \in [0,\infty)$. 
Assume that the family $\{\Op{Q}\rkhs{K}\}_{Q \in \mathcal{P}(\mathbb{R}^d)}$ approximates $\polyorder$-growth $\finitemomentspace{}$-uniformly (cf. \Cref{def: uni-q-growth}). 
Then $\KGD_K(Q_\seqidx) \rightarrow 0$ implies that $(Q_n)_{\seqidx \geq 1} \subset \mathcal{P}(\mathbb{R}^d)$ has uniformly integrable $\polyorder$th moments (cf. \Cref{def:UI}). 
\end{lemma}
\begin{proof}
Fix $\varepsilon > 0$.
From the $\polyorder$-growth approximation property, 
there exists $r_\varepsilon > 0$ and $v_{ \varepsilon} \in \rkhs{K}$ such that
\begin{align*}
    \int_{\Verts{x} > r_\varepsilon} \Verts{x}^\polyorder\; \dd Q_\seqidx(x) 
    \leq \int \Op{Q_\seqidx} v_\varepsilon(x) \dd Q_\seqidx(x) + \varepsilon
    \leq \Verts{v_\varepsilon}_{\rkhs{K}} \cdot \KGD_K(Q_\seqidx) + \varepsilon
\end{align*}
where the integral on the RHS in the first inequality is well-defined due to the boundedness of $(\Op{Q}v_\varepsilon)_{-}$; 
the second line follows from the definition of the \ac{kgd} (and the linearity of $\Op{Q}$).  
Thus, letting $\seqidx \to \infty$ on both sides, $\polyorder$th moment uniform integrability is established. 
\end{proof}

It therefore suffices to identify conditions under which we can find $\polyorder$-growth approximation as in \Cref{def: uni-q-growth}. 
To this end, we will require the following generalisation of \Cref{def: gen diss}:

\begin{definition}[Dissipativity for distribution-dependent vector fields]
\label{def:dist-dependent-vector-field-dissipativity}
    Let $\mathcal{Q} \subset \mathcal{P}(\mathbb{R}^d)$ be a family of probability distributions. 
    Let $v:\mathcal{Q} \times \mathbb{R}^d \to \mathbb{R}^d$ be a distribution-dependent vector field. 
    The vector field $ v$ is said to satisfy \emph{$\mathcal{Q}$-uniform generalised dissipativity}, if 
    \begin{align*}
        -v(Q, x)\cdot x \geq r_1 \|x\|^{2\diss} - r_2 %
    \end{align*}
    holds for any $x \in \mathbb{R}^d$ and $Q \in \mathcal{Q}$, 
    where $r_1 >0$, $r_2 \geq 0$ and $\diss \geq 1/2$ are constants independent of $x$ and $Q$. 
\end{definition}

Our next result, which provides such a condition based on \citet[Lemma 3.2]{KanBarGreMac2025}. 

\begin{lemma}[Sufficient conditions for $\polyorder$-growth approximation]
\label{lem: growth-approx-general}
Fix $\polyorder \in [0, \infty)$. 
Let $\mathcal{Q} \subset \finitemomentspace{}$. 
Suppose $(Q,x) \mapsto \genscore{P}{Q}(x)$ satisfies $\mathcal{Q}$-uniform generalised dissipativity with $\diss \geq 1/2$ as in \Cref{def:dist-dependent-vector-field-dissipativity}. 
Assume there exists $C > 0$ such that 
$
    \Verts{\genscore{P}{Q}(x)} \leq C \left( 1 + \Verts{x}^\beta \right)
$
for some $\beta \in [2\diss -1, \infty)$ and all $x \in \mathbb{R}^d$ and $Q \in \mathcal{Q}$. 
Fix $c > 0$.
Define kernel $K$ by 
\begin{equation*}
 K(x,y) \coloneqq \weight_{\polyorder-\beta}(x)\left(L(x,y)+\bar{k}_{\mathrm{lin}}(x,y)\idmat\right)\weight_{\polyorder-\beta}(y)
,%
\end{equation*}
where we denote, for $s \in \mathbb{R}$, $\weight_s(x) \coloneqq \bigl(c^{2}+\Verts x_{2}^{2}\bigr)^{s/2}$, 
$L$ is a matrix-valued kernel universal to $C_0^1(\mathbb{R}^d,\mathbb{R}^{d})$ with $\rkhs{L} \subset \mathcal{C}^1_0(\mathbb{R}^d, \mathbb{R}^d)$, 
and 
\[
\bar{k}_{\mathrm{lin}}(x,y) \coloneqq \frac{k_{\mathrm{lin}}(x,y)}{\sqrt{k_{\mathrm{lin}}(x,x)} \sqrt{k_{\mathrm{lin}}(y,y)}}
\]
is the normalized version of a linear kernel $k_{\mathrm{lin}}(x,y)= c^2 + x \cdot y$.  
Then, for each $\varepsilon>0$, 
there exists $r_\varepsilon>0$ and $v_\varepsilon \in \rkhs{K}$ such that 
\[
    \Op{Q} v_{\varepsilon} (x) \geq \Verts{x}_2^\polyorder 1\{\Verts{x}_2 > r_\varepsilon\} - \varepsilon
\]
for all $Q \in \mathcal{Q}$. 
\begin{proof}
Recall that the Langevin Stein operator for a distribution $P$ with density $p$ is defined by 
\[
    \LStOp{P} v (x) \coloneqq \nabla \log p(x) \cdot v(x) + \nabla \cdot v(x).
\]
For each fixed $Q\in\mathcal{Q}$, the operator $\Op{Q}$ can be seen as the Langevin Stein operator 
$\LStOp{P}$ with $\nabla \log p$ replaced with $\genscore{P}{Q}$. 
We assume that $(Q, x)\mapsto \genscore{P}{Q}$ satisfies $\mathcal{Q}$-uniform dissipativity (\Cref{def:dist-dependent-vector-field-dissipativity}). 
This assumption allows us to choose $v_\varepsilon$ and $r_\varepsilon$ independent of $Q$ 
by simply following the construction techniques for the Langevin Stein operator developed by \cite{barp2024targeted} and \citet{KanBarGreMac2025}; 
this step can be performed by verifying that required assumptions hold independently of $Q$. 
To obtain the desired result, we use \citet[Lemma 3.2]{KanBarGreMac2025}: 
Assumptions 1 and 2 of \citet{KanBarGreMac2025} are verified in \Cref{lem: verify-dis-growth-assm}; 
and our kernel choice corresponds to their recommended kernel with diffusion matrix $m(x) = \weight_{1-\beta}(x)\idmat$. 
\end{proof}
\end{lemma}

\subsubsection{Supporting Lemmas}

This appendix provides the remaining details for convergence control, 
establishing \Cref{lem:sum-dissipativity-justgrowth}, 
which was used in the proof of \Cref{thm: convergence control}, and \Cref{lem: verify-dis-growth-assm}, which was used in the proof of \Cref{lem: growth-approx-general}.

\begin{lemma}[Dissipativity is unaffected by a vector field of slower growth]\label{lem:sum-dissipativity-justgrowth}
    Let $\mathcal{Q} \subset \finitemomentspace{}$. 
    Let $v_1,v_2 : \mathcal{Q} \times \mathbb{R}^d \rightarrow \mathbb{R}^d$. 
    Suppose that $v_1$ satisfies $\mathcal{Q}$-uniform dissipativity condition for some $\diss > 1/2$
    and $r_1 > 0$, $r_2 \geq 0$. 
    Suppose there exists $C \geq 0 $ and $0 < \varepsilon \leq 2\diss - 1$ such that 
    $\| v_2(Q,x) \| \leq C(1 + \|x\|^{2\diss -1-\varepsilon})$ 
    for all $Q \in \mathcal{Q}$ and $x \in \mathbb{R}^d$. 
    Then $v_1 + v_2$ satisfies $\mathcal{Q}$-uniform generalised dissipativity with the same $\diss \geq 1/2$. 
\end{lemma}

\begin{proof}
Let $v \coloneqq v_1 + v_2$.
By assumption, for any $x \in \mathbb{R}^d$ and $Q \in \mathcal{Q}$, 
    \[ 
    - v_1(Q,x) \cdot x \geq r_1 \|x\|^{2\diss} - r_2.
    \]
Then, 
\begin{align*}
    - v(Q,x) \cdot x 
    &= - v_1(Q,x) \cdot x - v_2(Q,x) \cdot  x \\
    &\geq r_1 \|x\|^{2\diss} - r_2  - \langle v_2(Q,x),x \rangle  \\
    & \geq r_1 \|x\|^{2\diss} \left( 1 - \frac{r_2}{r_1 \|x\|^{2\diss}}  - c\frac{\Verts{v_2(Q,x)} }{ r_1 \|x\|^{2\diss-1}}\right), 
\end{align*}
where $c > 0$ is some constant such that $\Verts{x}_* \leq c \Verts{x}$ for the dual norm $\Verts{\cdot}_*$. 
Given the growth condition on $v_2$,  the term in parentheses converges to 1 uniformly over $Q$ as $\Verts{x}_2 \to \infty$. 
Thus, there exist $\tilde{R}>0$ and $0< r'_1 < r_1 $ such that for every $x$ with $\|x\| > \tilde{R}$,
\[
 - v(Q,x)\cdot x > r'_1 \|x\|^{2\diss}.
\]
For this $\tilde{R}$, in the case $\|x\| \leq \tilde{R}$ we have 
\begin{align*}
    - v(Q,x) \cdot x
    &= - v_1(Q,x) \cdot x -  v_2(Q,x) \cdot  x\\ 
    &\geq r_1 \Verts{x}^{2\diss} - r_2  - C(1+\Verts{x}_2)^{2\diss -1-\varepsilon} \Verts{x}  \eqqcolon \ell(x) \geq -r'_2, 
\end{align*}
where we have defined $r'_2 \coloneqq  \max_{\Verts{x}\leq \tilde{R}} \verts{\ell(x)}$ and have used both the dissipativity of $v_1$ and the growth condition of $v_2$ to obtain the inequality. 
These considerations yield
\[
 - v(Q,x)\cdot x \geq r'_1 \|x\|^{2\diss} - r_1'  R^{2\diss}
- r'_2
\]
for every $x \in \mathbb{R}^d$ and $Q \in \mathcal{P}(\mathbb{R}^d)$, as desired. 
\end{proof}

\begin{lemma}[Verifying the assumptions in \citealt{KanBarGreMac2025}]
\label{lem: verify-dis-growth-assm}
Let $\mathcal{Q} \subset \finitemomentspace{}$. 
Suppose  $\genscore{P}{Q}$ satisfies $\mathcal{Q}$-uniform generalised dissipativity with $\diss > 1/2$ as in \Cref{def:dist-dependent-vector-field-dissipativity}. 
Assume 
there exists $C>0$ such that 
     $\Verts{\genscore{P}{Q}(x)} \leq C (1 + \Verts{x}^{2\diss - 1})$
for all $x \in \mathbb{R}^d$ and $Q \in \mathcal{Q}$. 
For each $Q\in \mathcal{Q}$, define $\tau = 2\diss - 1$, 
$b(x) \coloneqq \{a_{1-\tau}(x) \genscore{P}{Q}(x) + \nabla \weight_{1-\tau}(x)\}/2$
and 
$m(x) \coloneqq a_{1-\tau}(x)\idmat$, where $\weight_{1-\tau}(x) \coloneqq (c^2 + \Verts{x}_2^2)^{(1-\tau)/2}$. 
Then, $b$ and $m$ satisfies the dissipativity condition (\citealt[Assumption 1]{KanBarGreMac2025}): 
\begin{align*}
    2 b(x) \cdot x+ \mathrm{tr}[m(x)] \leq -\alpha \Verts{x}_2^2 + \beta,
\end{align*}
where $\alpha>0$ and $\beta>0$ are constants independent of $Q$. 
Moreover, $b$ and $m$ satisfy the linear growth condition (\citealt[Assumption 2 with $q_m=0$]{KanBarGreMac2025}): 
\begin{align*}
\Verts{b(x)}  \leq \lambda_b (1+\Verts{x}_2)\ \text{and}\  \Verts{m(x)}_{\mathrm{op}} \leq \lambda_m (1 + \Verts{x}_2)
\end{align*}
where $\lambda_b > 0$ and $\lambda_m > 0$ are $Q$-independent. 
\end{lemma}
\begin{proof}
To check the dissipativity, we use the following estimate
\begin{align*}
    2 b(x) \cdot x+ \mathrm{tr}[m(x)] 
    &= a_{1-\tau}(x) \left( \genscore{P}{Q}(x)\cdot x + \nabla \log \weight_{1-\tau}(x) \cdot x +  d\right)\\
    &\leq  a_{1-\tau}(x) \left( -r_1\Verts{x}^{2\diss}+ \underbrace{r_2 + \frac{\verts{1-\tau}}{2}+ d}_{\tilde{r}_2}\right) \eqqcolon \text{UB}(x),
\end{align*}
which follows from the $\mathcal{Q}$-uniform dissipativity of $\genscore{P}{Q}$ 
(note $r_1$, $r_2$ and $u$ are independent of $Q$). 
Let $r_0 = \{2\tilde{r}_2 / r_1\}^{1/(2\diss)} \lor 1$. 
If $\Verts{x} > r_0$, 
\begin{align*}
    \text{UB}(x) 
    &= -a_{1-\tau}(x) \Verts{x}^{2\diss} \left(r_1 - \tilde{r}_2 \Verts{x}^{-2\diss}\right) \\
    &= -\left(\frac{c^{2}+\Verts{x}^2}{\Verts{x}^2}\right)^{(1-\tau)/2}\Verts x^{2\diss +(1-\tau)} \left(r_1 - \tilde{r}_2 \Verts{x}^{-2\diss}\right) \\
    &< -(c^2+1)^{\verts{1-\tau}/2}r_1/2 \cdot \Verts{x}^{2}.
\end{align*}
where we have used 
$\tau = 2\diss -1$ 
to obtain the final estimate. 
Let $\beta_0 \coloneqq \max_{\Verts{x}\leq r_0} \tilde{r}_2 \weight_{1-\tau}(x) > 0$, which satisfies $\text{UB}(x) \leq \beta_0$ for $\Verts{x} \leq r_0$. 
Thus, we obtain 
\begin{align*}
    2 b(x) \cdot x + \mathrm{tr}[m(x)] \leq -\alpha \Verts{x}^2 + \beta
\end{align*}
with $\alpha = (c^2+1)^{\verts{1-\tau}/2} r_1/2$ and 
$\beta = \beta_0 + \alpha  r_0^2$. 
As a result, the desired claim holds with $\alpha$ being replaced with $a \alpha$, where $a > 0$ is a constant such that $a \Verts{\cdot}_2 \leq  \Verts{\cdot}$. 

The linear growth conditions can be verified as follows. 
Let $A > 0$ be a constant such that $\Verts{\cdot} \leq A \Verts{\cdot}_2$. 
By assumption, there exists $Q$-independent $C > 0$ such that $\Verts{\genscore{P}{Q}(x)} \leq C(1+\Verts{x}^\tau)$, and thus 
\begin{align*}
2\Verts{b(x)} 
&\leq \verts{a_{1-\tau}(x)}\left( \Verts{\genscore{P}{Q}(x)} + \Verts{\nabla \log \weight_{1-\tau}(x)}\right)\\
&\leq \left( C(1+A^\tau\Verts{x}_2^\tau) + \verts{1-\tau}\frac{\Verts{x}}{c^2 + \Verts{x}_2^2}\right) \bigl(c^2+\Verts{x}_2^2\bigr)^{-\tau/2}  \sqrt{c^2+\Verts{x}_2^2}  \\
&\leq \left\{C+ c^{-\tau}\left(CA^\tau + \frac{\verts{1-\tau}}{2c} \right)\right\}  \sqrt{c^2+\Verts{x}_2^2} 
\end{align*}
and 
\begin{equation*}
    \Verts{m(x)}_{\mathrm{op}} = \verts{a(x)} \leq (1 \lor (c^2+1)^{1-\tau}) \sqrt{c^2+\Verts{x}_2^2}. 
\end{equation*}
Thus the assumptions of \citet{KanBarGreMac2025} are satisfied.
\end{proof}

\section{Experimental Detail}
\label{app: experiments}

This appendix contains full experimental details for reproducing \Cref{ex: mfnn} (\Cref{app: detail mfnn}) and \Cref{ex: LV} (\Cref{app: detail pcuq}) in the main text.
A common theme for these experiments is that we avoid using the same kernel for both the numerical method (\ac{kgd} descent, extensible sampling, and \ac{vgd}) and performance assessment (\ac{kgd}), to avoid favouring those methods that directly optimise \ac{kgd}.
Below the specific kernels that we use for experiments are specified in full.

\subsection{Details for \Cref{ex: mfnn}}
\label{app: detail mfnn}

This appendix contains full details needed to reproduce the experiments concerning \Cref{ex: mfnn}.

\subsubsection{Experimental Protocol}

A univariate regression task was considered, where the dataset $\{(z_i,y_i)\}_{i=1}^N$ of size $N = 300$ was generated by first sampling each covariate $z_i \in \R$ from $\mathcal{U}(0,1)$ and then sampling each response $y_i \in \R$ from a normal distribution with mean $f(x_i) = 3\tanh (3 x_i + \frac{1}{2}) -3$ and standard deviation $\sigma = 0.1$.
The neural network $\Phi$ was taken to have the form $\Phi(z,x) = w_2 \cdot \tanh (w_1 \cdot z + b_1) + b_2$ where $x = (w_1,b_1,w_2,b_2)$, and the reference distribution $Q_0$ was $\mathcal{N}(0,0.5 I_d)$ with $d = 4$.
We chose  $\lambda = 300$ for $\cL$, and the test error was computed using an independent dataset of size $N=300$. 

\subsubsection{Details for \Cref{fig: tuning}}
\label{app: detail tuning}

The kernel density estimate $\hat{Q}_\epsilon$ used for the comparison in \Cref{fig: tuning} was constructed using the Gaussian kernel with bandwidth set equal to $1/\sqrt{n}$ with $n=100$. 
The gradient $\vargrad \cL (\hat{Q}_\epsilon)$ is given by a Gaussian convolution, and we used an estimate of this based on independent samples from $\hat{Q}_\epsilon$ to further compute an estimate of $\KGD(\hat{Q}_\epsilon)$.

\subsubsection{Mean Field Langevin Dynamics}

For the evaluation in \Cref{fig: different methods} we ran \ac{mfld} with step size $\epsilon = 10^{-5}$.
The $n$ particles were initialised as independent samples from $\mathcal{N}(0, 9\idmat)$. 

\subsubsection{Parametric Variational Inference}

The loss minimised in this method is a $U$-statistic of the squared \ac{kgd} based on a sample of size $100$ from the pushforward distribution $T^{\theta}_{\#}\mu_0$. 
For initialisation we pre-train by selecting $\theta$ to approximately minimise $\mathrm{MMD}^2\left(T^{\theta}_{\#}\mu_0, \mathcal{N}(0,9\idmat)\right)$.
The reference distribution was $\mu_0 =\mathcal{U}([-3,3]^{d_0})$, where $d_0 = 4$.
Training was performed using the Adam optimiser with learning rate $10^{-3}$.
The kernel $k$ was taken to be $k(x,x') = (1 + \|x-x'\|^2/0.3^2)^{-1}$.

\subsubsection{KGD Descent}

Numerical simulation of \ac{kgd} descent was performed using Adam optimiser with step size $\varepsilon = 10^{-2}$. 
The kernel was taken to be $K(x,x') = k(x,x') \idmat$ with an inverse multi-quadric kernel, $k(x,x') = (1 + \|x-x'\|^2/4^2)^{-1/2}$.
The $n$ particles were initialised as independent samples from $\mathcal{N}(0, 9\idmat)$.

\subsubsection{Assessment}

To generate the contours in \Cref{fig: different methods} we ran $t = 10^5$ steps of \ac{mfld} with $n = 300$ particles and step size $\epsilon = 10^{-6}$.

For all numerical methods evaluated in \Cref{fig: different methods} we aimed to generate $n = 100$ particles such that the associated empirical distribution is an accurate approximation of $P$.
\\

\Cref{fig: different methods kgd eval} and \Cref{fig: different methods recommended kernel} reproduce \Cref{fig: different methods}, but for performance assessment we use \ac{kgd} instead of the test generalisation error. \ac{kgd} evaluation on \Cref{fig: different methods kgd eval} was performed using the kernel $K(x,x') = k(x,x') I_d$ where $k(x,x') = (1 + \|x - x'\|^2/0.1^2)^{-1/2}$.

\Cref{fig: different methods recommended kernel} uses the recommended kernel of \Cref{def: recommended kernel} with $\polyorder = 2$, $\beta = 1$, $c = 1$ and $L(x,x') = l(x,x') I_d$ with $l(x,x') = (1 + \|x - x'\|^2/5^2)^{-1}$. This kernel controls first and second order moments, and is therefore a stricter criterion compared to the kernel used to produce \Cref{fig: different methods}.
The results in \Cref{fig: different methods recommended kernel} accordingly provide insight which \Cref{fig: different methods} does not; using this stronger criterion we can see that none of the three numerical methods appear to have converged with respect to capturing the tail behaviour of the target.

\begin{figure}[t]
\centering
    \includegraphics[width=\textwidth]{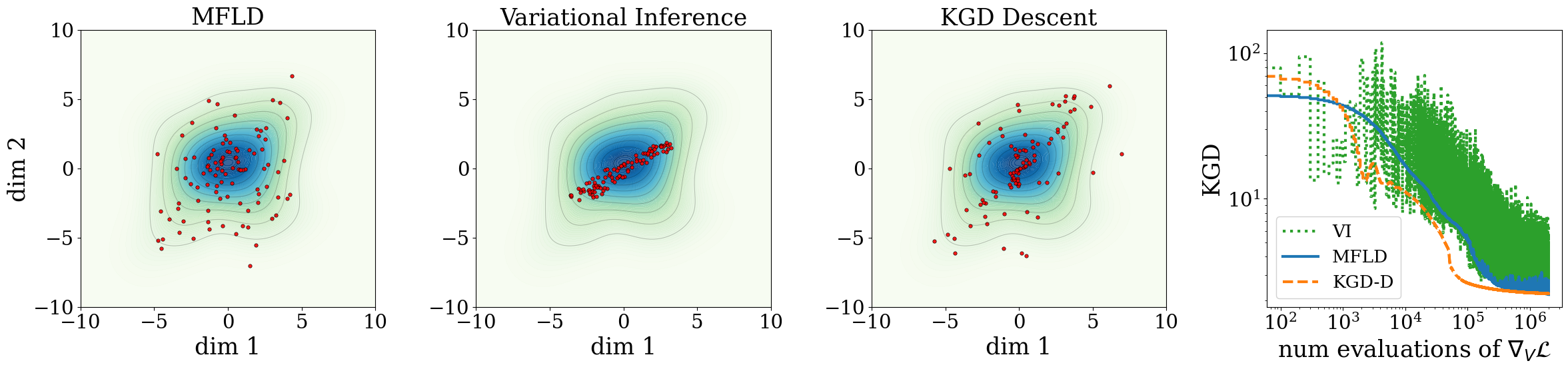}
    \caption{Reproduction of \Cref{fig: different methods}, but where the rightmost panel now displays \ac{kgd} evaluation as a function of the number of backpropagations (evaluation of $\nabla_V \cL$). 
    The kernel used for \ac{kgd} evaluation is an inverse multi-quadratic kernel. (The first three panels are unchanged from \Cref{fig: different methods}.)
    }
    \label{fig: different methods kgd eval}
\end{figure}

\begin{figure}[t!]
\includegraphics[width=\textwidth]{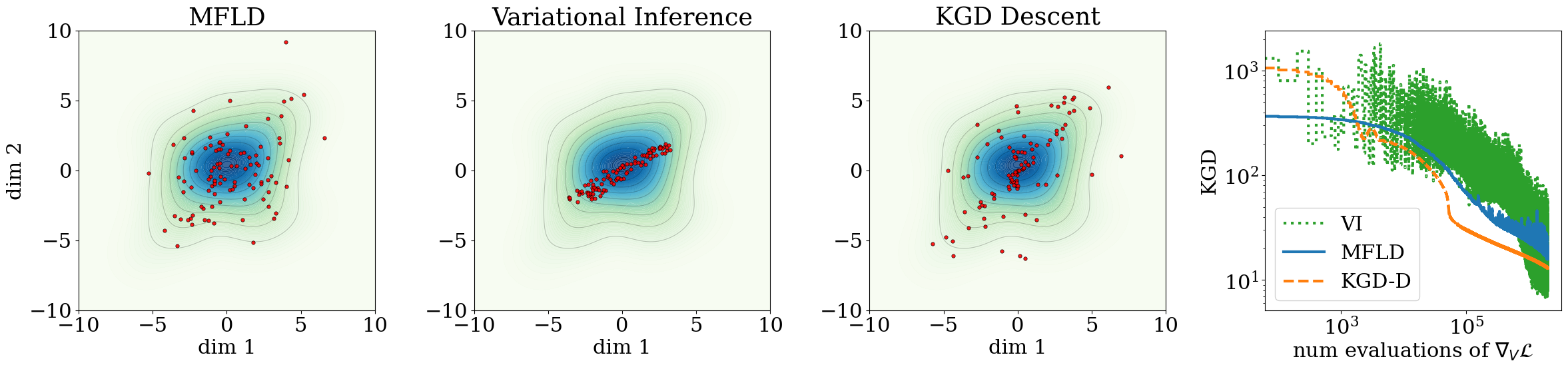}
    \caption{Reproduction of \Cref{fig: different methods kgd eval} where, compared to the figure in the main text, here the rightmost panel is instead based on the recommended kernel for controlling first and second order moments, cf. \Cref{def: recommended kernel}. (The first three panels are unchanged from \Cref{fig: different methods}.)
    }
    \label{fig: different methods recommended kernel}
\end{figure}

\subsection{Details for \Cref{ex: LV}}
\label{app: detail pcuq}

This appendix contains full details needed to reproduce the experiments concerning \Cref{ex: LV}.

\subsubsection{Lotka--Volterra Model}

Here we recall the Lotka--Volterra experiment presented in \citet{shen2024prediction}.
As a prototypical model from population ecology, the Lotka--Volterra model describes the interaction between a prey ($u_1$) and predator ($u_2$) as a pair of coupled \acp{ode}
\begin{align*}
    \frac{\dd u_1}{\dd t} & = \alpha u_1 - \beta u_1 u_2, \qquad u_1(0) = \xi_1 \\
    \frac{\dd u_2}{\dd t} & = \delta u_1 u_2 - \gamma u_2, \qquad u_2(0) = \xi_2
\end{align*}
for some $\alpha,\beta,\gamma,\delta,\xi_1,\xi_2 \geq 0$.
Data arise as noisy observations of one or more species at discrete times; \citet{shen2024prediction} supposed that $\tilde{y}_i$ are noisy measurements of the populations $u = (u_1,u_2)$ at times $t_i$. 
The noise was treated as independent Gaussian noise across time and across species, with standard deviation $\sigma \geq 0$.
In the experimental setting of \citet{shen2024prediction} all parameters are fixed and known except for $x_1 := \mathrm{logit}(\alpha)$ and $x_2 = \mathrm{logit}(\beta/\alpha)$, which are to be inferred (i.e. $d = 2$).
The data-generating parameters were $\alpha = \mathrm{logit}^{-1}(-1)$, $\beta = \mathrm{logit}^{-1}(-3)$, $\gamma = 0.4$, $\delta = 0.02$, $\xi_1 = 10$, $\xi_2 = 15$, $\sigma = 1$.
Data were simulated at times $t_i$ ranging over a uniform grid from $0$ to $n=60$ with increments of $1$.

\subsubsection{Predictively Oriented Posteriors}

For our experiments we adopt a similar set-up to \citet{shen2024prediction}.
The main idea \citep[following][]{cherief2020mmd} is to `lift' the regression data $\tilde{y}_i$ to the augmented data $y_i \coloneq (t_i,\tilde{y}_i) \in \R \times \R$ and compare the empirical distribution of these augmented data to that predicted by a mixture model.
We set the variational objective function as
\begin{align*}
    \mathcal{J}(Q) = \frac{1}{2\lambda_N} \mathrm{MMD}^2(P_Q, P_N) + \mathrm{KLD}(Q\Vert Q_0),
\end{align*}
where $P_Q$ denotes the predictive distribution $p_Q(t,\tilde{y}) \coloneq \frac{1}{N} \sum_{i=1}^N \delta_{t_i}(t) \cdot \int p(\tilde{y}_i | t_i, x) \; \mathrm{d}Q(x)$, and $P_N$ denotes the empirical measure $\frac{1}{N} \sum_{i=1}^N \delta_{t_i}(t)\delta_{y_i | t_i}$ on $\R \times \R$. 
Following \citet{cherief2020mmd,shen2024prediction}, we select a kernel for the \ac{mmd} that does not depend on the time component of $y$; denote this $\kappa : \R \times \R \rightarrow \R$.
Then the terms in \ac{mmd} take the form $\mathrm{MMD}^2(P_Q,P_N) = C+\iint \kappa_{P_N}(x, x') \; \mathrm{d}Q(x) \mathrm{d}Q(x')$, where
\begin{align}
    \kappa_{P_N}(x, x') & \coloneq \frac{1}{N^2} \sum_{i=1}^N \sum_{j=1}^N \iint \kappa(y, y') \; \mathrm{d}P(y|x, t_i)\mathrm{d}P(y'|x', t_j) \label{eq: kappaN} \\
    & \qquad - \frac{1}{N} \sum_{i=1}^N \int \kappa(y_i, y) \; \mathrm{d}P(y|x, t_i) - \frac{1}{N} \sum_{i=1}^N \int \kappa(y_i, y) \; \mathrm{d}P(y|x', t_i). \nonumber
\end{align}
Here we select $Q_0$ as the standard Gaussian distribution, $\kappa$ as a Gaussian kernel $\kappa(y, y')=\exp(- \lVert y-y' \rVert^2 / 2)$, and $\lambda_N = 0.1 / N$. 
As the measurement model $P(y\vert x, t_i)$ is Gaussian, the integrals appearing in \eqref{eq: kappaN} can be analytically evaluated.

To generate synthetic data, we corrupt the Lotka--Volterra model by simulating the system with the addition of an intrinsic noise:
\begin{align*}
    \mathrm{d}u_1 &= (\alpha u_1 - \beta u_1 u_2) \; \mathrm{d}x + \epsilon_1 \; \mathrm{d}W_1, & u_1(0) = \xi_1 , \\
    \mathrm{d}u_2 &= (\delta u_1 u_2 - \gamma u_2) \; \mathrm{d}x + \epsilon_2 \; \mathrm{d} W_2, & u_2(0) = \xi_2,
\end{align*}
and we choose $\epsilon_1 = 0.1$, $\epsilon_2 = 0.2$. 
This is to simulate the intended use case of prediction-centric methods; they aim for robust predictive performance in settings where the parametric statistical model is misspecified.

\begin{figure}[t!]
    \centering
    \includegraphics[width=0.7\textwidth]{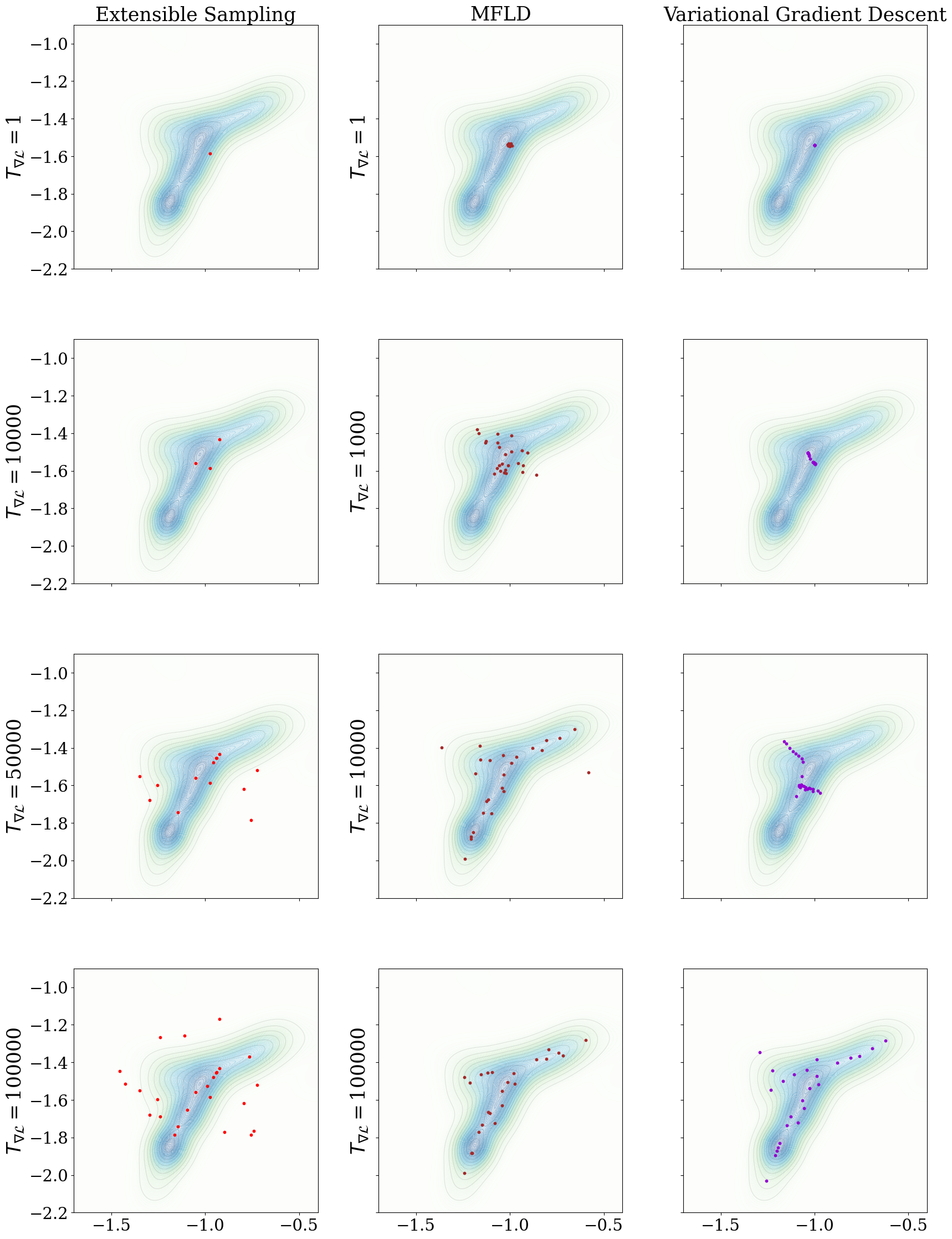}
    \caption{Examining the approximations produced by \ac{mfld}, extensible sampling, and \ac{vgd} in the context of \Cref{fig: different methods 2}.
    Here $T_{\vargrad \cL}$ denotes the total number of evaluations of $\vargrad \cL$ that were required.}
    \label{fig: different methods 2 evolution}
\end{figure}

\subsubsection{Mean Field Langevin Dynamics}

The explicit form of \ac{mfld} for \Cref{ex: LV} is obtained using the variational gradient
\begin{align*}
    \vargrad \cL(Q_n)(x) & = \frac{1}{n\lambda_N}\sum_{j=1}^n \nabla_1 \kappa_{P_N}(x, x_i) , \qquad Q_n = \frac{1}{n} \sum_{i=1}^n \delta_{x_i} .
\end{align*}
For the evaluation in \Cref{fig: different methods 2} we ran \ac{mfld} with step size $\epsilon = 10^{-5}$.
The $n$ particles were initialised as a set of perturbed samples near the (misspecified) true data parameter $(x_1, x_2)$ with a Gaussian variance $10^{-6}$.

\subsubsection{Extensible Sampling}

Numerical optimisation of \ac{kgd} was performed by stochastic search over $4050$ samples from $\mathcal{N}((-1,1.6),0.01 I_d)$.
The kernel $k$ that defines the \ac{kgd} objective was $k(x,x') = (1 + \|x-x'\|^2/0.03^2)^{-1/2} + (1 + \|x-x'\|^2/2^2)^{-1/2}$.

\subsubsection{Variational Gradient Descent}

The kernel $k$ was taken to be $k(x,x') = (1/3) \cdot \sum_{i=1}^3 (1 + \|x-x'\|^2/\ell^2_i)^{-1/2}$  with length scales $\ell_1 = 0.001$, $\ell_2 = 0.01$ and $\ell_3 = 0.1$.
The $n$ particles were initialised as independent samples from $Q_0$. 
The \ac{ode} was simulated using the Adam optimiser with step size $\varepsilon = 10^{-3}$.

\subsubsection{Assessment}

To generate the contours in \Cref{fig: different methods 2} we ran $t = 50,000$ steps of \ac{mfld} with $n = 2,000$ particles and step size $\epsilon = 10^{-6}$.

For all numerical methods evaluated in \Cref{fig: different methods 2} we aimed to generate $n = 25$ particles such that the associated empirical distribution is an accurate approximation of $P$.
The quantitative assessment in the right hand panel was performed based on \ac{kgd} using kernel $k(x,x') =  (1 + \|x-x'\|^2/0.1^2)^{-1/2}$.

\Cref{fig: different methods 2 evolution} displays the evolution of the approximations produced by \ac{mfld}, extensible sampling, and \ac{vgd} in the context of \Cref{fig: different methods 2}.
\Cref{fig: different methods 2 recommended kernel} reproduces \Cref{fig: different methods 2}, but where the kernel used for performance assessment is the recommended kernel with same parameters as for \Cref{fig: different methods recommended kernel}.
Here we observed similar conclusions from \Cref{fig: different methods 2,fig: different methods 2 recommended kernel}, which we attribute to the numerical methods performing equally well with respect to approximating the tails of the target.

\begin{figure}[t!]
\includegraphics[width=\textwidth]{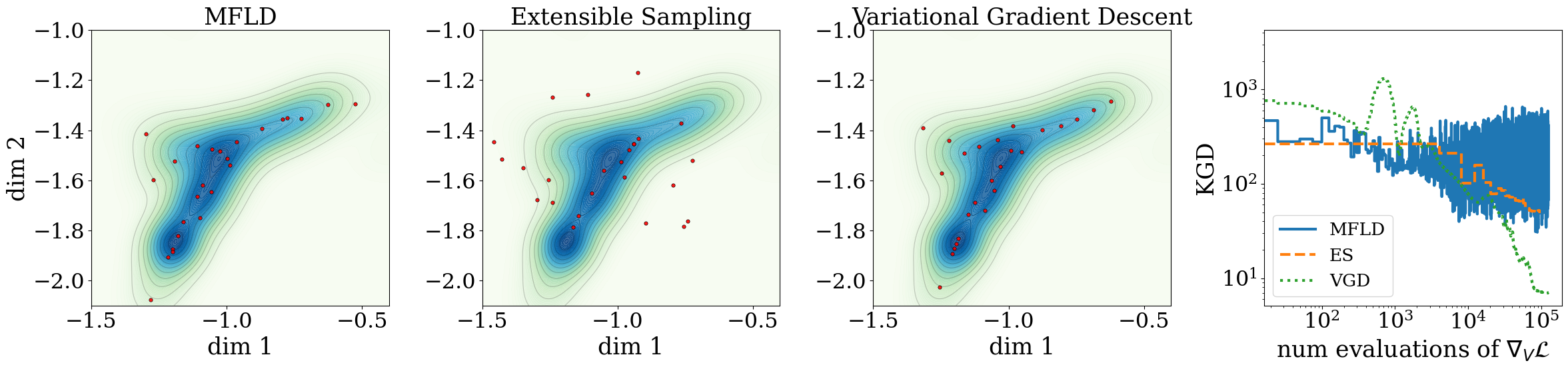}
    \caption{Reproduction of \Cref{fig: different methods 2} where, compared to the figure in the main text, here the rightmost panel is instead based on the recommended kernel for controlling first and second order moments, cf. \Cref{def: recommended kernel}. (The first three panels are unchanged from \Cref{fig: different methods 2}.)
    }
    \label{fig: different methods 2 recommended kernel}
\end{figure}

\subsection{Understanding the Performance of \ac{kgd}}
\label{app: marginals}

An empirical investigation into \ac{kgd}, to understand how closely it captures the quality of an approximation $Q$ to the solution $P$ of \eqref{eq: objective}, is difficult because $P$ cannot be accessed in general.
However, we can construct instances of \eqref{eq: objective} for which we know certain marginals of $P$ are identical by a symmetry argument.
Since we have explicit access to $Q$, it is then straightforward to check how similar the corresponding marginal distributions are for $Q$, and this provides a route to understanding of the empirical performance of \ac{kgd}.

To construct a variational objective \eqref{eq: objective} for which certain marginals of $P$ are identical, let $\bar{\mathcal{L}} : \mathcal{P}(\mathbb{R}^{2d}) \rightarrow \mathbb{R}$ where $\bar{\mathcal{L}}(\bar{Q}) = \mathcal{L}(Q_1) + \mathcal{L}(Q_2)$ where $Q_1$ and $Q_2$ are the marginals of $\bar{Q}$; i.e. $Q_1(S) = \bar{Q}(S \times \mathbb{R}^d)$ and $Q_2(S) = \bar{Q}(\mathbb{R}^d \times S)$.
Let $\bar{Q}_0 = Q_0 \otimes Q_0$ and let $\bar{P}$ minimise the variational objective $\bar{J}(\bar{Q}) = \bar{\mathcal{L}}(\bar{Q}) + \mathrm{KLD}(\bar{Q}||\bar{Q}_0)$.
For any $\bar{Q}$ with $\mathrm{KLD}(\bar{Q}||\bar{Q}_0) < \infty$, with marginals $Q_1$ and $Q_2$, the chain rule for the \ac{kld} gives
$$
\mathrm{KLD}(\bar{Q}||\bar{Q}_0) = \mathrm{KLD}(Q_1||Q_0) + \mathrm{KLD}(Q_2||Q_0) + \mathrm{KLD}(\bar{Q}||Q_1 \otimes Q_2) ,
$$
and hence $\bar{J}(\bar{Q}) = \mathcal{J}(Q_1) + \mathcal{J}(Q_2) + \mathrm{KLD}(\bar{Q}||Q_1 \otimes Q_2)$, where $\mathcal{J}(Q) = \mathcal{L}(Q) + \mathrm{KLD}(Q||Q_0)$. Since the final term is non-negative, and equal to zero if and only if $\bar{Q} = Q_1 \otimes Q_2$, a minimiser of $\bar{J}$ must be a product measure whose factors each minimise $\mathcal{J}$.
Thus we deduce that $\bar{P} = P \otimes P$ where $P$ minimises $\mathcal{J}(Q) = \mathcal{L}(Q) + \mathrm{KLD}(Q||Q_0)$ (assumed to have a unique minimiser; sufficient conditions are given in \Cref{prop: exist unique}).  
In particular, the marginals $P_1$ and $P_2$ of $\bar{P}$, defined as $P_1(S) = \bar{P}(S \times \mathbb{R}^d)$ and $P_2(S) = \bar{P}(\mathbb{R}^d \times S)$, are equal.

For the purpose of our empirical investigation we took $d = 3$, $Q_0 = \mathcal{N}(0,I_d)$, and 
$$
\mathcal{L}(Q) = \int \sin(x_1) + \cos(x_2) + \sin(x_3^2) \; \mathrm{d}Q(x) + \iint \exp(-\|x-x'\|^2) \; \mathrm{d}Q(x) \mathrm{d}Q(x') ,
$$
and we measured the difference between the marginals $Q_1$ and $Q_2$ using \ac{mmd} based on the IMQ kernel $k(x,x') = (1 + \|x-x'\|^2)^{-1/2}$.
Results are contained in \Cref{fig: understanding}, where we plot typical marginal distributions produced by \ac{mfld}, extensible sampling (\Cref{sec: extensible}), and \ac{vgd} (\Cref{sec: svgd}), together with the \ac{kgd} and \ac{mmd} as a function of the number of evaluations of the variational gradient $\nabla_{\mathrm{V}} \bar{\mathcal{L}}$. 
It can be seen that \ac{mfld} exhibits the most disagreement between the corresponding marginals (esp. $x_2$ and $x_5$), and that \ac{mfld} also achieved the largest values of both \ac{kgd} and \ac{mmd}.
At the same time, the marginals produced from extensible sampling appear to be about as similar as the marginals produced from \ac{vgd}, and this is reflected in similar values for these methods in terms of \ac{kgd} and \ac{mmd}. 
These results are therefore consistent with \ac{kgd} being a meaningful measure of approximation quality, and serve to validate the theoretical results we present in \Cref{sec: theory}.

\begin{figure}[t!]
    \centering
    \includegraphics[width =\textwidth]{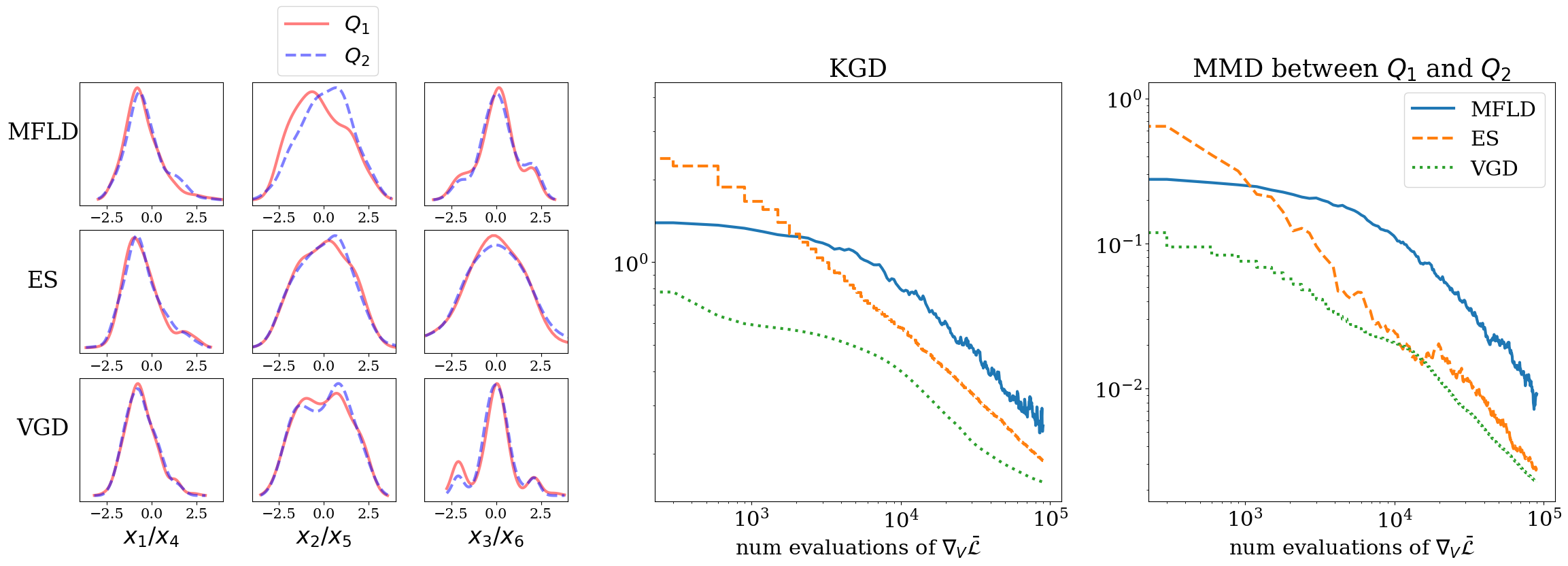}
    \caption{Understanding the Performance of \ac{kgd}.
    Left:  Overlaying pairs of marginal distributions obtained from numerical approximations to $P$, which should be identical if the numerical approximation is exact.
    Middle:  \ac{kgd} as a function of the number of evaluations of the variational gradient $\nabla_{\mathrm{V}} \bar{\mathcal{L}}$.
    Right:  \ac{mmd}, quantifying dissimilarity among the corresponding marginals, also as a function of the number of evaluations of the variational gradient $\nabla_{\mathrm{V}} \bar{\mathcal{L}}$.
    (The left panel was produced based on the final particle sets, corresponding to the maximum number of evaluations of the variational gradient.)
    }
    \label{fig: understanding}
\end{figure}

\subsubsection{Mean Field Langevin Dynamics}
For this experiment we initialised $n=300$ samples from $Q_1^{\text{init}} \times Q_2^{\text{init}}$ where $Q_1^{\text{init}} = \mathcal{N}(0,I_d)$ and $Q_2^{\text{init}} = \mathcal{N}(1,0.25 I_d)$. This choice is made in order to observe the convergence of $Q_1$ and $Q_2$ towards the same distribution. We run \ac{mfld} for $T=300$ iterations and step size $\varepsilon = 10^{-2}$.

\subsubsection{Variational Gradient Descent}
The kernel for \ac{vgd} was taken to be $k(x,x') =  (1 + \|x-x'\|^2)^{-1/2}$.
The initialisation, the number of particles and the number of iterations are the same as for \ac{mfld}. The step size is $\varepsilon = 1$.

\subsubsection{Extensible Sampling}

The initial particle $x_1$ was sampled from $Q_1^{\text{init}} \times Q_2^{\text{init}}$ and the subsequent particles were obtained using numerical optimisation of \ac{kgd} by stochastic search over $300$ samples from $\mathcal{N}(0,9 I_d)$. 
The kernel $k$ that defines the \ac{kgd} objective was $k(x,x') = (1 + \|x-x'\|^2)^{-1/2}$.

\end{document}